\documentclass{article}
\usepackage[a4paper, margin=1.25in]{geometry}

\usepackage{hyperref}

\usepackage{array}
\usepackage[caption=false,font=normalsize,labelfont=sf,textfont=sf]{subfig}
\usepackage{textcomp}
\usepackage{stfloats}
\usepackage{textcomp}
\usepackage{verbatim}
\usepackage{tabularx}
\usepackage{booktabs}       
\usepackage{nicefrac}       
\usepackage{microtype}      
\usepackage{bm}				
\usepackage{cite}			
\usepackage{graphicx}		
\usepackage{mathtools}		
\usepackage{algorithm}
\usepackage{nicefrac}
\usepackage{comment}
\usepackage{enumitem}		
\usepackage{multirow}       
\usepackage{tabularx}	    
\usepackage{hhline}
\usepackage{arydshln}		
\usepackage{xcolor}
\usepackage{float}
\usepackage{calc}      
\usepackage{mathtools,nccmath}
\usepackage{array,booktabs}
\usepackage{amsmath}
\newcommand{\triangleq}{\mathrel{\overset{\triangle}{=}}}
\usepackage{tikz}
\usepackage{pgfplots}
\pgfplotsset{compat=1.18}
\usetikzlibrary{spy,calc}
\definecolor{lightgreen}{rgb}{.9,1,.9}
\definecolor{red}{rgb}{1,0,0}

\newcolumntype{L}[1]{>{\raggedright\arraybackslash}p{#1}}
\newcolumntype{C}[1]{>{\centering\arraybackslash}p{#1}}
\newcolumntype{R}[1]{>{\raggedleft\arraybackslash}p{#1}}
\usepackage{lipsum}
\usepackage{amsfonts}
\usepackage{graphicx}
\usepackage{epstopdf}
\usepackage{algorithmic}
\ifpdf
  \DeclareGraphicsExtensions{.eps,.pdf,.png,.jpg}
\else
  \DeclareGraphicsExtensions{.eps}
\fi

\usepackage[amsthm]{ntheorem}

\newtheorem{theorem}{Theorem}

\newtheorem{remark}{Remark}

\newcommand{\mynewtheorem}[2]{
  \theoremstyle{plain}
  \theoremheaderfont{\normalfont\sffamily}
  \theorembodyfont{\normalfont\itshape}
  \theoremseparator{.}
  \theoremsymbol{}
  \newtheorem{#1}{#2}
}

\newcommand{\myrenewremark}[2]{
  \theoremstyle{plain}
  \theoremheaderfont{\normalfont\itshape}
  \theorembodyfont{\normalfont}
  \theoremseparator{.}
  \theoremsymbol{}
  \renewtheorem{#1}{#2}
}

\mynewtheorem{assumption}{Assumption}
\myrenewremark{remark}{Remark}
\mynewtheorem{proposition}{Proposition}
\mynewtheorem{corollary}{Corollary}

\usepackage{cleveref}

\crefname{remark}{remark}{remarks}
\Crefname{remark}{Remark}{Remarks}
\crefname{assumption}{assumption}{assumptions}
\Crefname{assumption}{Assumption}{Assumptions}

\crefname{appendix}{appendix}{appendices}
\Crefname{appendix}{Appendix}{Appendices}

\definecolor{lightgreen}{rgb}{.9,1,.9}
\definecolor{red}{rgb}{1,0,0}

\def\lim{\mathop{\mathsf{lim}}} 
\def\min{\mathop{\mathsf{min}}}

\def\log{\mathrm{log}}

\newcommand{\by}{\mathbf{y}}

\newcommand{\bx}{\mathbf{x}}

\newcommand{\bz}{\mathbf{z}}

\newcommand{\bh}{\mathbf{h}}

\newcommand{\bu}{\mathbf{u}}

\newcommand{\bc}{\mathbf{c}}

\newcommand{\bH}{\mathbf{H}}

\newcommand{\bI}{\mathbf{I}}

\def\R{\mathbb{R}}

\usepackage{tikz}
\usepackage{pgfplots}
\usetikzlibrary{spy,calc}

\newif\ifblackandwhitecycle
\gdef\patternnumber{0}

\pgfkeys{/tikz/.cd,
    zoombox paths/.style={
        draw=green,
        very thick
    },
    black and white/.is choice,
    black and white/.default=static,
    black and white/static/.style={ 
        draw=white,   
        zoombox paths/.append style={
            draw=white,
            postaction={
                draw=black,
                loosely dashed
            }
        }
    },
    black and white/static/.code={
        \gdef\patternnumber{1}
    },
    black and white/cycle/.code={
        \blackandwhitecycletrue
        \gdef\patternnumber{1}
    },
    black and white pattern/.is choice,
    black and white pattern/0/.style={},
    black and white pattern/1/.style={    
        draw=white,
        postaction={
            draw=black,
            dash pattern=on 2pt off 2pt
        }
    },
    black and white pattern/2/.style={    
        draw=white,
        postaction={
            draw=black,
            dash pattern=on 4pt off 4pt
        }
    },
    black and white pattern/3/.style={    
        draw=white,
        postaction={
            draw=black,
            dash pattern=on 4pt off 4pt on 1pt off 4pt
        }
    },
    black and white pattern/4/.style={    
        draw=white,
        postaction={
            draw=black,
            dash pattern=on 4pt off 2pt on 2pt off 2pt on 2pt off 2pt
        }
    },
    zoomboxarray inner gap/.initial=5pt,
    zoomboxarray columns/.initial=1,
    zoomboxarray rows/.initial=1,
    figurename/.initial={zoombox},
    zoomboxarray/.style={
        execute at begin picture={
            \begin{scope}[
                spy using outlines={%
                    zoombox paths,
                    width=0.4*\imagewidth,  
                    height=0.4*\imageheight, 
                    magnification=3,        
                    every spy on node/.style={
                        zoombox paths
                    },
                    every spy in node/.style={
                        zoombox paths
                    }
                }
            ]
        },
        execute at end picture={
            \end{scope}
            \gdef\patternnumber{0}
        },
        spymargin/.initial=0.1em,
        zoomboxes xshift/.initial=0,
        zoomboxes yshift/.initial=0,
        caption margin/.initial=4ex,
    },
    adjust caption spacing/.code={},
    image container/.style={
        inner sep=0pt,
        at=(image.north),
        anchor=north,
        adjust caption spacing
    },
    zoomboxes container/.style={
        inner sep=0pt,
        at=(image.south),
        anchor=north,
        name=zoomboxes container,
        adjust caption spacing
    },
    calculate dimensions/.code={
        \pgfpointdiff{\pgfpointanchor{image}{south west} }{\pgfpointanchor{image}{north east} }
        \pgfgetlastxy{\imagewidth}{\imageheight}
        \global\let\imagewidth=\imagewidth
        \global\let\imageheight=\imageheight
        \gdef\columncount{1}
        \gdef\rowcount{1}
        
    },
    image node/.style={
        inner sep=0pt,
        name=image,
        anchor=south west,
        append after command={
            [calculate dimensions]
        }
    },
    color code/.style={
        zoombox paths/.append style={draw=#1}
    },
    connect zoomboxes/.style={
        spy connection path={\draw[draw=none,zoombox paths] (tikzspyonnode) -- (tikzspyinnode);}
    },
    help grid code/.code={
        \begin{scope}[
            x={(image.south east)},
            y={(image.north west)},
            font=\footnotesize,
            help lines,
            overlay
        ]
        \foreach \x in {0,1,...,9} { 
            \draw(\x/10,0) -- (\x/10,1);
            \node [anchor=north] at (\x/10,0) {0.\x};
        }
        \foreach \y in {0,1,...,9} {
            \draw(0,\y/10) -- (1,\y/10);                        
            \node [anchor=east] at (0,\y/10) {0.\y};
        }
        \end{scope}    
    },
    help grid/.style={
        append after command={
            [help grid code]
        }
    },
}

\newcommand\zoombox[2][]{
    \begin{scope}[zoombox paths]
        \pgfmathsetmacro\xpos{
            0 
        }
        \pgfmathsetmacro\ypos{
            0 
        }
        \edef\dospy{\noexpand\spy [
            #1,
            zoombox paths/.append style={
                black and white pattern=\patternnumber
            },
            every spy on node/.append style={#1},
            x=\imagewidth,
            y=\imageheight
        ] on (#2) in node [anchor=south west] at ($(image.south west)+(\xpos pt,\ypos pt)+(1 pt,1 pt)$);}
        \dospy
        \pgfmathtruncatemacro\pgfmathresult{ifthenelse(\columncount==\pgfkeysvalueof{/tikz/zoomboxarray columns},\rowcount+1,\rowcount)}
        \global\let\rowcount=\pgfmathresult
        \pgfmathtruncatemacro\pgfmathresult{ifthenelse(\columncount==\pgfkeysvalueof{/tikz/zoomboxarray columns},1,\columncount+1)}
        \global\let\columncount=\pgfmathresult
        \ifblackandwhitecycle
            \pgfmathtruncatemacro{\newpatternnumber}{\patternnumber+1}
            \global\edef\patternnumber{\newpatternnumber}
        \fi
    \end{scope}
}

\usepackage{amsopn}

\makeatletter
\newcommand*{\addFileDependency}[1]{
  \typeout{(#1)}
  \@addtofilelist{#1}
  \IfFileExists{#1}{}{\typeout{No file #1.}}
}
\makeatother

\def\mytitle{Bregman geometry-aware split Gibbs sampling \\for Bayesian Poisson inverse problems}

\hypersetup{
  pdftitle={\mytitle},
  pdfauthor={Elhadji Cisse Faye, Mame Diarra Fall,  Nicolas Dobigeon, Eric Barat}
}

\title{\mytitle
\thanks{{This work was funded in part by the ANR AI.iO Project (ANR-20-THIA-0017), the BACKUP project (ANR-23-CE40-0018-01) and the Artificial Natural Intelligence Toulouse Institute (ANITI, ANR-23-IACL-0002).}}}

\author{
Elhadji~C.~Faye\thanks{Institut Denis Poisson, UMR CNRS University of Orléans, University of Tours, Orléans, France (\texttt{elhadji-cisse.faye@univ-orleans.fr}).}
\and Mame~Diarra~Fall\thanks{Univ Rouen Normandie, INSA Rouen Normandie, Universit\'e Le Havre Normandie,
Normandie Univ, LITIS UR 4108, F-76000 Rouen, France
  (\texttt{diarra.fall@univ-rouen.fr}).}  
\and Nicolas~Dobigeon\thanks{University of Toulouse, IRIT/INP-ENSEEIHT, CNRS, 2 rue Charles Camichel, BP 7122, 31071 Toulouse Cedex 7, France 
  (\texttt{Nicolas.Dobigeon@enseeiht.fr}).}
\and \'Eric~Barat\thanks{CEA, University of Paris-Saclay, France 
  (\texttt{eric.barat@cea.fr}).}  
}

\begin{document}

\maketitle

\begin{abstract}
  This paper proposes a novel Bayesian framework for solving Poisson inverse problems by devising a Monte Carlo sampling algorithm which accounts for the underlying non-Euclidean geometry. To address the challenges posed by the Poisson likelihood -- such as non-Lipschitz gradients and positivity constraints -- we derive a Bayesian model which leverages exact and asymptotically exact data augmentations. In particular, the augmented model incorporates two sets of splitting variables both derived through a Bregman divergence based on the Burg entropy. Interestingly the resulting augmented posterior distribution is characterized by conditional distributions which benefit from natural conjugacy properties and preserve the intrinsic geometry of the latent and splitting variables. This allows for efficient sampling via Gibbs steps, which can be performed explicitly for all conditionals, except the one incorporating the regularization potential. For this latter, we resort to a Hessian Riemannian Langevin Monte Carlo (HRLMC) algorithm which is well suited to handle priors with explicit or easily computable score functions. By operating on a mirror manifold, this Langevin step ensures that the sampling satisfies the positivity constraints and more accurately reflects the underlying problem structure. Performance results obtained on denoising, deblurring, and positron emission tomography (PET) experiments demonstrate that the method achieves competitive performance in terms of reconstruction quality compared to  optimization- and  sampling-based approaches.
\end{abstract}

\section{Introduction}\label{introduction}
Reconstructing images from measurements corrupted by Poisson noise is a fundamental problem in computational imaging, with critical applications in low-light photography, astronomy~\cite{hanisch1994restoration}, emission tomography~\cite{shepp2007maximum}, and fluorescence microscopy~\cite{agard1983three, sarder2006deconvolution}. In such settings, the data acquisition process involves photon-limited imaging, where the discrete and stochastic nature of photon arrivals leads to signal-dependent noise accurately modeled by a Poisson distribution~\cite{janesick2007photon}. Mathematically, the measurements $\by = (y_i)_{1 \le i \le m} \in \mathbb{N}^m$ are modeled as
\begin{equation}\label{eq:obs_intro}
    y_i \sim \mathcal{P}(\alpha \bh_i^\top \bx), \quad 1 \le i \le m,
\end{equation}
where $\bx \in \mathbb{R}_+^n$ denotes the unknown intensity image, $\alpha > 0$ controls the intensity level and reflects the severity of shot noise affecting the measurements, and each $\bh_i \in \mathbb{R}^n$ defines the forward model, collected row-wise in a matrix $\bH \in \mathbb{R}_+^{m \times n}$. From equation \eqref{eq:obs_intro}, the negative log-likelihood  of the model is given by
\begin{align}\label{eq:likelihood}
    - \log\, p(\by | \bx) &= f(\bx; \by)\nonumber \\
    				 &=  \sum_{i=1}^m  \left[ \alpha \bh_i^\top \bx - y_i \log(\alpha \bh_i^\top \bx) +  \log(y_i!)   \right].
\end{align}
Recovering the signal or image of interest $\bx$ from the measurements $\by$ is typically an ill-posed inverse problem due to noise, incomplete data, and ill-conditioning of $\bH$, necessitating suitable regularization strategies.

To address the ill-posedness of Poisson inverse problems, numerous optimization-based methods have been developed. These approaches typically formulate the reconstruction task as a variational problem, combining the data fidelity term $f(\cdot; \by)$ derived from the Poisson likelihood \eqref{eq:likelihood} with a regularization term $g(\cdot)$ that encodes prior knowledge about the image. It is important to note that the Poisson log-likelihood does not exhibit a Lipschitz-continuous gradient. This is a sufficient condition for classical algorithms to be applicable with guaranteed convergence. To mitigate these issues, researchers have explored various approaches. One classical method is the Poisson Image Deconvolution by Augmented Lagrangian (PIDAL) \cite{figueiredo2010restoration} algorithm, which employs the Alternating Direction Method of Multipliers (ADMM) \cite{afonso2010augmented} framework with total variation (TV) regularization \cite{rudin1992nonlinear}. Alternatively, Bauschke \emph{et al.} addressed this problem by proposing a proximal gradient descent (PGD) algorithm within the framework of Bregman divergence, called Bregman proximal gradient (BPG) \cite{bauschke2018regularizing}. The advantage of BPG lies in relaxing the smoothness requirement on the negative log-likelihood needed for PGD convergence and instead introduces the \emph{NoLips} condition. Whatever their algorithmic structures, conventional instances of these optimization approaches often employ explicit handcrafted priors, such as total variation (TV)~\cite{rudin1992nonlinear} or sparsity-inducing priors~\cite{lee2010hierarchical, park2008bayesian, dobigeon2009hierarchical}, which may not capture the complex structures present in natural images. To overcome these limitations, alternative approaches involving implicit data-driven priors have been introduced. Among these, Plug-and-Play (PnP)~\cite{venkatakrishnan2013plug} and Regularization-by-Denoising (RED)~\cite{romano2017little} methods have gained significant success. These approaches modify traditional proximal splitting schemes by replacing proximal or gradient steps with the application of powerful denoisers, such as BM3D \cite{dabov2007image}, demonstrating improved performance in Poisson reconstruction tasks \cite{marais2017proximal}. Such strategies have gained in popularity after the recent advances in the field of deep learning, leveraging the expressive power of neural network-based denoisers \cite{milanfar2024denoising}. In this direction, Hurault \emph{et al.}  propose a PnP algorithm specifically tailored for Poisson inverse problems, which ensures convergence by embedding a denoising operator within a Bregman geometry \cite{hurault2023convergent}. The method addresses the incompatibility between standard PnP schemes, which generally assume Gaussian noise and rely on Euclidean proximity operators, and the structure of Poisson noise, whose negative log-likelihood is neither Lipschitzian nor smooth. While these optimization-based methods have shown success in various scenarios, they primarily provide point estimates of the reconstructed image and do not inherently quantify the uncertainty associated with the reconstruction.
To overcome this limitation, Bayesian methods offer a probabilistic alternative, where the solution $\bx$ is treated as a random variable equipped with a prior distribution 
\begin{equation} \label{eq:prior}
    p(\bx) \propto \exp\{ - \beta g(\bx) \}
\end{equation}
where $\beta$ is an hyperparameter. Inference is conducted from the posterior distribution
\begin{align}\label{eq:posterior}
   \pi(\bx) \triangleq p(\bx | \by) &\propto p(\by | \bx) \, p(\bx) \nonumber \\
        &\propto \exp\{ - f(\bx; \by) - \beta g(\bx) \}.
\end{align}
As with the aforementioned optimization-based methods, Bayesian inference for Poisson inverse problems faces challenges due to the non-smoothness, non-convexity, and constraints such as non-negativity associated with the data fitting term $f(\cdot;\by)$. To address these challenges, the authors of \cite{Vono2019icassp} propose a split Gibbs sampler (SGS) that samples from an asymptotically exact approximation of the target distribution \cite{vono2020asymptotically}. By introducing auxiliary variables through a splitting strategy, SGS decomposes the original sampling problem into simpler subproblems that can be more easily managed, even in high-dimensional settings. More recently, Melidonis \emph{et al.} propose a reflected and regularized Langevin stochastic differential equation (SDE) that is well-posed and exponentially ergodic under mild conditions \cite{melidonis2023efficient}. This framework leads to the development of four reflected proximal Langevin Markov chain Monte Carlo (MCMC) algorithms, demonstrating effectiveness in image deblurring, denoising, and inpainting tasks under various noise models, including Poisson noise. Building upon this, {Klatzer \emph{et al.}} develop a novel PnP Langevin sampling methodology tailored for low-photon Poisson imaging problems \cite{klatzer2025efficient}. This approach incorporates two strategies: \emph{i)} a PnP Langevin method with reflections and a Poisson likelihood approximation, and \emph{ii)} a mirror sampling algorithm utilizing Riemannian geometry to handle constraints and the irregularity of the likelihood without approximations.  {Each of these sampling methods have been designed for handling distinct  priors. They can embed conventional model-based regularizations such as TV \cite{Vono2019icassp,melidonis2023efficient}, providing stochastic alternatives of PIDAL. They can also leverage data-driven priors, in particular those implicitly defined through deterministic denoisers following a PnP approach as in  \cite{klatzer2025efficient}. Another emerging line of work consists in modeling the prior distribution through a deep generative model, e.g., a score-based or diffusion model \cite{song2020score,ho2020denoising}. Given a pre-trained model, sampling from the posterior distribution \eqref{eq:posterior} then requires to adapt the diffusion-based solvers initially designed for prior sampling, e.g., through an approximation of the likelihood \cite{chung2022diffusion,singh2023score,webber2025likelihood}. One PnP-flavored alternative consists in leveraging a  splitting-based algorithmic scheme whose prior step is interpreted as a stochastic denoising, then performed by a score-based denoising diffusion process \cite{melidonis2025score}.}

This work introduces a novel Bayesian model for addressing Poisson inverse problems that can be instantiated for a wide range of regularization potentials $g(\cdot)$, in particular those associated with implicit, data-driven priors. This model builds on several key ingredients. First, it leverages an exact data augmentation strategy, a modeling trick that has proven effective in previous works on Poisson inversion while preserving the physical interpretability of the model. {The resulting augmented likelihood is shown to be separable with respect to the $n$ components $x_j$ of the image $\bx$ ($j\in \{1,\ldots,n\}$), bypassing the $m$ coupling initially induced by the forward operators $\bh_i$ in \eqref{eq:likelihood} ($i\in \left\{1,\ldots,m\right\}$)}. Second, it follows a Bregman geometry-aware double augmentation scheme that not only decouples the inference into simpler subproblems but also maintains favorable conjugacy properties for subsequent exact sampling. {Interestingly, this double level of augmentation is shown to relate to an equivalent single-level augmentation in an appropriate geometry. Up to authors' knowledge, this is the first time that a non-Euclidean variable splitting and a double-level augmentation are operated for conducting Bayesian inference. Finally, a split Gibbs sampler (SGS) is implemented to handle the resulting augmented posterior efficiently. Thanks to the conditional structure underlying the derived augmented Bayesian model, three of the four steps of this SGS boil down to sampling according to standard conditional distributions, which can be performed in an exact manner straightforwardly.} The remaining step consists in sampling according to the conditional distribution which embeds the regularization potential $g(\cdot)$ whose specification remains at the discretion of the end-user. {This task can be interpreted as a stochastic denoising when the observations are corrupted by an inverse-gamma noise, as also encountered when solving Poisson inverse problems with Bregman plug-and-play algorithms \cite{hurault2023convergent}.} For the sake of generality, we therefore propose to resort to a versatile yet powerful sampling step that can handle, in practice, a broad class of prior distributions -- provided their score function $\nabla \log\, g(\mathbf{\cdot})$ is either explicit or readily computable. {More precisely, this sampling step is based on the Hessian Riemannian Langevin Monte Carlo (HRLMC) algorithm \cite{zhang2020wasserstein}, the mirror Langevin algorithm already considered in \cite{klatzer2025efficient} which explicitly accounts for the underlying geometry. Once combined, these components form a coherent Bayesian framework for geometry-aware Monte Carlo sampling in Poisson imaging. The proposed so-called HRLMC-within-SGS  (HRLwSGS) algorithm allows for sampling from posterior distributions granted with data-driven prior models such as those based on denoising operators. One key distinction between the proposed HRLwSGS and the PnP-MLA proposed in \cite{klatzer2025efficient} lies in  the construction of the sampling scheme. Rather than applying MLA directly to the posterior distribution, as in their approach, we beforehand introduce a double-level AXDA scheme that decouples the data-fitting and the regularization term. Then the MLA step operates solely on the conditional distribution driven by the regularization term. Such a separation is not merely technical: in related Gaussian settings, similar decoupling strategies have been shown to yield significant practical advantages, notably by making the choice of hyperparameters and the convergence behavior insensitive to the potential ill-conditioning induced by the data-fitting term \cite{vono2019split}.}

The remainder of this paper is structured as follows. \Cref{background} reviews the technical preliminaries required to the derivation of the proposed Bayesian framework subsequently exposed in \Cref{sec:bayesian_model}. The Monte Carlo algorithm proposed to sample from the target posterior distribution is described in  \Cref{sec:proposed_algorithm}. Numerical results illustrating the efficiency of the method when solving various Poisson inversion tasks are reported in  \Cref{sec:experiments}. Finally, \Cref{sec:conclusion} summarizes the main findings and concludes the paper.

\section{Background}\label{background}
{This section reviews the three key concepts underlying the methodology derived in the sequel of this paper: Bregman divergences, the asymptotically exact data augmentation (AXDA) framework and the split Gibbs sampler (SGS). These elements form the foundation of the proposed sampling algorithm  to solve Poisson inverse problems.}

\subsection{Bregman divergences}\label{subsec:bregman_div}
Given a strictly convex and differentiable function $h: \mathbb{R}^n \to \mathbb{R}$, the Bregman divergence  between two points $\bx$ and $\bz$ is defined as \cite{bregman1967relaxation}
\begin{equation}\label{eq:bregman_divergence}
    d_h(\bx, \bz) = h(\bx) - h(\bz) - \langle \nabla  h(\bz), \bx - \bz \rangle.
\end{equation}
In the context of inverse problems, Bregman divergences were introduced in \cite{eggermont1993maximum, burger2004convergence} and are often used to encode constraints or non-Euclidean geometries. Unlike conventional metrics, Bregman divergences are generally asymmetric and do not necessarily satisfy the triangle inequality. Nevertheless, they capture important geometric properties and enable projections that are adapted to the underlying structure of the data. Several choices of the convex generator $h(\cdot)$ yield Bregman divergences adapted to different problem geometries. For instance, when $h(\bx) = \frac{1}{2} \|\bx\|^2$, the resulting Bregman divergence reduces to the squared Euclidean distance
\begin{equation}\label{eq:squared_Euclidean_distance}
    d_{\mathrm{E}}(\bx, \bz) = \frac{1}{2} \| \bx - \bz \|^2,
\end{equation}
which corresponds to the standard geometry {common to many} numerous problems underlying Gaussian noise models. Alternatively, choosing the unnormalized Shannon entropy $h(\bx) = \sum_{i=1}^n x_i \log x_i - x_i $ yields the generalized Kullback-Leibler (KL) divergence
\begin{equation}
    d_{\mathrm{KL}}(\bx, \bz) = \sum_{i=1}^n \left( x_i \log \frac{x_i}{z_i} - x_i + z_i \right),
\end{equation}
which is particularly suitable for modeling non-negative variables such as intensities and for tackling inverse problems under Poisson noise. 
Another relevant choice is the Burg entropy \cite{burg1975maximum} defined on $\mathbb{R}_{++}^n$ by 
\begin{equation}\label{eq:burg}
  h(\bx) = - \sum_{i=1}^n \log \, x_i.  
\end{equation}
The associated Bregman divergence, also known as the Itakura-Saito divergence, is given by
\begin{equation}\label{eq:itakura_saito}
    d_{\mathrm{IS}}(\bx, \bz) = \sum_{i=1}^n \left( \frac{x_i}{z_i} - \log \frac{x_i}{z_i} - 1 \right),
\end{equation}
which is known to be particularly appropriate for modeling strictly positive variables and arises naturally in information geometry and Poisson-like models. This divergence respects the multiplicative structure of the positive orthant and enforces positivity implicitly and has been shown to be particularly well suited for modeling multiplicative and Poisson noise processes \cite{banerjee2005clustering, fevotte2009nonnegative, cavalcanti2019factor}.

\subsection{Asymptotically exact data augmentation (AXDA)}\label{subsec:AXDA}
{The AXDA framework~\cite{vono2020asymptotically} aims at reformulating a complex or intractable distribution into an augmented (or \emph{split}) distribution with a more tractable structure.
For a target distribution $\pi(\bx)$, it consists in introducing an auxiliary (or \emph{splitting}) variable $\bz \in \mathbb{R}^n$ to define the joint (augmented) distribution
\begin{equation}\label{eq:split_dist_general_case}
    \pi_\rho(\bx, \bz) \triangleq \pi(\bz)\, \kappa_\rho(\bz, \bx)
\end{equation}
where the coupling kernel $\kappa_\rho(\cdot,\cdot)$ is designed such that $\pi_{\rho}$ defines a proper joint distribution and such that it satisfies the following assumption.
\begin{assumption}\label{assump:kernel}
    The kernel $\kappa_\rho(\cdot, \bx)$ weakly converges towards the Dirac measure at $\bx$ as $\rho \rightarrow 0$.
\end{assumption}
Under \Cref{assump:kernel}, the augmentation scheme \eqref{eq:split_dist_general_case} defines a so-called \emph{valid} AXDA model in the sense that the marginal posterior
\begin{equation}
    {\tilde{\pi}}_\rho(\bx) = \int {\pi}_\rho(\bx,\bz) d\bz
\end{equation}
satisfies the following property. 
\begin{theorem}[see \cite{vono2020asymptotically}]\label{prop:AXDA_TV}
    If \Cref{assump:kernel} is satisfied, the following result holds
\begin{equation}\label{eq:convTV}
    \left\| \pi(\bx) - {\tilde{\pi}}_\rho(\bx) \right\|_{\mathrm{TV}} \xrightarrow[\rho \to 0]{} 0.
\end{equation}
\end{theorem}
This result guarantees that the original distribution $\pi$ is asymptotically recovered from the marginal ${\tilde{\pi}}_\rho$ in the limit of vanishing coupling parameter $\rho$. Vono \emph{et al.} \cite{vono2020asymptotically} discuss various design choices for the kernel function $\kappa_\rho(\cdot,\cdot)$, which plays a central role in the formulation of this augmented posterior distribution and the subsequent sampling scheme (see \cref{subsec:SGS} below). Among functions satisfying \Cref{assump:kernel}, a general class of admissible geometry-aware coupling (or splitting) kernels can be defined as} 
\begin{equation}
\kappa_\rho(\bz, \bx) = \exp \left\{ -\frac{1}{\rho} d_h(\bz, \bx) + \varphi(\bz) \right\},
\label{eq:general_kernel}
\end{equation}
where $\rho > 0$ is a coupling parameter, $d_h$ denotes the Bregman divergence \eqref{eq:bregman_divergence} associated with the strictly convex and $\mathcal{C}^2$ function $h \colon \mathbb{R}^n \to \mathbb{R}$ defining the Bregman geometry and $\varphi : \mathbb{R}^n \to \mathbb{R}$ is a normalization potential. A natural and widely used choice is the quadratic generator $h(\bx) = \frac{1}{2} \| \bx \|^2$ which leads to the standard Euclidean distance defined in \eqref{eq:squared_Euclidean_distance} and yields the following Gaussian kernel
\begin{equation}
\kappa_\rho(\bz, \bx)  \propto_{\bz} \exp\left\{ - \frac{1}{2\rho} \| \bz - \bx \|^2 \right\},
\label{eq:gaussian_kernel}
\end{equation}
{where $\propto_{\bz}$ means up to a normalization constant that does not depend on $\bz$.} This choice has been adopted in many works from the literature dedicated to inverse or regression problems underlying Gaussian noises \cite{vono2018sparse, Coeurdoux2024pnp, wu2024principled, Sun2024provable,Faye2024} or even Poisson noise \cite{Vono2019icassp}.

{Beyond its tight connections to existing statistical models deeply discussed in \cite{vono2020asymptotically}, AXDA has been shown to have a particular value when targeting distributions $\pi(\cdot)$ defined by composite potentials, in particular such as  posterior distributions of the form \eqref{eq:posterior}. In a Bayesian framework, the augmentation usually involves solely one of the two terms composing the full potential, i.e., 
\begin{equation}\label{eq:split_dist}
    \pi_\rho(\bx, \bz) \propto \exp\left\{ - f(\bx;\by) - \beta g(\bz) \right\}\, \kappa_\rho(\bz, \bx).
\end{equation}
This formulation clearly exhibits an interesting feature of the augmented posterior: the data-fitting and regularization potentials have been decoupled, i.e., each conditional distribution associated to $\pi_\rho(\bx, \bz)$ involves only one of the two potentials, which yields a simple two-step Gibbs sampling procedure. This motivates the development of the split Gibbs sampler recalled in what follows.}

\subsection{Split Gibbs sampler (SGS)}\label{subsec:SGS}
{The SGS was introduced in \cite{vono2019split} as a particular instance of Gibbs sampling designed to sample from AXDA models associated with posterior distributions. When considering the split distribution $\pi_\rho(\bx,\bz)$ defined by \eqref{eq:split_dist}, it alternatively draw samples according to the two associated conditional distributions following the algorithmic scheme
\begin{align}
    \bx^{(t)} &\sim \pi_{\rho}(\bx | \bz^{(t-1)}) \propto \exp\left\{ - f(\bx;\by) \right\}\, \kappa_\rho(\bz^{(t-1)}, \bx) \label{eq:SGS_x}\\
    \bz^{(t)} &\sim \pi_{\rho}(\bz | \bx^{(t)})   \propto \exp\left\{ - \beta g(\bz) \right\}\, \kappa_\rho(\bz, \bx^{(t)}) \label{eq:SGS_z}.
\end{align}
Interestingly, since SGS defines a Gibbs sampler, the set of samples $\left\{\bx^{(t)}\right\}_{t\geq T_{\mathrm{bi}}}$ collected after a burn-in period $T_{\mathrm{bi}}$ are asymptotically distributed according to the marginal distribution $\tilde{\pi}_\rho(\bx)$. The pivotal role of SGS in the AXDA framework is clearly highlighted by \Cref{prop:AXDA_TV}: SGS enables Bayesian inference from the approximate model $\tilde{\pi}_\rho(\bx)$ in place of $\pi(\bx)$.}

{As highlighted earlier, this splitting allows the two terms $f(\cdot;\by)$ and $g(\cdot)$ defining the full potential to be dissociated and involved into two distinct conditional distributions \eqref{eq:SGS_x} and \eqref{eq:SGS_z}, respectively. Thus, SGS shares strong similarities with ADMM \cite{afonso2010augmented} and HQS  \cite{Geman1995nonlinear}. This divide-and-conquer strategy is known to enable simpler, scalable and more efficient sampling schemes, as long as the coupling kernel $\kappa_\rho$ has been carefully chosen to ensure that the sampling steps \eqref{eq:SGS_x} and \eqref{eq:SGS_z} remain manageable. For example, when the data fitting term $f(\cdot,\by)$ is associated to a Gaussian likelihood, choosing a conjugate Gaussian splitting kernel \eqref{eq:gaussian_kernel} leads to a Gaussian conditional distribution \eqref{eq:SGS_x} easy to sample from \cite{vono2022high}. However, except in trivial cases, sampling from \eqref{eq:SGS_z} is not straightforward and highly depends on the choice of the regularization potential $g(\cdot)$. When this potential is differentiable, one solution consists in performing one step of the unadjusted Langevin algorithm (ULA), leading to a Langevin-within-SGS (LwSGS) algorithm \cite{plassier2021dg,Faye2024}. When this potential is non-smooth but proximable (e.g., associated with sparsity-promoting priors), the authors in \cite{vono2018sparse,vono2019bayesian,eldaly2026bayesian} have considered its Moreau-Yosida approximation for embedding a proximal Monte Carlo (MYULA) step into SGS \cite{durmus2018efficient}. In both cases, one key advantage brought by SGS when decoupling the data fitting and regularization terms is that the step size required by the (proximal) gradient  steps is independent of the Lipschitz constant associated with the data fitting potential. Instead it now depends on the characteristics of the coupling kernel (e.g., the coupling parameter $\rho$). This can lead to faster algorithms when compared to a direct sampling of the original posterior \eqref{eq:posterior}  by ULA or MYULA, at the price of the bias induced by the AXDA approximation \cite{vono2018sparse,vono2019bayesian,vono2019split,pereyra2023split}.}

{Another popular use of SGS relies on an implicit definition of the regularization potential $g(\cdot)$. Following a plug-and-play approach, several recent approaches interpret the conditional distribution \eqref{eq:SGS_z} as the posterior distribution associated with the denoising of the variable $\bx^{(t)}$ corrupted by a noise whose statistical model is defined by the divergence underlying the splitting kernel. This interpretation has been exploited to derive several plug-and-play sampling algorithms with denoising-based priors \cite{Coeurdoux2024pnp, bouman2023generative,Sun2024provable,wu2024principled}.}

\section{Bayesian model}\label{sec:bayesian_model}
Starting from a posterior distribution of the form \eqref{eq:posterior} with the data-fitting term $f(\cdot;\by)$ defined in  \eqref{eq:likelihood}, this section builds the proposed Bayesian model step-by-step. As already emphasized in the introduction, this model can embed a large variety of prior distributions $p(\bx) \propto \exp \left\{-\beta g(\bx)\right\}$. Thus for the sake of generality and simplicity, this model is derived without specifying the regularization potential $g(\cdot)$, even if particular choices of this potential will be considered later. {We first describe an exact data augmentation strategy, which reformulates the Poisson likelihood in a way that enables the independent sampling of the image pixels while preserving physical interpretability. We then introduce a geometry-aware double-level AXDA scheme that allows for a favorable decoupling of the data-fitting term and the regularization potential. It is worth noting that this decoupling is a systematic asset of AXDA, regardless the choice of the coupling kernel. However, the design of an AXDA scheme should also be guided by two additional key  rationales. First, this augmentation should preserve the geometry of the problem, i.e., ensuring that the constraints imposed to the main variable $\bx$ remain unchanged. Second, as emphasized in \cref{subsec:SGS},  the choice of the  kernel should lead to conditional distributions with tractable forms, thereby facilitating the implementation of an efficient SGS. Depending on the augmentation scheme and the choice of the splitting kernel, these objectives may be achieved with varying degrees of effectiveness. Accordingly, the Bayesian model detailed in the following paragraphs is derived step-by-step with these two rationales in mind. To further motivate the proposed approach, the derivation  is also systematically accompanied by remarks discussing the possible alternative augmentation schemes and, in particular, highlighting their limitations (see \cref{sec:alternate_AXDA_scheme} for further details).}

\subsection{Exact data augmentation} \label{subsec:exact_augmentation}
{The data-fitting term $f(\cdot;\by)$ defined in  \eqref{eq:likelihood} raises several challenges when conducting statistical inference. First, even if it is concave and smooth, the Poisson log-likelihood is non-quadratic, defined on the positive orthant and not gradient Lipschiz. More importantly, for each count measurement $y_i$ ($i=1,\ldots,m$), the presence of the forward operator $\mathbf{h}_i$ induces a strong dependence across the unknown image intensities $x_1,\ldots,x_n$ to be recovered. To get rid of this dependency, a common approach adopted by optimization-based inversion techniques consists in introducing  auxiliary variables following a splitting technique similar to AXDA \cite{figueiredo2010restoration,setzer2010deblurring,carlavan2011sparse}. Although this augmentation is conceptually attractive, its naive application through AXDA in the probabilistic setting framing this work would remain limited. Indeed, it would fail to fully exploit the statistical structure and key properties of the observation model \eqref{eq:obs_intro}. Instead, to mitigate the challenges raised by the non-Lipschitz property and variable dependencies of the data-fitting term \eqref{eq:likelihood}, we propose to exploit an exact data augmentation initially adopted in the seminal and popular works introducing the maximum likelihood-expectation maximization (ML-EM) algorithms for emission and transmission tomography \cite{vardi1982maximum, lange1984reconstruction}. This exact augmentation has been explicitly used in more recent works related to tomography \cite{filipovic2018pet,goncharov2023nonparametric} and, more generally, to imaging problems under Poisson noise \cite{clerc2026reinterpreting}.
The key idea consists in first introducing some latent variables $\bc   \in \mathbb{N}^{m \times n}$ with $c_{ij} = [\bc]_{i,j}$ in the data space\footnote{Note that an alternative augmentation in the image space have been considered  in \cite{Fall2019_IJB, Fall2013_ICIP, Fall2011_ICIP, sitek2010reconstruction}.} such that
\begin{equation}\label{eq:n_count}
	c_{ij} |  x_j  \sim \mathcal{P}(\alpha h_{ij} x_j)
\end{equation}
where the latent variables $c_{ij}$, associated with data count, are mutually independent for all $(i,j) \in \{1,\ldots,m\} \times \{1,\ldots,n\} $. Then, the observed measurements $\by \in \mathbb{N}^m$ are deterministically related to the latent variable $\bc$ via the set of $m$ coherence conditions
\begin{equation}\label{eq:coherence_condition}
   \sum_{j=1}^n c_{ij} = y_i, \ \forall i \in \left\{ 1, \ldots, m\right\}.
\end{equation} 
These conditions can be alternatively expressed in term of the feasible set $\mathcal{C}(\by) $ defined by
\begin{equation}\label{eq:constraints}
    \mathcal{C}(\by) = \prod_{i=1}^m \mathcal{C}_i(y_i) \quad \text{with}  \quad \mathcal{C}_i(y_i)=
    \Big\{
        \bc_{i:} \triangleq (c_{i1}, \dots, c_{in}) \in \mathbb{N}^{n} :
        \sum_{j=1}^n c_{ij} = y_i
    \Big\}
\end{equation} 
Accordingly, the conditional distribution of $\by$ given $\bc$ writes
\begin{equation}
    -\log \, p(\by | \bc)
    = \iota_{\mathcal{C}(\by)}(\bc)
\end{equation}
where $\iota_{\mathcal{C}(\by)}(\cdot)$ denotes the indicator function of the set $\mathcal{C}(\by)$, i.e.,
\begin{equation*}
    \iota_{\mathcal{C}(\by)}(\bc) =
    \begin{cases}
        0, & \text{if } \bc \in \mathcal{C}(\by),\\[4pt]
        +\infty, & \text{else,}
    \end{cases}
\end{equation*}
with $\iota_{\mathcal{C}(\by)}(\bc) = \sum_{i=1}^n \iota_{\mathcal{C}_i(y_i)}(\bc_{i:})$.}

{Actually, the introduction of the latent variables \eqref{eq:n_count} and the coherence conditions \eqref{eq:coherence_condition} can be motivated by the infinite divisibility property of the Poisson distribution \cite[Chap. XVII]{feller1971introduction}. They jointly ensure the equivalence between the statistical observation models  defined either by the latent variables $\bc_{i:}$ or by the measured data $y_i$, $i\in \{1,\ldots,m\}$. Indeed, reciprocally, thanks to the stability of the Poisson distribution, prescribing the statistical model \eqref{eq:n_count} to the latent counts  $\bc_{i:}$ implies that  the measurement $y_i$ follows the observation model  defined by \eqref{eq:obs_intro}. This equivalence can be formulated as an exact data augmentation scheme \cite{van2001art}, in the sense that marginalizing out the latent variables from the augmented likelihood $p(\by , \bc | \bx)$ recovers the original likelihood  $p(\by | \bx)$, as stated in the following proposition.}
\begin{proposition}\label{prop:XDA}
	 The augmented (complete) likelihood defined as
	\begin{align}\label{eq:augmented_likelihood}
	 	p(\by , \bc | \bx) &= p(\by | \bc) \, p(\bc | \bx) 
	\end{align}
	satisfies
	\begin{equation}\label{eq:exact_marginalization}
		 p(\by | \bx) =  \sum_{\bc \in \mathbb{N}^{m \times n}} p(\by ,  \bc  |  \bx) 
	\end{equation}
    where $p(\by | \bx)$ is the Poisson likelihood defined in \eqref{eq:likelihood}.
\end{proposition}
\begin{proof}
    See \Cref{app:proof}.
\end{proof}

{The introduction of these latent variables transforms the Poisson likelihood \eqref{eq:likelihood} into an augmented likelihood \eqref{eq:augmented_likelihood} granted with interesting properties to perform statistical inference. In the context of ML estimation, this augmentation leads to closed-form updates of the E- and M-steps when deriving the popular ML-EM algorithm \cite{vardi1982maximum, lange1984reconstruction}. Conversely, in the context of Bayesian estimation framing the work conducted in this paper, assigning the prior distribution \eqref{eq:prior} leads to the joint posterior over $(\bx, \bc)$
\begin{align}\label{eq:posterior_with_exact_dt}
	\pi(\bx, \bc) &\triangleq p(\bx, \bc | \by)\\
    &\propto p(\by | \bc) \, p(\bc | \bx) \, p(\bx) \nonumber \\
	&\propto \exp \left\{ - f(\bx, \bc; \by) - \beta g(\bx) \right\}  \nonumber
\end{align}
where $p(\bc | \bx) = \prod_i\prod_j p(c_{ij}|x_j)$ and the extended data-fitting term is defined as
	\begin{equation}\label{eq:extended_data_fitting_term}
		f(\bx, \bc; \by) \triangleq \sum_{i=1}^m \sum_{j=1}^n \left[ \alpha h_{ij} x_j - c_{ij} \log(\alpha h_{ij} x_j) +  \log(c_{ij}!) \right] + \iota_{\mathcal{C}(\by)}(\bc).
	\end{equation}
\Cref{prop:XDA} ensures that the marginal distribution of $\bx$ under the augmented model \eqref{eq:posterior_with_exact_dt} coincides exactly with the true posterior \eqref{eq:posterior}. This guarantees that conducting Bayesian inference in the augmented space using $(\bx, \bc)$ targets the correct posterior distribution over $\bx$.}

{Moreover, the conditional structure of this augmented posterior will be shown to be well suited when deriving a Gibbs sampler that alternatively samples from the conditional  distributions associated with \eqref{eq:posterior_with_exact_dt}. Indeed the first conditional posterior distribution $p(\bc | \bx, \by)  \propto  \exp \left\{ - f(\bx, \bc; \by) \right\}$
of the latent variables is explicit and given by $m$ independent Multinomial distributions, enabling straightforward corresponding sampling steps (see \cref{sec:proposed_algorithm}). Moreover, the second conditional posterior distribution $p(\bx | \bc, \by)$ now depends on an extended data-fitting term \eqref{eq:extended_data_fitting_term} which is separable with respect to the unknown image intensities $x_1,\ldots,x_n$ to be recovered, i.e., 
\begin{equation*}
    f(\bx,\bc;\by) = \sum_{j=1}^n f_{:j}(x_j,\bc_{:j}) + \iota_{\mathcal{C}(\by)}(\bc)
\end{equation*}
with
	\begin{equation*}
		f_{:j}(x_j, \bc_{:j}) \triangleq \sum_{i=1}^m \left[ \alpha h_{ij} x_j - c_{ij} \log(\alpha h_{ij} x_j) +  \log(c_{ij}!) \right]
	\end{equation*}
and $\bc_{:j} \triangleq (c_{1j},\ldots,c_{mj}) \in \mathbb{R}^m$. However, at this stage of the derivation of the Bayesian model, this separability cannot be fully exploited, e.g., for sampling from this conditional distribution. Indeed this latter simultaneously involves not only the extended data-fitting term $f(\cdot,\bc;\by)$ but also the potential $g(\cdot)$ associated with the prior distribution. Thus, except in  trivial cases (e.g., an independent prior $p(\bx) = \prod_{j=1}^n p(x_j)$, not suitable for imaging applications), it cannot be factorized or handled with simple sampling steps. Thus, to separate the contributions of the data-fitting term and the regularization potential in the augmented posterior distribution \eqref{eq:posterior_with_exact_dt}, we adopt the AXDA framework recalled in \cref{subsec:AXDA}. This strategy is described in the next section.}

\subsection{First level of AXDA} \label{subsubsec:AXDA_1}
{Conducting AXDA of the posterior distribution  \eqref{eq:posterior_with_exact_dt} involves introducing an auxiliary variable $\bz_1 \in \mathbb{R}_{++}^n$, leading to the following doubly augmented posterior distribution denoted by $\pi_{\rho}$}
\begin{equation}\label{eq:augmented_posterior}
		\begin{split}
		\pi_{\rho}(\bx, \bc, \bz_{1}) & \propto \exp \left\{ - f(\bx, \bc; \by) - \beta g(\bz_1) \right\} \, \kappa_{\rho}(\bz_1, \bx).
	\end{split}
\end{equation}
A central modeling component in the augmented formulation~\eqref{eq:augmented_posterior} is the choice of the kernel function $\kappa_{\rho}$, which governs the coupling between the auxiliary variable $\bz_1$ and the variable of interest $\bx$. As already pointed out in \cref{subsec:AXDA,subsec:SGS}, a popular and widely-adopted choice consists in using a Gaussian kernel \eqref{eq:gaussian_kernel}. While this kernel leads to simple updates when coupled with a Gaussian likelihood, it remains ill-suited to inverse problems involving Poisson noise which are considered in this work. First, a Gaussian choice would disregard the positivity constraints inherent to intensity variables and would fail to reflect the geometry of Poisson counts. Moreover, in the absence of conjugacy with respect to the Poisson likelihood, it does not result in particularly simple SGS steps.  To overcome these limitations, this work proposes to define the kernel $\kappa_\rho$ using a Bregman divergence that is specifically adapted to the statistical and geometric structure of Poisson inverse problems. It is worth noting that, according to the authors' knowledge, this is the first time that a non-Gaussian kernel $\kappa_\rho$ is adopted within an AXDA framework to solve real-world estimation problems beyond the toy examples considered in \cite{vono2020asymptotically}. More precisely, a particularly relevant choice for the convex generator $h$ in the presence of positivity constraints is the Burg entropy, defined in \eqref{eq:burg} \cite{burg1975maximum}. The corresponding coupling kernel has the following structure
\begin{equation}\label{eq:IS_kernel}
\kappa_{\rho}(\bz_1, \bx) = \exp \left\{ -\frac{1}{\rho} d_\mathrm{IS}(\bz_1, \bx)  + \varphi(\bz_1)  \right\}
\end{equation}
where  $d_\mathrm{IS}(\cdot, \cdot)$ is the Itakura-Saito divergence defined in \eqref{eq:itakura_saito} and $\varphi(\bz_1)$ is the normalizing potential
\begin{equation}
	 \varphi(\bz_1)  = - \sum_{j=1}^{n} \log \, z_{1j}.
\end{equation}
This construction introduces a non-Euclidean geometry that preserves strict positivity and reflects the information structure of photon-limited measurements. Examining the Itakura-Saito divergence, as defined in \eqref{eq:itakura_saito}, highlights that the kernel naturally embeds the latent space within the positive orthant, promoting multiplicative consistency between the auxiliary variable $\bz_1$ and the variable of interest $\bx$.

The resulting doubly augmented posterior distribution enriched with the latent variables $\bc$ and the splitting variable $\bz_1$ writes
\begin{equation}\label{eq:augmented_posterior_1}
		\pi_{\rho}(\bx, \bc, \bz_{1}) \propto \exp\bigg\{ - f(\bx, \bc; \by) - \beta g(\bz_1)  -\frac{1}{\rho} d_\mathrm{IS}(\bz_1, \bx) + \varphi(\bz_1)  \bigg\}.
\end{equation}
{This posterior benefits from several inherent advantages of the AXDA framework. First, according to \Cref{prop:AXDA_TV}, the choice of the kernel $\kappa_\rho(\cdot, \bx)$ defined by \eqref{eq:IS_kernel} ensures that the original posterior distribution $\pi(\bx, \bc)$ is asymptotically recovered from  the marginal posterior $\tilde{\pi}_{\rho}(\bx, \bc) = \int {\pi}_{\rho}(\bx, \bc, \bz_1) d\bz_1$. Moreover, as initially motivated, the introduction of the auxiliary variable $\bz_1$ decouples the data-fitting term $f(\cdot,\bc;\by)$ from the regularization potential $g(\cdot)$. As expected, the conditional distribution $\pi_{\rho}(\bx|\bz_1,\bc)$ of the image variable $\bx$ given $\bz_1$ and $\bc$ is now factorizable with respect to the image intensities $x_1,\ldots,x_n$, resulting in $n$ independent sampling steps in the corresponding SGS. Despite these noticeable advantages, this augmented posterior distribution exhibits a notable drawback, as detailed in  \cref{subsec:simple_AXDA}. Since it does not yield a conjugate structure, the $n$ independent steps to sample from the conditional distributions $\pi_{\rho}(x_j|\bz_1,\bc)$ cannot be performed straightforwardly. Instead they necessitate elaborate procedures, which are often inefficient because computationally demanding in high-dimensional settings. Fortunately, restoring tractable conditional updates by Gibbs steps can be achieved by judiciously leveraging the AXDA framework once more, as detailed in what follows.}

\begin{remark}[Alternative augmented Bayesian model]\label{rem:alternative_model1}
{The first level of AXDA derived above relies on a coupling kernel \eqref{eq:IS_kernel} defined through an Itakura-Saito divergence. This kernel is non-symmetric, which may question the place of the  variables $\bz_1$ and $\bx$ in the divergence. Specifically, an alternative would consist in resorting to the coupling kernel
\begin{equation}
{\kappa}^{\prime}_\rho(\bz_1, \bx) = \exp \left\{ -\frac{1}{\rho} d_\mathrm{IS}(\bx, \bz_1)  + \varphi(\bz_1)  \right\}
\end{equation}
where the roles of the two variables have been exchanged in the divergence with respect to the initial kernel $\kappa_\rho$ defined by \eqref{eq:IS_kernel}. Interestingly, as detailed in \cref{subsec:simple_AXDA_symmetrized}, this choice leads to a more convenient conditional structure of the doubly augmented posterior. 
Sampling from the conditional distributions $\pi_{\rho}(x_j|\bz_1,\bc)$ would then be easier. However, the second conditional distribution $\pi_{\rho}(\bz_1|\bx,\bc)$ would exhibit a  potential whose smoothness properties are less favorable for efficient sampling.}
\end{remark}

\subsection{Second level of AXDA}\label{subsubsec:AXDA_2}
{The introduction of the first splitting variable $\bz_1 \in \mathbb{R}_{++}^n$ through  AXDA as presented in \cref{subsubsec:AXDA_1} was primarily driven by the goal of decoupling the data-fitting term $f(\cdot,\bc;\by)$ from the regularization potential $g(\cdot)$. We now introduce a second splitting variable $\bz_2 \in \mathbb{R}_{++}^n$ to recover an amenable conjugate structure of the posterior and, subsequently,  to streamline the implementation of SGS with good mixing properties under mild requirements. Interestingly, granting a single-level AXDA model with another splitting variable can be shown to define an overall valid double-level AXDA model. This property is not specific to the considered Poisson inversion problem. Thus it is stated below in a general case, similarly to \cref{prop:AXDA_TV}.
\begin{theorem}\label{prop:AXDA_TV2}
    Let us consider the doubly augmented distribution 
\begin{equation}\label{eq:doubly_augmented}
    {\pi}_{\rho}(\bx, \bz_{1:2}) =  \pi(\bz_1) \kappa_{\rho}(\bz_1,\bz_2)\kappa_{\rho}^{\prime}(\bz_2,\bx)
\end{equation}
 derived from the distribution of interest $\pi(\bx)$ and   induced by two coupling kernels   $\kappa_{\rho}$ and $\kappa^{\prime}_{\rho}$.
 Define the corresponding doubly marginalized  distribution 
\begin{equation}
    \tilde{\tilde{\pi}}_{\rho}(\bx) \triangleq \int \pi_{\rho} (\bx, \bz_{1:2}) d\bz_1 d\bz_2.
\end{equation} 
If the coupling kernels  $\kappa_{\rho}$ and $\kappa^{\prime}_{\rho}$ satisfy \Cref{assump:kernel},  the following result holds
\begin{equation}\label{eq:convTV2}
    \left\| \pi(\bx) -\tilde{\tilde{\pi}}_{\rho}(\bx) \right\|_{\mathrm{TV}} \xrightarrow[\rho \to 0]{} 0.
\end{equation}
Therefore, the double augmentation underlying the distribution \eqref{eq:doubly_augmented}  properly defines a valid overall double-level AXDA scheme. 
\end{theorem}
\begin{proof}
    See \Cref{app:proof_AXDA_TV2}.
\end{proof}
This result straightforwardly generalizes the already established \Cref{prop:AXDA_TV} to the case of two levels of augmentation. In the particular case considered in this work,  the distribution of interest $\pi(\bx,\bc)$ defined by \eqref{eq:posterior_with_exact_dt} has already followed a first level of AXDA, as detailed in \cref{subsubsec:AXDA_1}. We then propose to build on the augmented distribution \eqref{eq:augmented_posterior_1} as
\begin{equation}\label{eq:augmented_posterior_22}
	\begin{split}
		\pi_{\rho}&(\bx, \bc, \bz_{1:2}) \propto \exp\bigg\{ - f(\bx, \bc; \by) - \beta g(\bz_1)   - \frac{1}{\rho} d_\mathrm{IS}(\bz_1, \bz_2) + \varphi(\bz_1) \bigg\}  \, {\kappa}^{\prime}_\rho(\bz_2,\bx)
	\end{split}
\end{equation}
where the second coupling kernel ${\kappa}^{\prime}_\rho$ also derives from an Itakura-Saito divergence, yet evaluated with reflected arguments
\begin{equation}\label{eq:IS_kernelbis}
{\kappa}^{\prime}_\rho(\bz_2,\bx) = \exp \left\{ -\frac{1}{\rho} d_\mathrm{IS}(\bx, \bz_2)  + \varphi(\bz_2)  \right\}.
\end{equation}
}
This finally leads to the following triply augmented posterior distribution 
\begin{equation}\label{eq:augmented_posterior_2}
		\pi_{\rho} (\bx, \bc, \bz_{1:2}) \propto \exp\bigg\{ - f(\bx, \bc; \by) - \beta g(\bz_1)  
         - \frac{1}{\rho} d_\mathrm{IS}(\bz_1, \bz_2) + \varphi(\bz_1) - \frac{1}{\rho} d_\mathrm{IS}(\bx, \bz_2) + \varphi(\bz_2) \bigg\}.
\end{equation}
{Since both coupling kernels satisfy \Cref{assump:kernel}, this posterior distribution fulfills the conditions of \Cref{prop:AXDA_TV2}. This means that the posterior distribution resulting from the marginalization over the two splitting variables asymptotically coincides with the target posterior $\pi(\bx, \bc )$ when the coupling parameter $\rho$ vanishes. Moreover, the combination of the exact augmentation introduced in \cref{subsec:exact_augmentation} with this double level of AXDA gives the augmented posterior distribution \eqref{eq:augmented_posterior_2} a conditional structure that exhibits desirable properties. Indeed, as discussed in the next section, this structure greatly simplifies the implementation of SGS which samples iteratively from the four conditionals of $\pi_{\rho}(\bx, \bc, \bz_1, \bz_2)$: \emph{i)} three of these  conditionals belong to known and tractable families and \emph{ii)} the fourth conditional, which depends on the regularization potential $g(\cdot)$, exhibits a potential with interesting regularity properties.  This sampling procedure will be detailed  in \cref{sec:proposed_algorithm}. Finally, according to the authors’ knowledge, this is also the first time that two levels of successive AXDA is envisioned, theoretically grounded (i.e., \Cref{prop:AXDA_TV2})  and practically implemented. Yet, the sequel of this section draw connections between the proposed double-level AXDA scheme to existing related models. These connections exhibit interesting intrinsic properties of the augmented posterior \eqref{eq:augmented_posterior_2}.} 

\begin{remark}[Alternative augmented Bayesian models]
{As in \cref{subsubsec:AXDA_1} (see \Cref{rem:alternative_model1}), again, the non-symmetry of the non-Gaussian coupling kernels $\kappa_{\rho}$ and ${\kappa}^{\prime}_{\rho}$ questions the place of the splitting variables $\bz_1$ and $\bz_2$. The alternative choices which would have consisted in exchanging the roles played by the splitting variables or the variable of interest would have lead to more complex  sampling schemes, as further discussed in \cref{subsec:double_AXDA}.}
\end{remark}

\subsubsection{Double-level AXDA as a single-level AXDA}
{As stated above, the conditional structure offered by the augmented posterior \eqref{eq:augmented_posterior_2} ensures that three of its four corresponding conditional distributions are known and easy to sample. This also means that this posterior distribution can be easily marginalized over any of these three variables individually. Specifically, marginalizing over the second splitting variable $\bz_2$ leads to interesting insights into any double-level AXDA models. Since the scope of this property goes beyond the sole Poisson inversion problem addressed in this work, again, it is stated in a general case as a corollary of \cref{prop:AXDA_TV2}.
\begin{corollary}\label{prop:marginal_kernel}
    Let us consider the double-level AXDA distribution 
\begin{equation}\label{eq:marginalized_doubly_augmented}
    {\pi}_{\rho}(\bx, \bz_{1:2}) =  \pi(\bz_1) \kappa_{\rho}(\bz_1,\bz_2)\kappa_{\rho}^{\prime}(\bz_2,\bx)
\end{equation}
 derived from the distribution of interest $\pi(\bx)$ and induced by the coupling kernels   $\kappa_{\rho}$ and $\kappa^{\prime}_{\rho}$ which satisfy \Cref{assump:kernel}.
 Then the simply marginalized  distribution
\begin{equation}\label{eq:marginalized_double_aug_distribution}
    \bar{{\pi}}_{\rho}(\bx,\bz_1) \triangleq \int \pi_{\rho} (\bx, \bz_{1:2}) d\bz_2
\end{equation} 
can be written as $\bar{{\pi}}_{\rho}(\bx,\bz_1) =  \pi (\bz_{1}) \bar{{\kappa}}_{\rho}(\bz_1,\bx)$ where the kernel
\begin{equation}\label{eq:marginal_kernel}
    \bar{{\kappa}}_{\rho}(\bz_1,\bx) = \int \kappa_{\rho}(\bz_1,\bz_2)\kappa_{\rho}^{\prime}(\bz_2,\bx) d\bz_2.
\end{equation} 
also verifies \Cref{assump:kernel}. Thus the distribution \eqref{eq:marginalized_double_aug_distribution} results from a single-level AXDA induced by the coupling kernel \eqref{eq:marginal_kernel} and
\begin{equation}\label{eq:convTV_marg}
    \left\| \pi_\rho(\bx,\bz_1) - {\bar{\pi}}_\rho(\bx,\bz_1) \right\|_{\mathrm{TV}} \xrightarrow[\rho \to 0]{} 0
\end{equation}
with $\pi_\rho(\bx,\bz_1) \triangleq \pi(\bz_1) \kappa_{\rho}(\bz_1,\bx) $.
\end{corollary}
\Cref{prop:marginal_kernel} shows that a double-level AXDA induced by two coupling kernels $\kappa_{\rho}$ and $\kappa_{\rho}^{\prime}$  can be interpreted as a single-level AXDA defined by a  coupling kernel $\bar{\kappa}_{\rho}$.  Applying this corollary to the double-level AXDA distribution \eqref{eq:augmented_posterior_2} leads to the marginalized posterior
\begin{align}\label{eq:augmented_posterior_marginalz2}
		{\bar{\pi}}_{\rho} (\bx, \bc, \bz_{1}) & \triangleq \int {{\pi}}_{\rho} (\bx, \bc, \bz_{1},\bz_{2}) d\bz_2 \\
        &\propto \exp\bigg\{ - f(\bx, \bc; \by) - \beta g(\bz_1)           \bigg\} \bar{\kappa}_\rho(\bz_1,\bx) \nonumber
\end{align}
with
\begin{equation}
		 \bar{\kappa}_\rho(\bz_1,\bx) \propto_{\bz_1} \prod_{j=1}^n \frac{1}{z_{1j}} \left(\frac{\sqrt{z_{1j} x_j}}{{z_{1j}+ x_j}}\right)^{2/{\rho}}.
\end{equation}
This kernel can be identified as a product of independent scaled Beta prime distributions \cite{Perez2017}, i.e.,  $z_{1j}| x_j \sim \mathrm{sBeta}_2(\tfrac{1}{\rho},\tfrac{1}{\rho},x_j)$ with 
\begin{equation}
    \mathrm{E}[z_{1j}|x_j] = \frac{1}{1-\rho}x_j \ \text{and} \ \mathrm{var}[z_{1j}|x_j] = \frac{\rho(2-\rho)}{(1-\rho)^2(1-2\rho)} x_j^2.
\end{equation}
As expected, the kernel $\bar{\kappa}_\rho$ satisfies \Cref{assump:kernel}. Thus the marginal posterior \eqref{eq:augmented_posterior_marginalz2} defines a valid single-level AXDA granted with \Cref{prop:AXDA_TV}. Besides, this kernel can also be rewritten with a form similar to \eqref{eq:general_kernel}, i.e., 
\begin{equation}
\label{eq:BR_kernel}
		 \bar{\kappa}_\rho(\bz_1,\bx) = \exp \left\{ -\frac{2}{\rho} d_{\mathrm{BR}}(\bz_1,\bx)  + \psi(\bz_1)\right\}
\end{equation}
where $\psi(\bz_1)$ is a normalizing potential and the underlying divergence $d_{\mathrm{BR}}$ is defined by the log-ratio of the arithmetic and geometric means
\begin{align}
		  d_{\mathrm{BR}}(\bz_1,\bx)  &= \sum_{j=1}^n \log\left({\frac{z_{1j}+x_j}{2\sqrt{z_{1j}x_j}}}\right).
\end{align}
This divergence is a particular instance of the Burbea-Rao divergence \cite{nielsen2011burbea}
\begin{align}
d_{\mathrm{BR}}(\bz_1,\bx) &= \sum_{j=1}^n \left[\frac{h(z_{1j}) + h(x_j)}{2} - h\left({\frac{z_{1j}+x_j}{2}}\right) \right]
\end{align}
where $h(\cdot) = - \log(\cdot)$. It results from a symmetrization of the Bregman divergence associated with a convex generator chosen as the Burg entropy, i.e., the symmetrized counterpart of the Itakura-Saito divergence
\begin{align}
d_{\mathrm{BR}}(\bz_1,\bx) &= \frac{1}{2} \left[d_{\mathrm{IS}}\left(\bz_1,\frac{\bx+\bz_1}{2}\right) +d_{\mathrm{IS}}\left(\bx,\frac{\bx+\bz_1}{2}\right) \right].
\end{align}
In other words, the proposed double-level AXDA can be interpreted as a single-level AXDA, induced by a symmetric geometry-aware splitting kernel rather than a non-symmetric one.} 

\subsubsection{Double-level AXDA as a generalized split-and-augmented model}\label{subsubsec:SPA_AXDA}
{By inspecting the structure of the augmented distribution \eqref{eq:augmented_posterior_2}, it appears that the splitting variable $\bz_2$ acts as a probabilistic buffer which interacts with both $\bz_1$ and $\bx$. It is worth noting that a similar approach has been adopted in \cite{vono2019split} and used in \cite{thouvenin2022distributed} to overcome the limitations brought by the \emph{split} distribution \eqref{eq:augmented_posterior} obtained after a single-level AXDA scheme. Specifically, to improve the mixing properties of SGS, the authors of \cite{vono2019split} proposed to sample from the so-called \emph{split-and-augmented} distribution
\begin{equation}\label{eq:augmented_posterior_SPA}
		\pi_{\rho} (\bx, \bz_{1:2}) \propto \exp\bigg\{ - f(\bx, \by) - \beta g(\bz_1)  
         - \frac{1}{\rho} d(\bx, \bz_1-\bz_2) + \Phi(\bz_1,\bz_2) \bigg\}.
\end{equation}
It seems that the scheme of this second level of augmentation differs from the double-level AXDA conducted above. However, when the data-fitting term $f(\cdot, \by)$ derives from a Gaussian likelihood as considered in \cite{vono2019split}, the divergence $d(\cdot, \cdot)$  in \eqref{eq:augmented_posterior_SPA} has been chosen as the squared Euclidean distance \eqref{eq:squared_Euclidean_distance}. In this particular case, the split-and-augmented distribution could  have been also written equivalently in a form similar to \eqref{eq:augmented_posterior_2}, i.e., 
\begin{equation}\label{eq:augmented_posterior_SPA2}
		\pi_{\rho} (\bx, \bz_{1:2}) \propto \exp\bigg\{ - f(\bx, \by) - \beta g(\bz_1)  
         - \frac{1}{\rho} d_{\mathrm{E}}(\bz_1, \bz_2) - \frac{1}{\rho} d_{\mathrm{E}}(\bx, \bz_2) + \varphi(\bz_1) + \varphi(\bz_2) \bigg\}.
\end{equation}
Thus the double-level AXDA scheme adopted in this paper generalizes the split-and-augmented model, whose design initially followed the same motivation. Specifically, the second level of augmentation aims at introducing an auxiliary variable $\bz_2$ which serves as a stochastic mediator between $\bx$ and $\bz_1$. As empirically demonstrated  in \cite{vono2019split}, this auxiliary variable facilitates information exchange between the prior and the likelihood terms, improving the mixing properties of the resulting Gibbs sampler.}


\section{Hessian Riemannian Langevin Monte Carlo within split-Gibbs sampler (HRLwSGS)}
\label{sec:proposed_algorithm}
To perform posterior inference, we develop a Gibbs sampler targeting the joint distribution $\pi_{\rho}(\bx,\bc,\bz_{1:2})$ defined by \eqref{eq:augmented_posterior_2}, i.e. the posterior distribution of the variables $(\bx,\bc,\bz_1,\bz_2)$ conditioned on the observed counts $\by$. Thanks to the three augmentation schemes leading to the Bayesian model detailed in \cref{sec:bayesian_model}, the target distribution  exhibits a conditional structure that leads to three conditional distributions with closed-form classical sampling schemes, while the fourth requires a specialized geometry-aware sampling method. The description of the resulting SGS is decomposed into two parts, more precisely,
\begin{itemize}
    \item standard sampling steps associated with the conditional distributions $\pi_{\rho}(\bc|\bx,\bz_{1:2})$, $\pi_{\rho}(\bx|\bc,\bz_{1:2})$ and $\pi_{\rho}(\bz_2|\bx,\bc,\bz_{1})$,
    \item non-standard sampling step associated with the conditional distribution $\pi_{\rho}(\bz_1|\bx,\bc,\bz_{2})$.
\end{itemize}

\subsection{Standard sampling steps} \label{subsec:Standar_sampling_steps}
The latent counts $\bc = \{c_{ij}\}$ have been introduced in \cref{subsec:exact_augmentation} as latent variables relating the observed data $\by$ to the unknown image $\bx$ via the probabilistic model \eqref{eq:n_count} and the coherent conditions \eqref{eq:coherence_condition}. By exploiting the factorizable forms of this model and the feasible set $\mathcal{C}(\by)$ underlying these conditions, the conditional distribution of the latent counts writes
\begin{equation}
	\begin{split}
		p(\bc  | \bx , \by) &  \propto \exp \left\{ -f(\bx, \bc; \by) \right\} \\
		&\propto \prod_{i=1}^m \exp \left\{ \sum_{j=1}^n\left[c_{ij}\log(\alpha h_{ij}x_j) - \log(c_{ij})\right] + \iota_{\mathcal{C}_i(y_i)}(\bc_i) \right\}.
	\end{split}
\end{equation}
Conditioned on $\by$ and $\bx$, the latent variables $\bc$ follow a product of $m$ independent constrained Poisson distributions. By applying the well-known result that Poisson variables conditioned on their sum follow a multinomial distribution {(see, e.g., \cite[Chap. 5] {ross2023introduction} or \cite[Chap. 16]{gelman2013bayesian})}, it yields
\begin{equation} \label{eq:sample_n}
	\bc_{i:}  | \bx, y_i \sim \textrm{Multinomial} \left\{ y_i, \left\{ \frac{h_{ij} x_j}{\sum_{k=1}^n h_{ik} x_k} \right\}_{j=1}^n \right\}, \ \forall i \in \left\{ 1, \ldots, m\right\}.
\end{equation}
Each row $\bc_{i:} \in \mathbb{N}^{n} $ of the latent count matrix $\bc$ can thus be sampled independently.

Given the counts $\bc$ and the splitting variable $\bz_2$, the conditional distribution  of $\bx$ results from a combination of an augmented Poisson likelihood and a prior defined by the Bregman divergence associated with the Burg entropy
\begin{align}\label{eq:cond_x}
		p(\bx | \bc,\bz_2 )  &  \propto \exp \left\{ - f(\bx, \bc; \by) -\frac{1}{\rho} d_\mathrm{IS}(\bx, \bz_2)  \right\}    \\
		&\propto \prod_{j=1}^n \exp \left\{\sum_{i=1}^m \left[c_{ij}\log(\alpha h_{ij} x_j)- \alpha h_{ij}x_j\right] -\frac{1}{\rho} \left( \frac{x_j}{z_{2j}} - \log \, \frac{x_j}{z_{2j}} \right)  \right\}\nonumber
\end{align}
By denoting $\bc_{:,j}= (c_{1j},\ldots,c_{mj})$ ($j=1,\ldots,n$), this conditional distribution factorizes as a product of $n$ independent Gamma distributions 
\begin{equation} \label{eq:sample_x}
x_j  | \bc_{:,j}, z_{2j} \sim \textrm{Gamma} \left( \sum_{i=1}^m c_{ij} + \frac{1}{\rho} + 1,\; \alpha \sum_{i=1}^m h_{ij} + \frac{1}{\rho z_{2j}} \right), \ \forall j \in \left\{ 1, \ldots, n\right\},
\end{equation}
which can be sampled independently.

The conditional distribution of the second splitting variable $\bz_2$ results from the combination of the two coupling kernels and writes
\begin{equation*}
	\begin{split}
		 p(\bz_2 | \bx,\bz_1 ) &  \propto \exp\bigg\{  - \frac{1}{\rho} d_\mathrm{IS}(\bz_1, \bz_2) - \frac{1}{\rho} d_\mathrm{IS}(\bx, \bz_2) + \varphi(\bz_2) \bigg\}   \\ 
		 &\propto  \prod_{j=1}^n \exp \left\{-\frac{1}{\rho}  \left( \frac{x_{j}}{z_{2j}} - \log \, \frac{1}{z_{2j}} \right) -\frac{1}{\rho}  \left( \frac{z_{1j}}{z_{2j}} - \log \, \frac{1}{z_{2j}} \right) - \log\, z_{2j} \right\}
	\end{split}
\end{equation*}
This conditional distribution is also factorizable as a product of $n$ independent inverse Gamma distributions
\begin{equation} \label{eq:sample_z2}
z_{2j}  |x_j, z_{1j}  \sim \textrm{InvGamma} \left( \frac{2}{\rho},\; \frac{x_j + z_{1j}}{\rho} \right), \ \forall j \in \left\{ 1, \ldots, n\right\}.
\end{equation}
{\begin{remark}[Role of the splitting variable $\bz_2$]
    The conditional posterior distribution \eqref{eq:sample_z2}  confirms the interpretation of the role played by the second splitting variable $\bz_2$ already highlighted in 
    \cref{subsubsec:SPA_AXDA}. Indeed, from  \eqref{eq:sample_z2}, it clearly appears that the variable \(\bz_2\) acts as a stochastic mediator between $\bx$ and $\bz_1$. Specifically, one has $\mathrm{E}[z_{2j}  | z_{1j}, x_j] = \frac{z_{1j}+ x_j}{2-\rho}$ and $\mathrm{var}[z_{2j}  | z_{1j}, x_j] = \frac{ \rho (z_{1j}+ x_j)^2}{2(1-\rho)(2-\rho)^2}$. Notably, when $\rho \to 0$, the inverse-gamma distribution becomes increasingly concentrated, and its mean converges towards the arithmetic mean of the two variables. In this regime, the variable $\bz_2$ effectively interpolates between $\bx$ and $\bz_1$, behaving like a stochastic average. This behavior allows the information to propagate through the Bayesian model from the variable $\bx$ whose conditional distribution \eqref{eq:cond_x} is driven by the data-fitting term $f(\cdot,\bc;\by)$ towards the first splitting variable $\bz_1$ whose conditional distribution is driven by the regularization potential $g(\cdot)$ (see \cref{subsec:non_standard_sampling} below). In other words, each update of $\bz_2$ carries combined information brought by the data (through the variable $\bx$) and the regularization (through the variable $\bz_1$).
\end{remark}}

\subsection{Non-standard sampling} \label{subsec:non_standard_sampling} 
The last step of the proposed sampling scheme consists in sampling from the conditional distribution of $\bz_1$
\begin{align}\label{eq:eq:cond_z1}
		p(\bz_1 | \bz_2) &  \propto \exp\bigg\{ - \beta g(\bz_1) - \frac{1}{\rho} d_\mathrm{IS}(\bz_1, \bz_2) + \varphi(\bz_1). \bigg\} 
\end{align}
It involves the regularization term $g(\cdot)$ associated with the prior distribution defined in \eqref{eq:posterior}.  Its structure can be easily interpreted as the posterior distribution associated with a non-Gaussian denoising problem  with the regularization function $g(\cdot)$. {The discrepancy between the clean and observed variables is measured through the Bregman divergence $d_h(\bz_1, \bz_2)$, induced by the Burg entropy, which reflects the geometry of positive-valued data. The objective is then to reconstruct $\bz_1$ from a noisy observation $\bz_2$, which is assumed to be affected by a multiplicative noise, consistent with an inverse-gamma distribution underlying the Itakura-Saito divergence of the data-fitting term. In \cite{hurault2023convergent}, this denoising task has motivated the design of the so-called Bregman score denoiser (BSD) to be used within plug-and-play algorithms for solving Poisson inverse problems.} This conditional distribution can be rewritten as $p(\bz_1 | \bz_2) \propto \exp \left\{ - U(\bz_1) \right\}$ with
\begin{equation}\label{eq:potential}
U(\bz_1) = \beta g(\bz_1) + \left(1- \tfrac{1}{\rho}\right)\sum_{j=1}^n \log z_{1j} + \frac{1}{\rho} \sum_{j=1}^n  \frac{z_{1j}}{z_{2j}} .
\end{equation}
{For most regularization potentials $g(\cdot)$, this distribution does not belong to a standard family and cannot be sampled from directly. Taking inspiration from the LwSGS algorithm  \cite{plassier2021dg,Faye2024}, one generic sampling strategy would consist in performing one step of the Langevin Monte Carlo (LMC) algorithm. This strategy is sufficiently versatile to handle a wide class of regularizations, including implicit data-driven prior distributions, provided the gradient $\nabla g(\cdot)$ of the regularization potential is explicit or easily computable. However, the canonical instance of LMC would not account for the non-negativity constraints intrinsic to the considered estimation problem. Alternatively, to sample from \eqref{eq:eq:cond_z1}, one could implement reflected \cite{melidonis2023efficient} or projected \cite{laumont2022bayesian} discretization schemes of a reflected Langevin diffusion process, as considered in \cite{klatzer2025efficient}. Instead, we propose to use a constraint-compliant, geometry-aware extension of LMC, namely the Hessian Riemannian Langevin Monte Carlo (HRLMC) algorithm \cite{zhang2020wasserstein,li2022mirror}, which is particularly suited to constrained domains such as $\R^n_{++}$. This mirror Langevin algorithm is a discretization of Langevin dynamics on a Riemannian manifold, where the geometry is defined by the Hessian of a convex potential function $\phi(\cdot)$. It allows for efficient sampling from distributions supported on manifolds or constrained subsets of Euclidean space.} The HRLMC updating rule writes
\begin{equation} \label{eq:sample_z1}
\bz_1^{(t+1)} = \nabla \phi^\star \left( \nabla \phi(\bz_1^{(t)}) - \gamma \nabla U(\bz_1^{(t)}) + \sqrt{2\gamma \nabla^2 \phi(\bz_1^{(t)})} \boldsymbol{\varepsilon}^{(t)} \right),
\end{equation}
where $\boldsymbol{\varepsilon}^{(t)} \sim \mathcal{N}(0, \bI_n)$ and  $\phi^\star(\cdot)$ is the Legendre transform of $\phi(\cdot)$. In the considered Poisson inversion context, the mirror function $\phi(\cdot)$
is set as the Burg entropy  \eqref{eq:burg} and its gradient and Hessian are given\footnote{The inversion and exponentiation operations should be understood as component-wise.} by $\nabla \phi(\bz) = - \frac{1}{\bz}$ and  $\nabla^2 \phi(\bz) = \frac{1}{\bz^2} \bI_n$, respectively. This sampling strategy guarantees positivity and adapts to the manifold geometry. When the potential $U(\cdot)$ is given by \eqref{eq:potential}, its gradient writes 
\begin{equation}\label{eq:gradientU}
    \nabla U(\bz_1) = \beta \nabla g(\bz_1) + \frac{1}{\rho \bz_2} + \left(1 - \tfrac{1}{\rho}\right)\frac{1}{\bz_1}.
\end{equation}
{As for the LwSGS algorithm, the practical implementation of the HRLMC iterations within the proposed SGS requires only the gradient $\nabla g(\cdot)$ of the regularization potential to be explicit or easily computable -- a quantity also referred to as the prior score. This encompasses several recent and powerful data-driven priors, ranging from denoiser-based regularizations (see \cref{sec:denoisers}) to those derived from generative models. Moreover, the potential $U(\cdot)$ as defined in \eqref{eq:potential} exhibits favorable properties to get convergence guarantees of HRLMC. In particular, apart from the term $g(\cdot)$ whose definition is let to the end-user, it is smooth with respect to  $\bz_1$. This structure simplifies the convergence analysis of HRLMC, under certain conditions, including relative strong convexity and Lipschitz-smoothness of the potential function $U(\cdot)$. Specifically, the relative strong convexity constant remains positive under mild assumptions on the regularization potential $g(\cdot)$, and the commutator bound between $\nabla^2 \phi(\cdot)$ and $\nabla^2 U(\cdot)$ stays moderate as long as $\nabla^2 g(\cdot)$ is bounded. These conditions are further discussed in \cref{HRLMC_assumptions}, in particular when the regularization potential $g(\cdot)$ derives from the RED paradigm. It is worth noting that such favorable properties of the potential $U(\cdot)$ may be lost when considering alternative Bayesian models resulting from different single- and double-level AXDA schemes, as demonstrated in \cref{sec:alternate_AXDA_scheme}.}

{In this work, HRMLC is used as a sampling subroutine included into a 4-step SGS to handle one of the four conditional distributions associated with the posterior distribution \eqref{eq:augmented_posterior_2}. The strategy which consists in embedding one gradient-like sampling step into an SGS has been advocated in numerous works of the literature, including  the MYULA-within-SGS \cite{vono2018sparse,vono2019bayesian,eldaly2026bayesian}, the LwSGS algorithm \cite{Faye2024} and its distributed generalized counterpart \cite{plassier2021dg}, or the multi-GPU distributed MCMC algorithm \cite{bouton2026distributed} building on the proximal stochastic gradient Langevin-within-SPA \cite{thouvenin2022distributed}. Inspired by these algorithmic schemes, we propose to embed only one iteration\footnote{Experimental results not reproduced here show that performing several steps of HRLMC within one iteration of SGS yields only marginal performance improvements, while significantly increasing  the overall computational burden of the algorithm.} of the HRLMC updating rule \eqref{eq:sample_z1} into the proposed SGS. This results into the Hessian Riemmaninan Langevin-within-SGS (HRLwSGS) whose overall sampling procedure is summarized in \cref{alg:HRLwSGS}.}

\subsection{On the convergence of the proposed sampler}\label{subsec:convergence}
{
    The proposed HRLwSGS algorithm (\cref{alg:HRLwSGS}) is a particular instance of Gibbs sampling where the three first steps are performed by exact sampling while the last step is performed thanks to HRLMC. Conceptually, several inner steps of HRLMC could have been conducted within each iteration of the Gibbs sampler. Yet, following the approach already adopted in \cite{Faye2024,plassier2021dg}, we propose to restrict this number of HRLMC to one. Increasing this number would lead to a substantial computational overload without improving  the quality of the inversion significantly, according to experiments whose results are not reproduced in this manuscript.}
    
    {Besides, since the HRLMC-based updating rule \eqref{eq:sample_z1} results from the Euler-Maruyama discretization of the  mirror Langevin dynamics in continuous time, it produces samples affected by a bias, i.e., not exactly distributed according to the conditional distribution \eqref{eq:eq:cond_z1}. To mitigate this bias and to ensure that the proposed HRLwSGS exactly targets the posterior distribution \eqref{eq:augmented_posterior_2}, a strategy would consist in complementing this updating rule with a Metropolis-Hastings (MH) step, which would result in a well-grounded Metropolis-within-Gibbs algorithm \cite{robert1999monte}. However, performing this MH step at each iteration of the SGS would increase the overall computational burden significantly. Thus, as in \cite{Faye2024,plassier2021dg}, we propose to leave this MH correction out to avoid this additional computational cost, at the price of an approximation governed by the approximation step $\gamma$. }

{
    Finally, it is also worth noting that the proposed HRLwSGS produces a collection of samples $\left\{\bx^{(t)}, \bc^{(t)},\bz^{(t)}, \bu^{(t)}\right\}_{t=N_{\mathrm{bi}}+1}^{N_\textrm{MC}}$ distributed according to the augmented posterior distribution $\pi_{\rho}(\bx,\bc,\bz,\bu)$ defined by \eqref{eq:augmented_posterior_2}, up to the approximation acknowledged above. Interestingly, since HRLwSGS is a particular instance of Gibbs sampling, the collection $\left\{\bx^{(t)}\right\}_{t}$ is asymptotically distributed according to the marginal distribution
    \begin{equation}
        \tilde{\tilde{\pi}}_{\rho}(\bx) = \int \pi_{\rho}(\bx,\bc,\bz,\bu)  d\bc d \bz_1 d\bz_2.
    \end{equation}
    Since the triply augmented distribution results from an exact augmentation followed by a double-level augmentation, according to \Cref{prop:XDA,prop:AXDA_TV2}, one has
    \begin{equation}
        \left\| \pi(\bx) -\tilde{\tilde{\pi}}_{\rho}(\bx) \right\|_{\mathrm{TV}} \xrightarrow[\rho \to 0]{} 0.
    \end{equation}
However, in contrast to splitting based-optimization algorithms such as HQS and ADMM, the coupling parameter $\rho$ does not follow a decreasing scheme  along the Gibbs sampler iterations but remains constant beyond the burn-in period. Indeed, scheduling this parameter would result in non-stationary transition kernels, which would break the convergence of the sampler towards the target distribution $\pi_{\rho}$.  
In other words, since the coupling parameter remains constant along iterations, the samples $\left\{\bx^{(t)}\right\}_{t=1}$ produced by HRLwSGS are affected by a second bias with respect to the exact posterior distribution $\pi(\bx)$ defined by \eqref{eq:posterior}. This bias results from the double-level AXDA governed and is governed by the coupling parameter $\rho$ .}

{To control these two approximations with respect to the discretization step $\gamma$ and the coupling parameter $\rho$, a dedicated analysis should have been conducted, similar to the one conducted in \cite{Faye2024,plassier2021dg} in the case of an underlying Euclidean geometry. Since the AXDA schemes adopted in \cref{sec:bayesian_model} operates in a Bregman geometry, such analysis is not straightforward and goes significantly beyond the scope of this paper. Thus it is left for future works.
}

\newcommand{\algocomment}[1]{ \STATEx {\color[rgb]{0.5,0.8,0.5}{\% \textit{#1}}}}
\renewcommand\algorithmicrequire{\textbf{Input:}}
\renewcommand\algorithmicensure{\textbf{Output:}}

\begin{algorithm}[htbp]
\caption{HRLwSGS algorithm for Bayesian Poisson inversion.}
\label{alg:HRLwSGS}
\begin{algorithmic}[1]
\REQUIRE Observation $\by$, matrix $\bH=\{h_{ij}\}$ defining the forward model, regularization parameter $\beta$, coupling parameter $\rho$, step-size $\gamma$, number of burn-in iterations $N_{\mathrm{bi}}$, total number of iterations $N_{\mathrm{MC}}$
\renewcommand\algorithmicrequire{\textbf{Initialization:}}
\REQUIRE $\bx^{(0)}$, $\bz_1^{(0)}$, $\bz_2^{(0)}$
\FOR{$t = 0$ to $N_{\mathrm{MC}} - 1$}
    \STATE Sample count variable from \eqref{eq:sample_n}: \label{algo:sample_n}
    \begin{equation*}
    \bc^{(t+1)} \sim \prod_{i=1}^{m} \text{Multinomial} \left\{y_i, \left( \frac{h_{i1} x_1^{(t)}}{\sum_{j=1}^n h_{ij} x_j^{(t)}}, \dots, \frac{h_{in} x_n^{(t)}}{\sum_{j=1}^n h_{ij} x_j^{(t)}} \right) \right\}
    \end{equation*}
    
    \STATE Sample variable of interest from \eqref{eq:sample_x}:
    \begin{equation*}
    \bx^{(t+1)} \sim \prod_{j=1}^n \text{Gamma}\left( x_j^{(t)}; \sum_{i=1}^m c_{ij}^{(t+1)} + \frac{1}{\rho} + 1,\; \alpha \sum_{i=1}^m h_{ij} + \frac{1}{\rho z_{2j}^{(t)}} \right)
    \end{equation*}
    
    \STATE Sample splitting variable from \eqref{eq:sample_z1} using HRLMC:
    \begin{equation*}
    \bz_1^{(t+1)} = \nabla \phi^* \left( \nabla \phi(\bz_1^{(t)}) - \gamma \nabla U(\bz_1^{(t)}) + \sqrt{2\gamma \nabla^2 \phi(\bz_1^{(t)}) } \boldsymbol{\varepsilon}^{(t)} \right),\quad \boldsymbol{\varepsilon}^{(t)} \sim \mathcal{N}(\boldsymbol{0}, \mathbf{I})
    \end{equation*}

    \STATE  Sample splitting variable from \eqref{eq:sample_z2}:
    \begin{equation*}
    \bz_2^{(t+1)} \sim \prod_{j=1}^n \text{InvGamma}\left( z_{2j}^{(t)};\, \frac{2}{\rho},\, \frac{x_j^{(t+1)} + z_{1j}^{(t+1)}}{\rho} \right)
    \end{equation*}
\ENDFOR
\ENSURE Collection of samples $\left\{\bx^{(t)}, \bc^{(t)}, \bz_1^{(t)}, \bz_2^{(t)}\right\}_{t = N_{\mathrm{bi}} + 1}^{N_{\mathrm{MC}}}$
\end{algorithmic}
\end{algorithm}

\section{Numerical experiments}\label{sec:experiments}

In this section, we demonstrate the effectiveness of the proposed Bayesian sampling method by conducting a series of experiments related to Poisson inverse problems, namely image denoising, image deblurring, and positron emission tomography (PET) reconstruction. These tasks have been selected to demonstrate the ability of the proposed method to handle inverse problems of varying difficulty, including variations in the conditioning of the forward operator $\bH$, the dimensionality of the images $\bx$ and $\by$, and the severity of the noise $\alpha$. 

For denoising and deblurring, we use 30 RGB images of size $256 \times 256$ from the ImageNet dataset~\cite{deng2009imagenet}. For the PET experiment, we use a standard $128 \times 128$ synthetic phantom. In all experiments, the measurements are corrupted with Poisson noise, which is typical for photon-limited imaging scenarios. For the denoising and deblurring tasks, the noise level is controlled by the scaling factor $\alpha > 0$, which adjusts the intensity of the clean image before applying the Poisson degradation. It is worth noting that a lower value of $\alpha$ corresponds to a higher level of noise, i.e.,  a lower signal-to-noise ratio. To evaluate the proposed method under varying noise conditions, two distinct noise levels have been considered: $\alpha = 10$ (severe noise) and $\alpha = 40$ (moderate noise). 

{The Bayesian model and the corresponding HRLwSGS algorithm have been derived in \cref{sec:bayesian_model} and \cref{sec:proposed_algorithm} without specifying the regularization potential $g(\cdot)$. To illustrate the versatility of the proposed approach, which can accommodate various forms of regularization, it is instantiated for two particular choices of this potential. First, we consider the case where  $g(\cdot)$ is defined as the RED potential (see \cref{subsec:RED} for more details). The denoiser $\mathsf{D}_{\nu}(\cdot)$ is chosen as  DRUNet and has been used off-the-shelf\footnote{Available online at \url{https://github.com/cszn/DPIR}.}, i.e., it has been pretrained on images corrupted by Gaussian noise and not to specifically handle Poisson noise \cite{zhang2021plug}. In the sequel, the resulting inversion method is referred to as RED-HRLwSGS, i.e., by preceding the name of the proposed algorithm by the considered regularization. Moreover, the proposed HRLwSGS algorithm has been also instantiated when the regularization potential $g(\cdot)$ derives from a Bregman score denoiser (BSD) \cite{hurault2023convergent} (see \cref{sec:bregman_score_denoiser}). Again, in this case, the DRUNet denoiser is used off-the-shelf\footnote{Available online at \url{https://github.com/samuro95/BregmanPnP}}, i.e., pretrained on images corrupted by inverse Gamma noise, consistently with the task underlying the (prior) sampling step discussed in \cref{subsec:non_standard_sampling}. The corresponding variant of the algorithm will be referred to as BSD-HRLwSGS.}

We compare these methods with several state-of-the-art algorithms specifically designed for solving inverse problems under Poisson noise. These include the optimization-based ADMM algorithm TV-PIDAL~\cite{figueiredo2010restoration}, {the plug-and-play method ProxDiffPIR which embeds a diffusion model (DM) as an implicit regularization \cite{melidonis2025score}}, and several Monte Carlo sampling approaches. More precisely, we consider the following methods: \emph{i)} TV-SPA, which leverages an AXDA strategy combined with total variation (TV) regularization~\cite{Vono2019icassp}; \emph{ii)} {PnP-MLA \cite{klatzer2025efficient}, which derives a plug-and-play instance of HRLMC by incorporating the DRUNet denoiser} and \emph{iii)} RPnP-ULA, which relies on a reflected Langevin algorithm coupled with a plug-and-play denoiser, such as DRUNet (RED-RPnP-ULA) or BSD (BSD-RPnP-ULA), to sample from constrained posteriors~\cite{melidonis2023efficient}. 

TV-PIDAL and {ProxDiffPIR} only provide point estimates. In contrast, sampling-based methods such as TV-SPA, {PnP-MLA}, RPnP-ULA and the proposed algorithm HRLwSGS draw samples from the posterior distribution. For these methods, the solution to the Poisson inversion problem has been chosen as the minimum mean square error (MMSE) estimator, which is approximated by averaging the generated samples as follows 
\begin{equation}
\hat{\bx}_{\mathrm{MMSE}} = \frac{1}{N_{\mathrm{MC}}-N_{\mathrm{bi}}} \sum_{t=N_{\mathrm{bi}}+1}^{N_{\mathrm{MC}}} \bx^{(t)}
\end{equation}
where $N_{\mathrm{MC}}$  is the total number of iterations including $N_{\mathrm{bi}}$ burn-in iterations. Besides, these sampling-based methods not only enable point estimation, but also the quantification of uncertainty. Therefore, the results provided by these methods will be granted with uncertainty quantification in the form of pixelwise posterior standard deviation and coverage maps  {(see \cref{app:HPD} for details)}.  {For all tasks, the total number of iterations and the burn-in iterations of the sampling-based methods have been set to $N_{\mathrm{MC}}=100\,000$ and $N_{\mathrm{bi}}=30\,000$ respectively. These choices were primarily guided by empirical convergence observations across all experiments. Given the different  mixing rates across methods observed while assessing convergence (see \cref{subsec:tomography}), the burn-in periods may have been shortened for some compared methods.}

In addition to visual inspection, the methods are compared with respect to several quantitative figures-of-merit. Peak signal-to-noise ratio (PSNR) in decibels (dB) and structural similarity index (SSIM) \cite{wang2004image} are considered as image quality metrics (the higher the score, the better the reconstruction). To capture perceptual differences that are more aligned with human visual perception, we also include the Learned Perceptual Image Patch Similarity (LPIPS) \cite{zhang2018unreasonable} metric, with lower scores denoting closer resemblance to the reference image. {Besides, to assess the quality of the uncertainty estimates produced by the sampling methods, we report the Expected Calibration Error (ECE), which provides a quantitative measure of the agreement between the predicted confidence levels and the empirical frequency of the corresponding events. A lower ECE therefore indicates better calibration, which is particularly important in the Poisson noise setting, where the objective is not only to obtain accurate point estimates but also to correctly characterize the uncertainty induced by the observation process. In particular, a well-calibrated method should produce samples whose dispersion and uncertainty are consistent with the actual variability of the noisy observations.}

{The hyperparameters are selected according to a common protocol designed to ensure a fair comparison across the considered sampling methods. In particular, when several algorithms rely on the same denoising prior, the parameters $\beta$ and $\nu$ controlling the weight of the prior and the denoising strength, respectively , are fixed to the same values across methods for a given noise level. For HRLwSGS, the prior is reformulated so that its regularization gradient coincides with that used by PnP-MLA and RPnP-ULA. This formulation ensures that the three methods are driven by the same prior-induced regularization term. The remaining algorithm-dependent parameters, including the step size $\gamma$, are tuned to achieve a balance between reconstruction fidelity and the quality of the resulting uncertainty estimates. We use PSNR as the primary measure of reconstruction fidelity and ECE to assess calibration. This trade-off is essential in the considered Poisson setting, since the noise is signal-dependent and its variance changes with the underlying intensity. Consequently, a method may achieve a high PSNR while producing posterior samples that are either excessively concentrated or overly dispersed, leading respectively to overconfident or conservative uncertainty estimates. Calibration is therefore particularly important for the considered sampling-based methods, whose purpose is not restricted to recovering a high-quality point estimate but also includes characterizing the uncertainty associated with the reconstruction. The selected hyperparameters consequently aim at preserving reconstruction quality while ensuring that the uncertainty represented by the generated samples remains consistent with the empirical behavior of the observations.} Further information regarding the experimental settings and parameters of the compared methods are reported in \cref{app:parameters}.


\subsection{Poisson image denoising}\label{subsec:denoising}
This experiment assesses the performance of the compared algorithms {when} restoring images only corrupted by Poisson noise. In this simplified experimental scheme, the degradation operator is the identity matrix, $\bH = \bI$, meaning that the noise is applied directly to the image, with no additional transformations, such as convolution or masking, being applied. This setup enables a dedicated evaluation of the denoising capabilities of each method. Examples of noise-free and noisy images are depicted in \cref{fig:imagenet_denoising_1} for the two considered levels of noise ($\alpha=40$ and $\alpha=10$).

\newlength\colgap
\newlength\imw
\setlength\labelwidth{.05\linewidth}                
\setlength\colgap{.1pt}                               
\setlength\imw{(\linewidth - \labelwidth - \colgap)/{5}}  
\newcolumntype{M}{@{}m{\imw}@{}}                     
\newcolumntype{L}{@{}m{\labelwidth}@{}}              

\newcommand{\rotlabel}[1]{%
  \begin{minipage}[c][\imw][c]{\labelwidth}\centering
    \rotatebox{90}{\tiny#1}
  \end{minipage}%
}
\newcommand{\sqimg}[1]{%
  \makebox[\imw][c]{\includegraphics[width=\imw,height=\imw]{#1}}%
}

\begin{figure}[!h]
\centering
\setlength{\tabcolsep}{0pt}
\renewcommand{\arraystretch}{1}
\begin{tabular}{L  M  M  M  M}
& \makebox[\imw]{\scriptsize $\alpha=40$} &
  \makebox[\imw]{\scriptsize $\alpha=40$} &
  \makebox[\imw]{\scriptsize $\alpha=10$} &
  \makebox[\imw]{\scriptsize $\alpha=10$} \\
\rotlabel{Ground Truth} &
\begin{tikzpicture}[zoomboxarray]
    \node [image node] {\sqimg{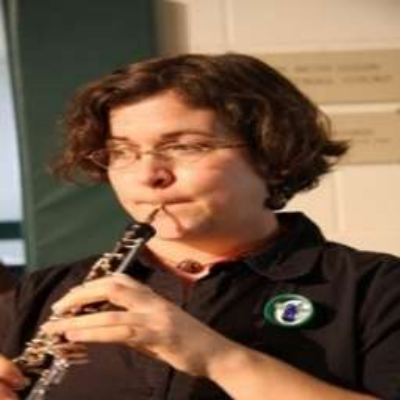}};
	\zoombox{0.3, 0.65}
\end{tikzpicture} &
\begin{tikzpicture}[zoomboxarray]
    \node [image node] {\sqimg{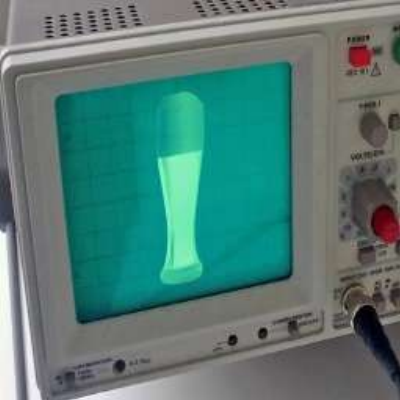}};
	\zoombox{0.88, 0.85}
\end{tikzpicture} &
\begin{tikzpicture}[zoomboxarray]
    \node [image node] {\sqimg{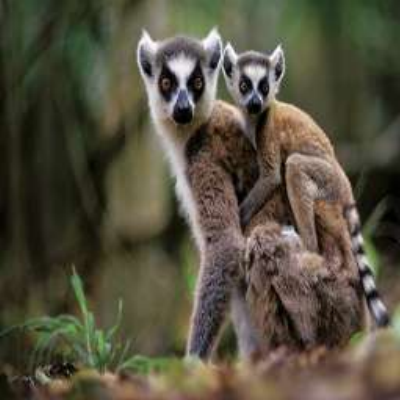}};
	\zoombox{0.45, 0.8}
\end{tikzpicture} &
\begin{tikzpicture}[zoomboxarray]
    \node [image node] {\sqimg{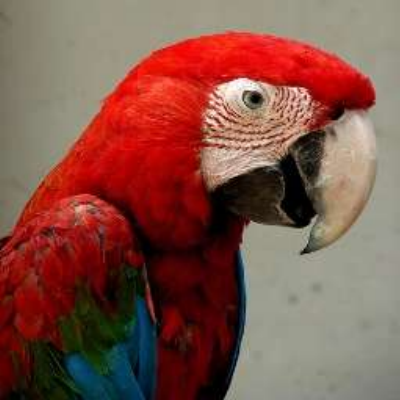}};
	\zoombox{0.6, 0.72}
\end{tikzpicture}\\[-2pt]
\rotlabel{Observation} &
\begin{tikzpicture}[zoomboxarray]
    \node [image node]{\sqimg{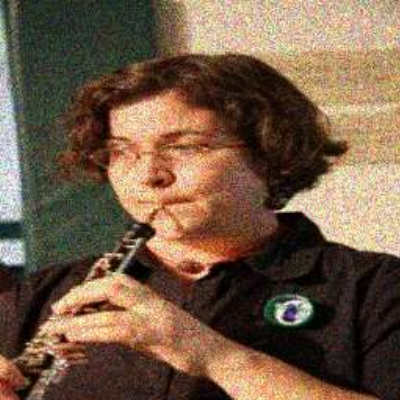}};
	\zoombox{0.3, 0.65}
\end{tikzpicture} &
\begin{tikzpicture}[zoomboxarray]
    \node [image node]{\sqimg{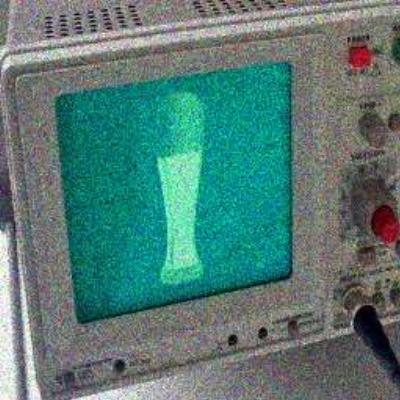} };
	\zoombox{0.88, 0.85}
\end{tikzpicture} &
\begin{tikzpicture}[zoomboxarray]
    \node [image node] {\sqimg{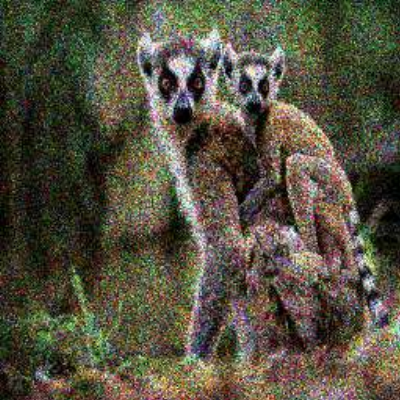}};
	\zoombox{0.45, 0.8}
\end{tikzpicture} &
\begin{tikzpicture}[zoomboxarray]
    \node [image node] {\sqimg{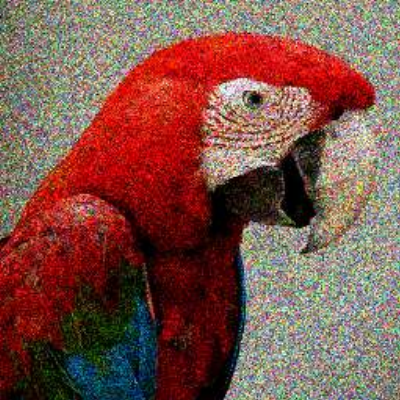}};
	\zoombox{0.6, 0.72}
\end{tikzpicture}\\[-2pt]
\end{tabular}
\caption{Denoising experiment: ground truth (top) and observed (bottom) images for $\alpha = 40$ (first two columns) and $\alpha = 40$ (last two columns).}
\label{fig:imagenet_denoising_1}
\end{figure}

{
Quantitative results are reported in \cref{tab:imagenet_denoising}. Under severe noise conditions ($\alpha=10$), ProxDiffPIR provides the best overall reconstruction quality, achieving the highest PSNR (28.06~dB), SSIM (0.769), and the lowest LPIPS (0.246). Among the sampling-based methods, BSD-HRLwSGS achieves the highest PSNR (27.39~dB) and SSIM (0.670), together with an LPIPS of 0.313, outperforming both RED-HRLwSGS and PnP-MLA in terms of reconstruction fidelity. In particular, BSD-HRLwSGS substantially improves upon BSD-RPnP-ULA, which obtains the lowest PSNR (23.89~dB) and SSIM (0.557) among the sampling-based approaches. This gain in reconstruction quality, however, comes with a less favorable calibration: BSD-RPnP-ULA achieves the lowest ECE among the sampling methods (0.039), whereas BSD-HRLwSGS obtains an ECE of 0.113. The same trade-off is observed for the RED prior, with RED-RPnP-ULA yielding a slightly lower ECE (0.102) than RED-HRLwSGS (0.114), but substantially lower PSNR and SSIM.
Under the moderate noise setting ($\alpha=40$), the reconstruction quality generally improves across all methods. ProxDiffPIR remains the best-performing method in terms of PSNR (31.33~dB), SSIM (0.817), and LPIPS (0.190). Among the sampling-based methods, BSD-HRLwSGS provides the highest PSNR (29.73~dB) and SSIM (0.734), followed by RED-HRLwSGS (29.37~dB and 0.703) and PnP-MLA (29.24~dB and 0.697). In contrast, the calibration results again favor BSD-RPnP-ULA, which achieves the lowest ECE (0.037), substantially lower than BSD-HRLwSGS (0.097). For the RED prior, RED-RPnP-ULA obtains a slightly lower ECE (0.090) than RED-HRLwSGS (0.097), while the latter provides better reconstruction fidelity across all three image-quality metrics. Overall, these results reveal a consistent trade-off between reconstruction accuracy and uncertainty calibration. HRLwSGS, particularly with the BSD prior, provides the strongest reconstruction performance among the sampling-based methods, whereas RPnP-ULA tends to produce better calibrated uncertainty estimates at the expense of reconstruction quality.
}

\begin{table}[h!]\small
\centering
\caption{Denoising experiment: average performance and corresponding standard deviations.\label{tab:imagenet_denoising}}%
\begin{tabular}{lllcccccc} 
\toprule
& & {\textbf{PSNR}} (dB)$\uparrow$ & {\textbf{SSIM}}$\uparrow$ & {\textbf{LPIPS}} $\downarrow$ & {\textbf{ECE}} $\downarrow$  \\
\midrule
\multirow{8}{*}{\rotatebox{90}{{\textbf{$\alpha=10$}}}}
	& TV-PIDAL & 25.655$\pm$2.388 & 0.7156$\pm$0.0955 & 0.3704$\pm$0.0688 & -- \\
    & ProxDiffPIR & \textbf{28.064$\pm$1.643} & \textbf{0.7693$\pm$0.0661} & \textbf{0.2458$\pm$0.0639} & -- \\    
    & TV-SPA & 24.561$\pm$1.444 & 0.6020$\pm$0.0467 & 0.3895$\pm$0.0508 & 0.0546$\pm$0.0526 \\
    & {PnP-MLA} & 26.671$\pm$2.019 & 0.6574$\pm$0.0911 & 0.3219$\pm$0.0706 & 0.1177$\pm$0.0470 \\    
	& RED-RPnP-ULA & 25.954$\pm$1.430 & 0.5898$\pm$0.1124 & 0.3352$\pm$0.0772 & 0.1017$\pm$0.0376 \\
	& BSD-RPnP-ULA & 23.894$\pm$1.390 & 0.5567$\pm$0.0875 & 0.3799$\pm$0.0520 & 0.0394$\pm$0.0237 \\
	& RED-HRLwSGS & 26.446$\pm$1.783 & 0.6580$\pm$0.0824 & 0.3225$\pm$0.0651 & 0.1136$\pm$0.0436 \\
	& BSD-HRLwSGS & 27.386$\pm$1.418 & 0.6695$\pm$0.0867 & 0.3127$\pm$0.0608 & 0.1131$\pm$0.0500 \\
\midrule
\multirow{8}{*}{\rotatebox{90}{{\textbf{$\alpha=40$}}}}
	& TV-PIDAL & 29.102$\pm$2.105 & 0.8051$\pm$0.0592 & 0.2585$\pm$0.0512 & -- \\
    & {ProxDiffPIR} & \textbf{31.334$\pm$1.389} & \textbf{0.8167$\pm$0.0432} & \textbf{0.1901$\pm$0.0621} & --  \\
	& TV-SPA & 27.599$\pm$1.830 & 0.7405$\pm$0.0519 & 0.2732$\pm$0.0533 & 0.0817$\pm$0.0531 \\
    & {PnP-MLA} & 29.235$\pm$1.040 & 0.6973$\pm$0.0825 & 0.2424$\pm$0.0694 & 0.0929$\pm$0.0461 \\
	& RED-RPnP-ULA & 28.846$\pm$1.104 & 0.6815$\pm$0.1035 & 0.2468$\pm$0.0749 & 0.0900$\pm$0.0236  \\
	& BSD-RPnP-ULA & 28.451$\pm$1.752 & 0.7135$\pm$0.1249 & 0.2613$\pm$0.0728 & 0.0368$\pm$0.0253 \\
	& RED-HRLwSGS & 29.367$\pm$1.022 & 0.7029$\pm$0.0805 & 0.2419$\pm$0.0687 & 0.0970$\pm$0.0472 \\
	& BSD-HRLwSGS & 29.730$\pm$1.650 & 0.7338$\pm$0.1045 & 0.2516$\pm$0.0671 & 0.0967$\pm$0.0485 \\
\bottomrule
\end{tabular}
\end{table}

{
\Cref{fig:imagenet_denoising} presents representative qualitative results for both noise levels ($\alpha=40$ for the first two columns and $\alpha=10$ for the last two columns). The visual comparisons are broadly consistent with the quantitative results, showing that HRLwSGS generally produces reconstructions with well-preserved structures and fewer visible artifacts among the sampling-based methods. The advantage of HRLwSGS is particularly noticeable when using the BSD prior. Under the severe-noise setting ($\alpha=10$), BSD-RPnP-ULA exhibits strong artifacts and fails to recover several fine-scale structures, in agreement with its substantially lower PSNR and SSIM. In contrast, BSD-HRLwSGS provides a considerably more stable reconstruction, with sharper structures and fewer artifacts, resulting in a PSNR improvement of more than $3.5$~dB over BSD-RPnP-ULA. This improvement in reconstruction quality is achieved at the cost of a higher ECE, illustrating the reconstruction-calibration trade-off observed in the quantitative results.
}

\begin{figure}[!h]
\centering
\setlength{\tabcolsep}{0pt}
\renewcommand{\arraystretch}{1}
\begin{tabular}{L  M  M  M  M}
& \makebox[\imw]{\scriptsize $\alpha=40$} &
  \makebox[\imw]{\scriptsize $\alpha=40$} &
  \makebox[\imw]{\scriptsize $\alpha=10$} &
  \makebox[\imw]{\scriptsize $\alpha=10$} \\

\rotlabel{ProxDiffPIR} & 
\begin{tikzpicture}[zoomboxarray]
    \node [image node] {\sqimg{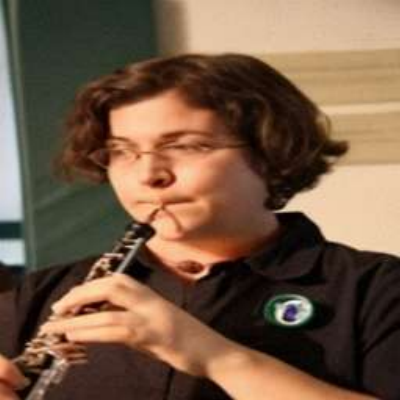}};
	\zoombox{0.3, 0.65}
\end{tikzpicture} &
\begin{tikzpicture}[zoomboxarray]
    \node [image node] {\sqimg{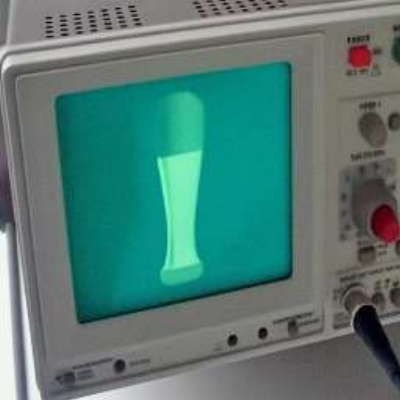}};
	\zoombox{0.88, 0.85}
\end{tikzpicture} &
\begin{tikzpicture}[zoomboxarray]
    \node [image node] {\sqimg{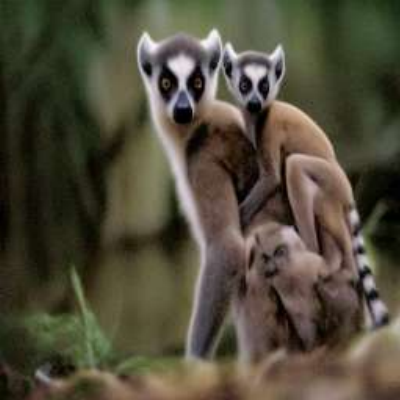}};
	\zoombox{0.45, 0.8}
\end{tikzpicture} &
\begin{tikzpicture}[zoomboxarray]
    \node [image node] {\sqimg{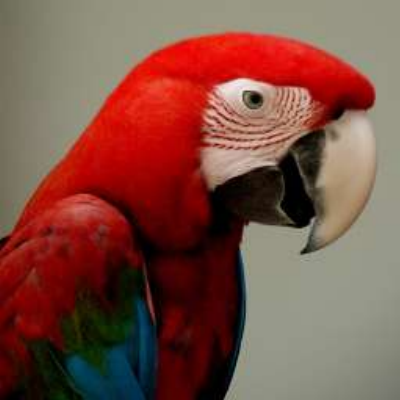}};
	\zoombox{0.6, 0.72}
\end{tikzpicture}\\[-2pt]

\rotlabel{PnP-MLA} &
\begin{tikzpicture}[zoomboxarray]
    \node [image node]{\sqimg{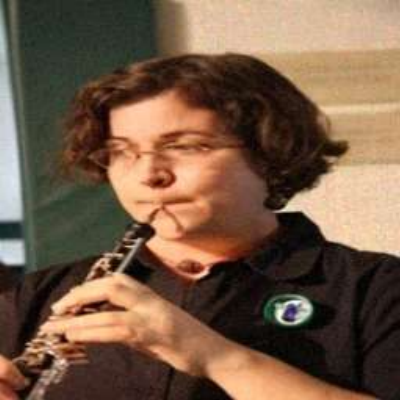}};
	\zoombox{0.3, 0.65}
\end{tikzpicture} &
\begin{tikzpicture}[zoomboxarray]
    \node [image node]{\sqimg{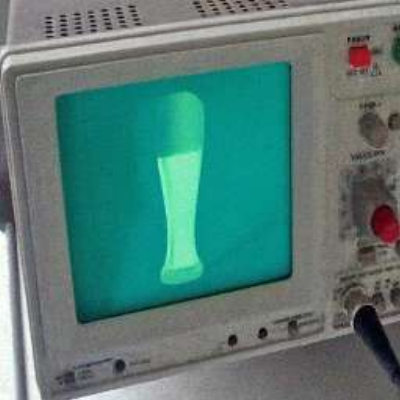} };
	\zoombox{0.88, 0.85}
\end{tikzpicture} &
\begin{tikzpicture}[zoomboxarray]
    \node [image node] {\sqimg{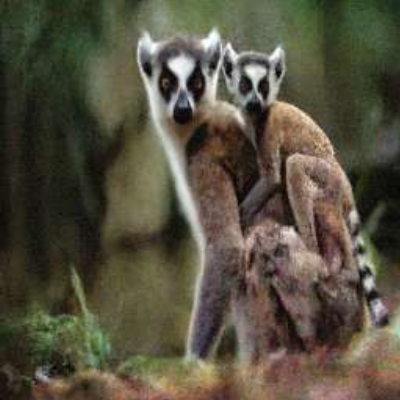}};
	\zoombox{0.45, 0.8}
\end{tikzpicture} &
\begin{tikzpicture}[zoomboxarray]
    \node [image node] {\sqimg{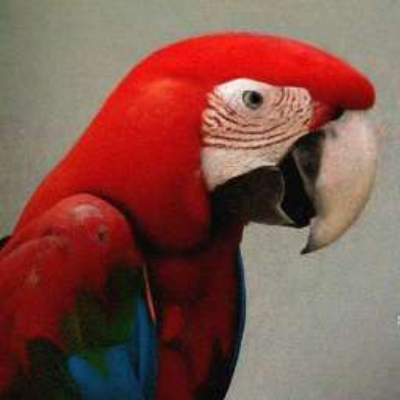}};
	\zoombox{0.6, 0.72}
\end{tikzpicture}\\[-2pt]

\rotlabel{RED-RPnP-ULA} &
\begin{tikzpicture}[zoomboxarray]
    \node [image node]{\sqimg{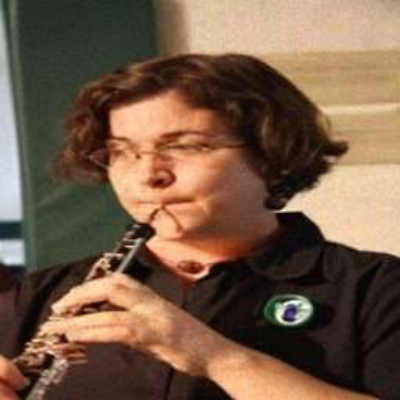} };
	\zoombox{0.3, 0.65}
\end{tikzpicture} &
\begin{tikzpicture}[zoomboxarray]
    \node [image node]{\sqimg{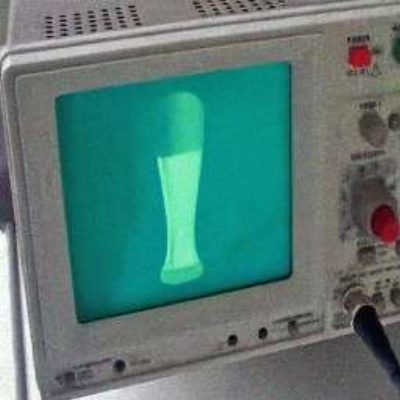}};
	\zoombox{0.88, 0.85}
\end{tikzpicture} &
\begin{tikzpicture}[zoomboxarray]
    \node [image node] {\sqimg{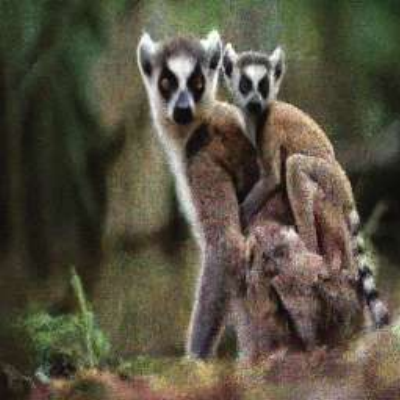} };
	\zoombox{0.45, 0.8}
\end{tikzpicture} &
\begin{tikzpicture}[zoomboxarray]
    \node [image node] {\sqimg{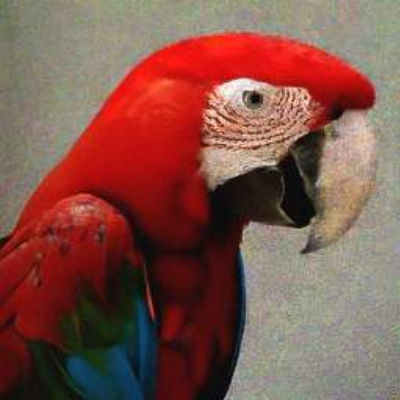} };
	\zoombox{0.6, 0.72}
\end{tikzpicture}\\[-2pt]

\rotlabel{BSD-RPnP-ULA} &
\begin{tikzpicture}[zoomboxarray]
    \node [image node]{\sqimg{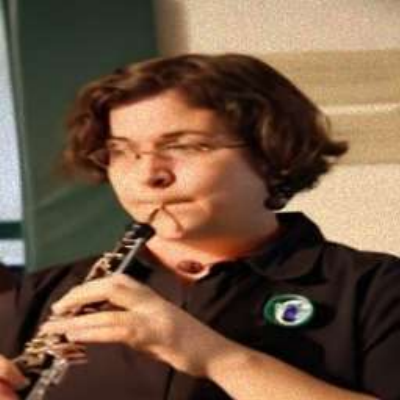}};
	\zoombox{0.3, 0.65}
\end{tikzpicture} &
\begin{tikzpicture}[zoomboxarray]
    \node [image node]{\sqimg{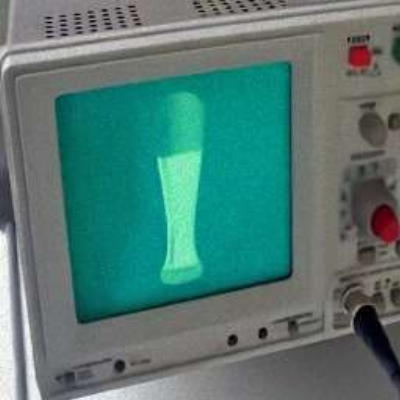}};
	\zoombox{0.88, 0.85}
\end{tikzpicture} &
\begin{tikzpicture}[zoomboxarray]
    \node [image node] {\sqimg{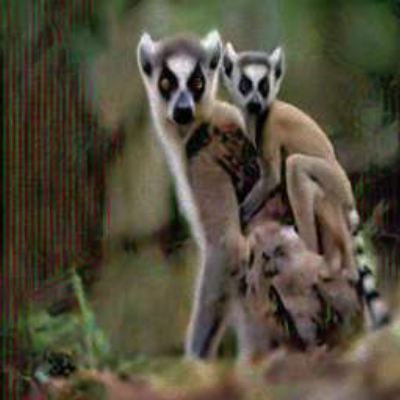}};
	\zoombox{0.45, 0.8}
\end{tikzpicture} &
\begin{tikzpicture}[zoomboxarray]
    \node [image node] {\sqimg{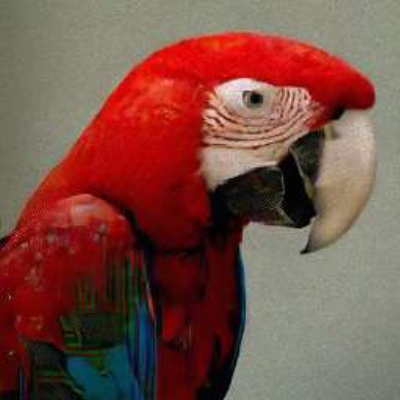} };
	\zoombox{0.6, 0.72}
\end{tikzpicture}\\[-2pt]

\rotlabel{RED-HRLwSGS} &
\begin{tikzpicture}[zoomboxarray]
    \node [image node]{\sqimg{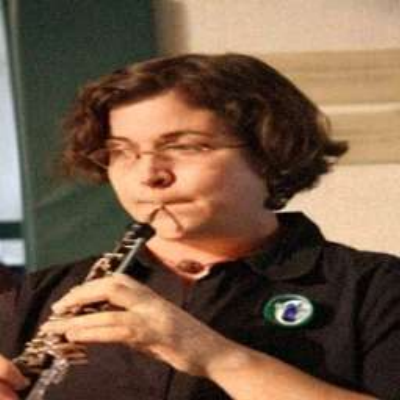}};
	\zoombox{0.3, 0.65}
\end{tikzpicture} &
\begin{tikzpicture}[zoomboxarray]
    \node [image node]{\sqimg{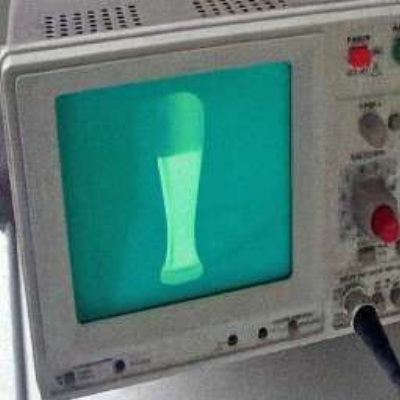} };
	\zoombox{0.88, 0.85}
\end{tikzpicture} &
\begin{tikzpicture}[zoomboxarray]
    \node [image node] {\sqimg{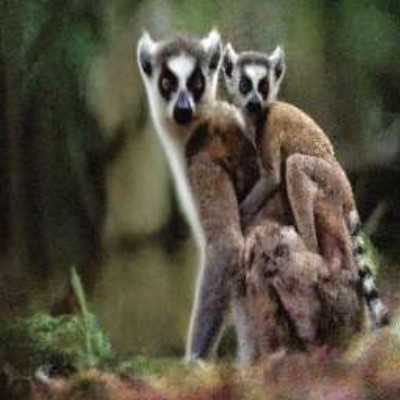} };
	\zoombox{0.45, 0.8}
\end{tikzpicture} &
\begin{tikzpicture}[zoomboxarray]
    \node [image node] {\sqimg{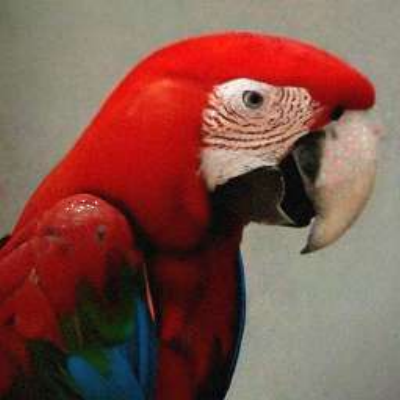}};
	\zoombox{0.6, 0.72}
\end{tikzpicture}\\[-2pt]


\rotlabel{BSD-HRLwSGS} &
\begin{tikzpicture}[zoomboxarray]
    \node [image node]{\sqimg{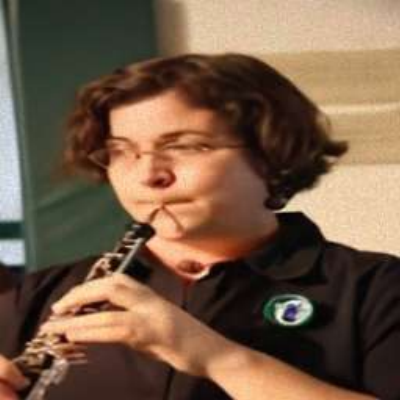}};
	\zoombox{0.3, 0.65}
\end{tikzpicture} &
\begin{tikzpicture}[zoomboxarray]
    \node [image node]{\sqimg{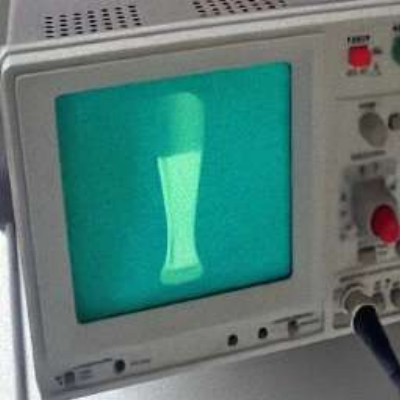}};
	\zoombox{0.88, 0.85}
\end{tikzpicture} &
\begin{tikzpicture}[zoomboxarray]
    \node [image node] {\sqimg{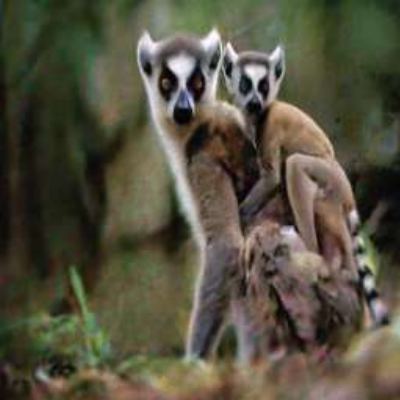} };
	\zoombox{0.45, 0.8}
\end{tikzpicture} &
\begin{tikzpicture}[zoomboxarray]
    \node [image node] {\sqimg{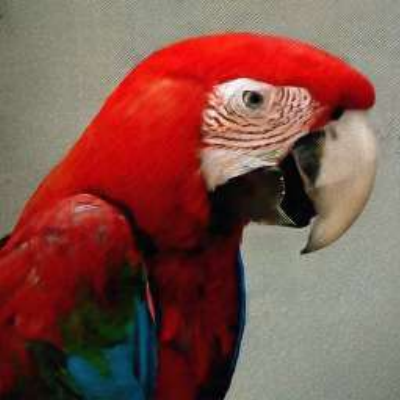} };
	\zoombox{0.6, 0.72}
\end{tikzpicture}\\[-2pt]
\end{tabular}
\caption{Denoising experiment: visual results. Due to space constraints, the results recovered by TV-SPA and TV-PIDAL are not reproduced, as their quality is significantly inferior. They are reported in \cref{app:complementary_results} for comprehensiveness.}
\label{fig:imagenet_denoising}
\end{figure}

\subsection{Poisson image deconvolution}\label{subsec:deconvolution}
We evaluate the proposed method on the Poisson image deconvolution task, where the degradation operator $\bH$ is modeled as a circulant convolution matrix associated with a spatially invariant Gaussian kernel of size $15 \times 15$ and standard deviation $1.6$. This setting simulates realistic optical blur and constitutes a more challenging inverse problem than  the denoising scenario due to the ill-posedness induced by the convolution. Examples of ground-truth and blurred images are depicted in \cref{fig:imagenet_deblurring_gt_and_obs} for the two considered levels of noise ($\alpha=40$ and $\alpha=10$).

\begin{figure}[!h]
\centering
\setlength{\tabcolsep}{0pt}
\renewcommand{\arraystretch}{1}
\begin{tabular}{L  M  M  M  M}
& \makebox[\imw]{\scriptsize $\alpha=40$} &
  \makebox[\imw]{\scriptsize $\alpha=40$} &
  \makebox[\imw]{\scriptsize $\alpha=10$} &
  \makebox[\imw]{\scriptsize $\alpha=10$} \\
\rotlabel{Ground Truth} &
\begin{tikzpicture}[zoomboxarray]
    \node [image node] {\sqimg{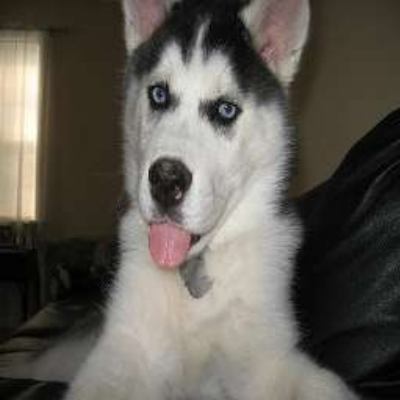}};
	 \zoombox{0.54, 0.74}
\end{tikzpicture} &
\begin{tikzpicture}[zoomboxarray]
    \node [image node] {\sqimg{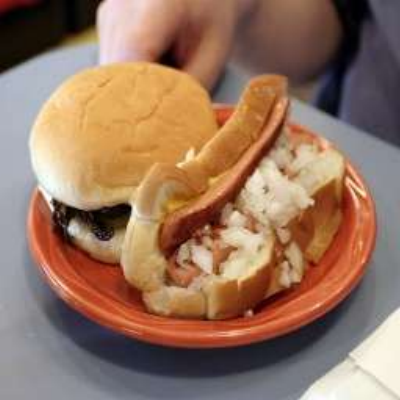}};
	\zoombox{0.58, 0.55}
\end{tikzpicture} &
\begin{tikzpicture}[zoomboxarray]
    \node [image node] {\sqimg{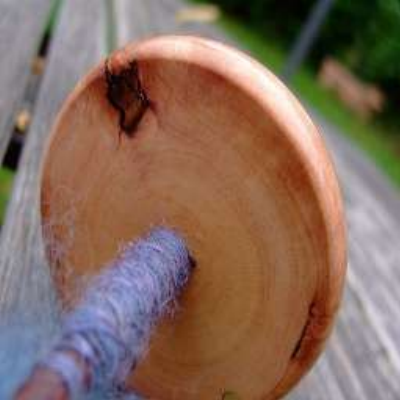}};
	\zoombox{0.3, 0.75}
\end{tikzpicture} &
\begin{tikzpicture}[zoomboxarray]
    \node [image node] {\sqimg{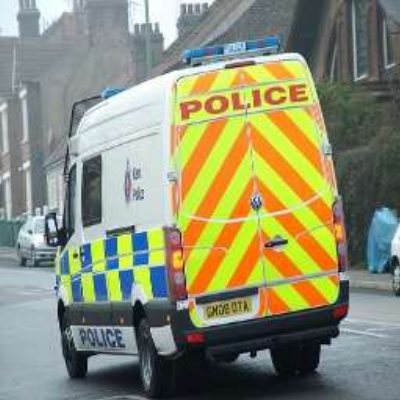}};
	\zoombox{0.52, 0.75}
\end{tikzpicture}\\[-2pt]
\rotlabel{Observation} &
\begin{tikzpicture}[zoomboxarray]
    \node [image node]{\sqimg{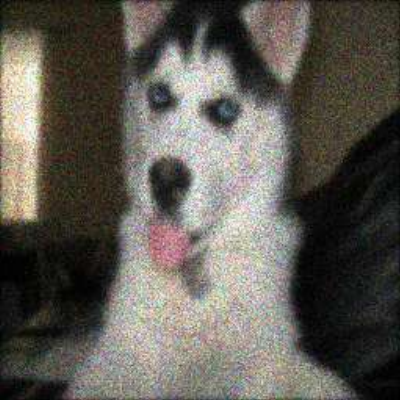}};
	\zoombox{0.54, 0.74}
\end{tikzpicture} &
\begin{tikzpicture}[zoomboxarray]
    \node [image node]{\sqimg{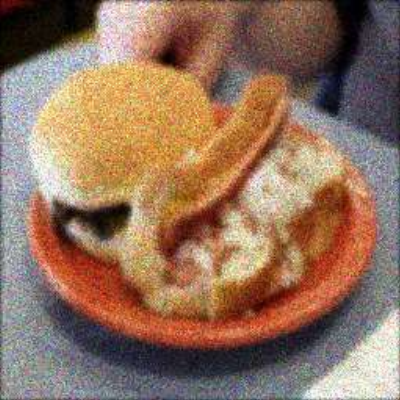} };
	\zoombox{0.58, 0.55}
\end{tikzpicture} &
\begin{tikzpicture}[zoomboxarray]
    \node [image node] {\sqimg{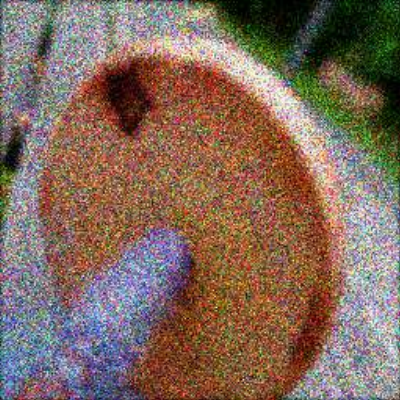}};
	\zoombox{0.3, 0.75}
\end{tikzpicture} &
\begin{tikzpicture}[zoomboxarray]
    \node [image node] {\sqimg{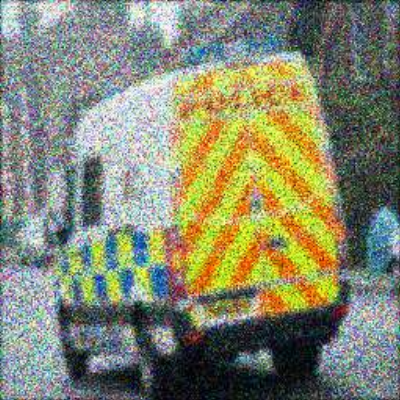}};
	\zoombox{0.52, 0.75}
\end{tikzpicture}\\[-2pt]
\end{tabular}
\caption{Deblurring experiment: ground truth (top) and observed (bottom) images for $\alpha = 40$ (first two columns) and $\alpha = 10$ (last two columns).}
\label{fig:imagenet_deblurring_gt_and_obs}
\end{figure}

Quantitative performance is reported in \cref{tab:imagenet_deblurring} for two noise levels $\alpha = 10$ and $\alpha = 40$, while qualitative results are provided in \Cref{fig:imagenet_deblurring}, and uncertainty quantification is analyzed in \Cref{fig:uncertainty_coverage_map}.

\begin{table}[h!]\small
\centering
\caption{Deblurring experiment: average performance and corresponding standard deviations.\label{tab:imagenet_deblurring}}
\begin{tabular}{lllcccccc} 
\toprule
& & {\textbf{PSNR}} (dB)$\uparrow$ & {\textbf{SSIM}}$\uparrow$ & {\textbf{LPIPS}} $\downarrow$ & {\textbf{ECE}} $\downarrow$  \\
\midrule
\multirow{8}{*}{\rotatebox{90}{{\textbf{$\alpha=10$}}}}
	& TV-PIDAL & 25.106$\pm$2.197 & 0.6833$\pm$0.0904 & 0.4392$\pm$0.0574 & -- \\
    & ProxDiffPIR & \textbf{26.953$\pm$2.644} & \textbf{0.7655$\pm$0.0899} & \textbf{0.3035$\pm$0.0811} & -- \\    
    & TV-SPA & 23.661$\pm$1.933 & 0.6212$\pm$0.0817 & 0.4377$\pm$0.0537 & 0.0671$\pm$0.0555 \\
    & {PnP-MLA} & 24.980$\pm$2.258 & 0.6240$\pm$0.1128 & 0.4036$\pm$0.0645 & 0.1063$\pm$0.0727 \\    
	& RED-RPnP-ULA & 25.252$\pm$2.343 & 0.6373$\pm$0.0759 & 0.3947$\pm$0.0645 & 0.1060$\pm$0.0648 \\
	& BSD-RPnP-ULA & 24.009$\pm$1.826 & 0.6401$\pm$0.0981 & 0.3850$\pm$0.0663 & 0.0639$\pm$0.0454 \\
	& RED-HRLwSGS & 25.238$\pm$2.415 & 0.6503$\pm$0.1219 & 0.3863$\pm$0.0738 & 0.1214$\pm$0.0791 \\
	& BSD-HRLwSGS & 25.548$\pm$2.524 & 0.6738$\pm$0.0832 & 0.3794$\pm$0.0675 & 0.1202$\pm$0.0832 \\
\midrule
\multirow{8}{*}{\rotatebox{90}{{\textbf{$\alpha=40$}}}}
	& TV-PIDAL & 26.763$\pm$2.415 & 0.7501$\pm$0.0779 & 0.3652$\pm$0.0569 & -- \\
    & {ProxDiffPIR} & \textbf{28.698$\pm$2.689} & \textbf{0.8177$\pm$0.0712} & \textbf{0.2493$\pm$0.0601} & --  \\
	& TV-SPA & 25.646$\pm$2.233 & 0.6926$\pm$0.0772 & 0.3612$\pm$0.0557 & 0.0973$\pm$0.0583 \\
    & {PnP-MLA} & 26.498$\pm$2.204 & 0.6734$\pm$0.1076 & 0.3422$\pm$0.0644 & 0.1169$\pm$0.0717 \\
	& RED-RPnP-ULA & 26.637$\pm$2.153 & 0.6556$\pm$0.0499 & 0.3502$\pm$0.0505 & 0.1085$\pm$0.0474  \\
	& BSD-RPnP-ULA & 26.495$\pm$2.070 & 0.6711$\pm$0.1111 & 0.3402$\pm$0.0589 & 0.0957$\pm$0.0512 \\
	& RED-HRLwSGS & 26.417$\pm$2.158 & 0.6690$\pm$0.1046 & 0.3460$\pm$0.0616 & 0.1142$\pm$0.0709 \\
	& BSD-HRLwSGS & 26.915$\pm$2.401 & 0.7236$\pm$0.0703 & 0.3196$\pm$0.0604 & 0.1324$\pm$0.0745 \\
\bottomrule
\end{tabular}
\end{table}

{
For $\alpha=10$, ProxDiffPIR achieves the best reconstruction performance across all three image-quality metrics, with the highest PSNR (26.953 dB) and SSIM (0.766), and the lowest LPIPS (0.304). Among the remaining methods, BSD-HRLwSGS obtains the second-highest PSNR (25.548 dB), followed by RED-RPnP-ULA (25.252 dB), while TV-PIDAL achieves the second-highest SSIM (0.683). BSD-HRLwSGS also provides the second-lowest LPIPS (0.379), indicating a favorable perceptual similarity relative to the other uncertainty-aware approaches. Within the HRLwSGS family, BSD-HRLwSGS improves upon RED-HRLwSGS in all three reconstruction metrics, with PSNR increasing from 25.238 to 25.548 dB, SSIM from 0.650 to 0.674, and LPIPS decreasing from 0.386 to 0.379. Compared with TV-SPA, RED-HRLwSGS achieves a higher PSNR (25.238 versus 23.661 dB) and SSIM (0.650 versus 0.621), together with a lower LPIPS (0.386 versus 0.438). These results indicate that, at $\alpha=10$, the BSD-HRLwSGS variant provides a better reconstruction quality than its RED-HRLwSGS counterpart, although both remain below ProxDiffPIR in terms of the three reconstruction metrics.
For $\alpha=40$, ProxDiffPIR again achieves the best reconstruction performance across all three metrics, with a PSNR of 28.698 dB, an SSIM of 0.818, and an LPIPS of 0.249. Among the remaining methods, BSD-HRLwSGS obtains the second-highest PSNR (26.915 dB), while TV-PIDAL achieves the second-highest SSIM (0.750). BSD-HRLwSGS also provides the second-lowest LPIPS (0.320), followed by BSD-RPnP-ULA (0.340) and PnP-MLA (0.342). Within the HRLwSGS family, BSD-HRLwSGS consistently improves upon RED-HRLwSGS, increasing PSNR from 26.417 to 26.915 dB and SSIM from 0.669 to 0.724, while reducing LPIPS from 0.346 to 0.320. RED-HRLwSGS nevertheless achieves a lower LPIPS than TV-SPA (0.346 versus 0.361) and TV-PIDAL (0.346 versus 0.365), despite having lower PSNR and SSIM than these methods. Overall, the results show that the BSD-HRLwSGS variant provides a stronger reconstruction performance than RED-HRLwSGS for both noise levels, particularly in terms of PSNR, SSIM, and LPIPS.
Regarding uncertainty calibration, the lowest ECE among the methods for which uncertainty estimates are available is obtained by BSD-RPnP-ULA at both noise levels, with ECE values of $0.064$ for $\alpha=10$ and $0.096$ for $\alpha=40$. TV-SPA provides closely comparable ECE values of $0.067$ and $0.097$, respectively. PnP-MLA and the two HRLwSGS variants exhibit higher ECE values, with the HRLwSGS methods ranging from $0.114$ to $0.132$ depending on the noise level and variant. In particular, BSD-HRLwSGS has an ECE of $0.120$ for $\alpha=10$ and $0.132$ for $\alpha=40$, whereas RED-HRLwSGS obtains $0.121$ and $0.114$, respectively. These results show that the methods with the lowest ECE are not necessarily those achieving the strongest reconstruction performance. They also highlight a trade-off between reconstruction quality and uncertainty calibration: BSD-HRLwSGS provides substantially stronger reconstruction metrics than BSD-RPnP-ULA, while BSD-RPnP-ULA achieves lower ECE values at both noise levels.
}

\begin{figure}[htbp]
\centering
\setlength{\tabcolsep}{0pt}
\renewcommand{\arraystretch}{1}
\begin{tabular}{L  M  M  M  M}
& \makebox[\imw]{\scriptsize $\alpha=40$} &
  \makebox[\imw]{\scriptsize $\alpha=40$} &
  \makebox[\imw]{\scriptsize $\alpha=10$} &
  \makebox[\imw]{\scriptsize $\alpha=10$} \\

\rotlabel{ProxDiffPIR} &
\begin{tikzpicture}[zoomboxarray]
    \node [image node] {\sqimg{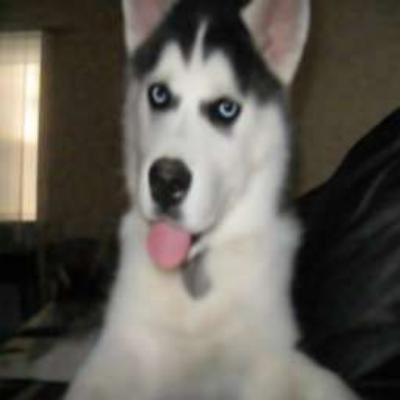}};
	 \zoombox{0.54, 0.74}
\end{tikzpicture} &
\begin{tikzpicture}[zoomboxarray]
    \node [image node] {\sqimg{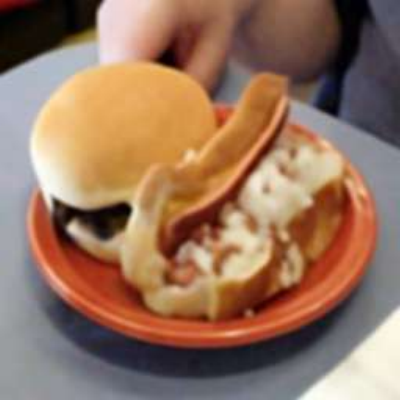}};
	\zoombox{0.58, 0.55}
\end{tikzpicture} &
\begin{tikzpicture}[zoomboxarray]
    \node [image node] {\sqimg{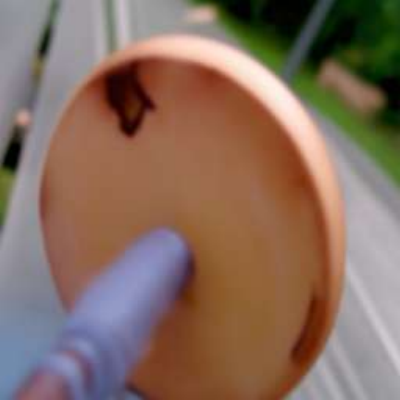}};
	\zoombox{0.3, 0.75}
\end{tikzpicture} &
\begin{tikzpicture}[zoomboxarray]
    \node [image node] {\sqimg{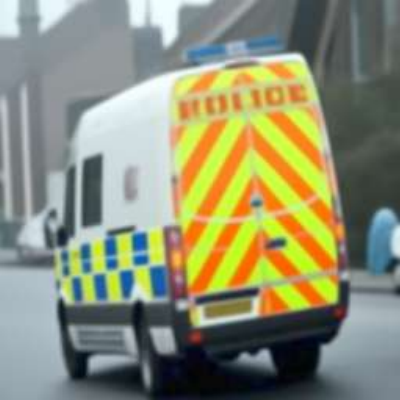}};
	\zoombox{0.52, 0.75}
\end{tikzpicture}\\[-2pt]

\rotlabel{PnP-MLA} &
\begin{tikzpicture}[zoomboxarray]
    \node [image node]{\sqimg{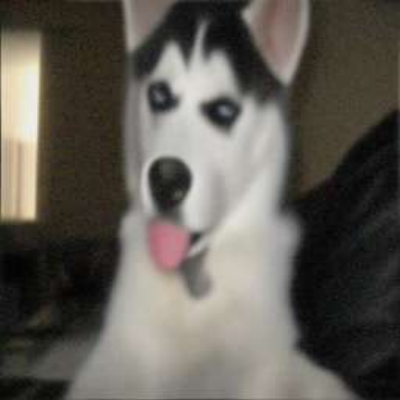}};
	\zoombox{0.54, 0.74}
\end{tikzpicture} &
\begin{tikzpicture}[zoomboxarray]
    \node [image node]{\sqimg{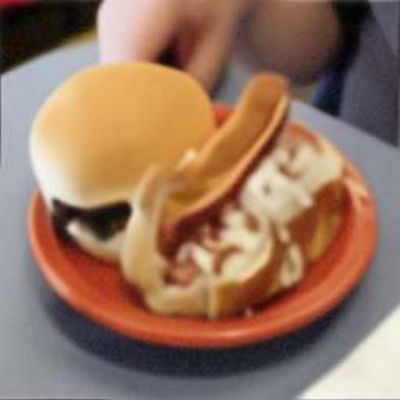} };
	\zoombox{0.58, 0.55}
\end{tikzpicture} &
\begin{tikzpicture}[zoomboxarray]
    \node [image node] {\sqimg{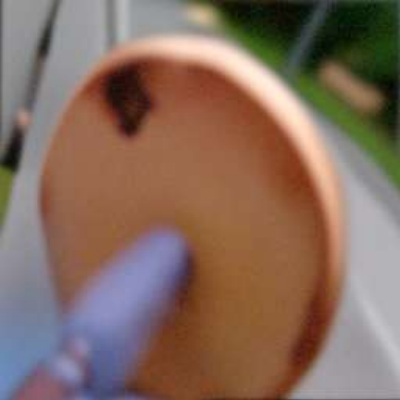}};
	\zoombox{0.3, 0.75}
\end{tikzpicture} &
\begin{tikzpicture}[zoomboxarray]
    \node [image node] {\sqimg{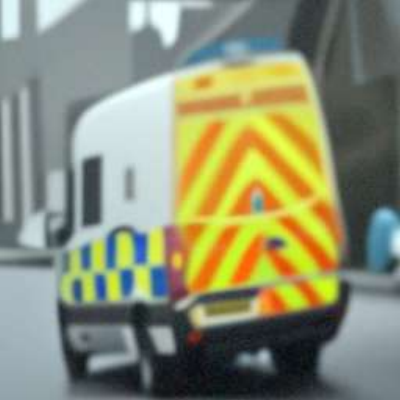}};
	\zoombox{0.52, 0.75}
\end{tikzpicture}\\[-2pt]

\rotlabel{RED-RPnP-ULA} &
\begin{tikzpicture}[zoomboxarray]
    \node [image node]{\sqimg{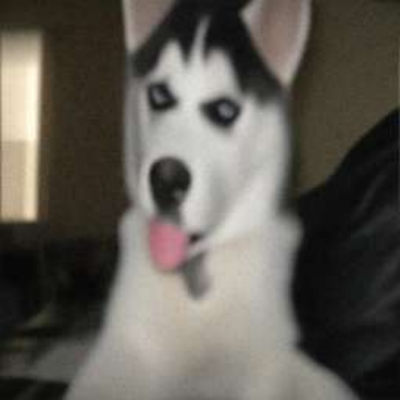} };
	\zoombox{0.54, 0.74}
\end{tikzpicture} &
\begin{tikzpicture}[zoomboxarray]
    \node [image node]{\sqimg{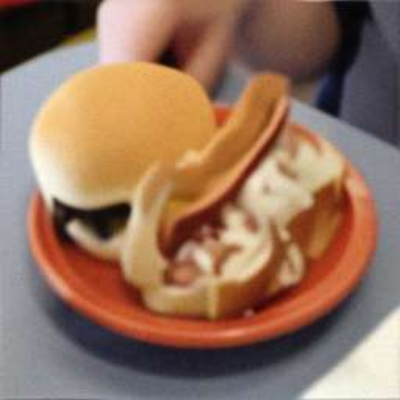}};
	\zoombox{0.58, 0.55}
\end{tikzpicture} &
\begin{tikzpicture}[zoomboxarray]
    \node [image node] {\sqimg{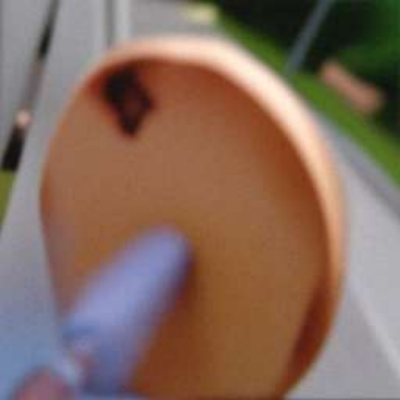} };
	\zoombox{0.3, 0.75}
\end{tikzpicture} &
\begin{tikzpicture}[zoomboxarray]
    \node [image node] {\sqimg{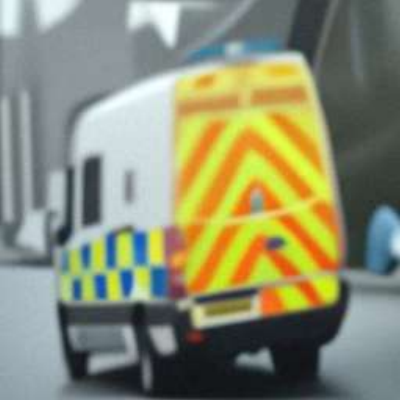} };
	\zoombox{0.52, 0.75}
\end{tikzpicture}\\[-2pt]

\rotlabel{BSD-RPnP-ULA} &
\begin{tikzpicture}[zoomboxarray]
    \node [image node]{\sqimg{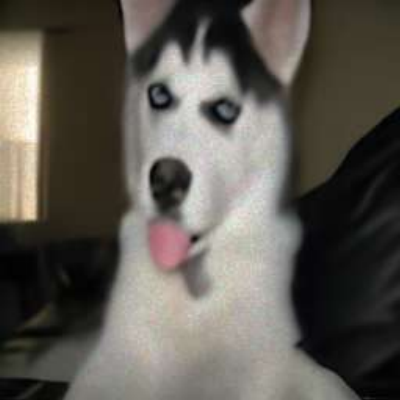}};
	\zoombox{0.54, 0.74}
\end{tikzpicture} &
\begin{tikzpicture}[zoomboxarray]
    \node [image node]{\sqimg{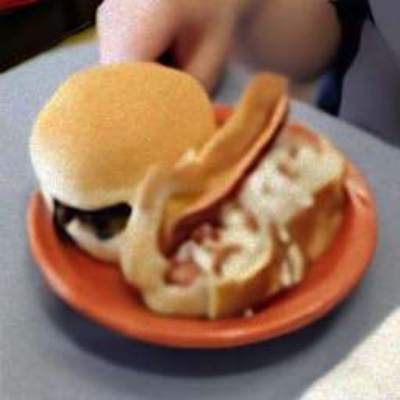}};
	\zoombox{0.58, 0.55}
\end{tikzpicture} &
\begin{tikzpicture}[zoomboxarray]
    \node [image node] {\sqimg{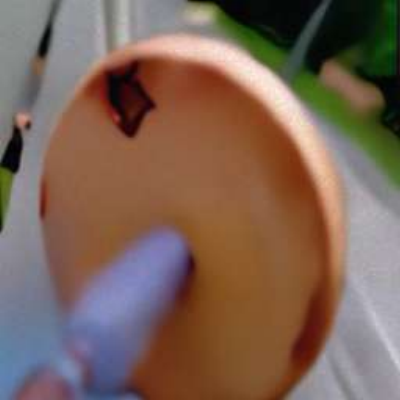}};
	\zoombox{0.3, 0.75}
\end{tikzpicture} &
\begin{tikzpicture}[zoomboxarray]
    \node [image node] {\sqimg{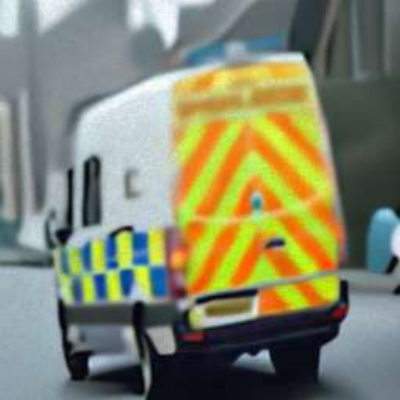} };
	\zoombox{0.52, 0.75}
\end{tikzpicture}\\[-2pt]

\rotlabel{\tiny{RED-HRLwSGS}} &
\begin{tikzpicture}[zoomboxarray]
    \node [image node]{\sqimg{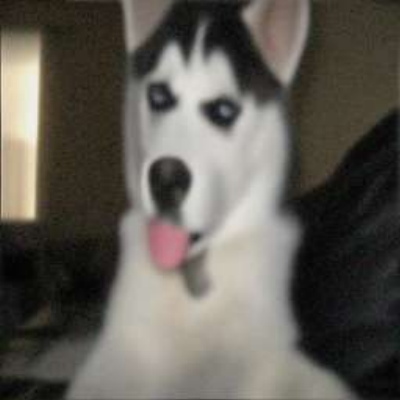}};
    \zoombox{0.54, 0.74}
\end{tikzpicture} &
\begin{tikzpicture}[zoomboxarray]
    \node [image node]{\sqimg{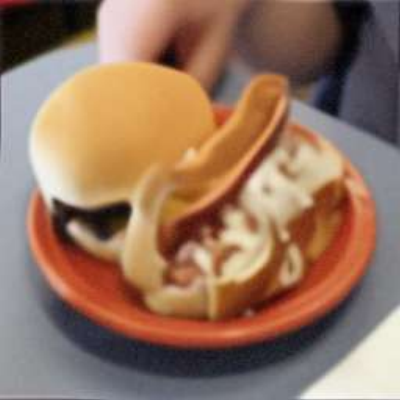} };
	\zoombox{0.58, 0.55}
\end{tikzpicture} &
\begin{tikzpicture}[zoomboxarray]
    \node [image node] {\sqimg{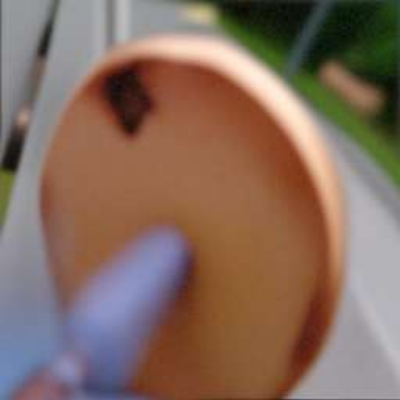} };
	\zoombox{0.3, 0.75}
\end{tikzpicture} &
\begin{tikzpicture}[zoomboxarray]
    \node [image node] {\sqimg{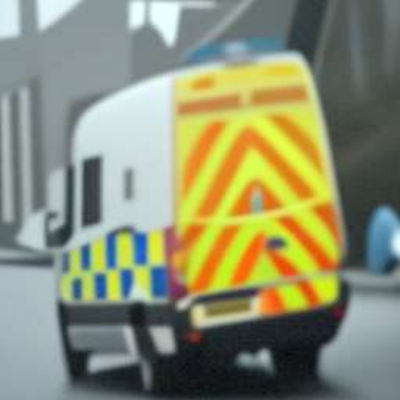}};
	\zoombox{0.52, 0.75}
\end{tikzpicture}\\[-2pt]

\rotlabel{\tiny{BSD-HRLwSGS}} &
\begin{tikzpicture}[zoomboxarray]
    \node [image node]{\sqimg{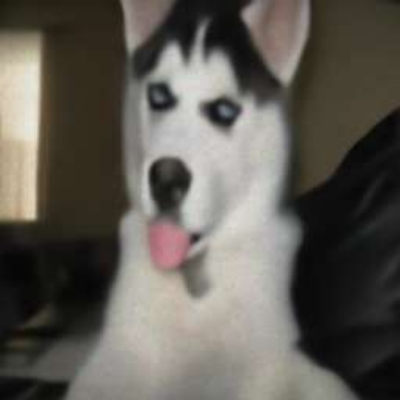}};
	\zoombox{0.54, 0.74}
\end{tikzpicture} &
\begin{tikzpicture}[zoomboxarray]
    \node [image node]{\sqimg{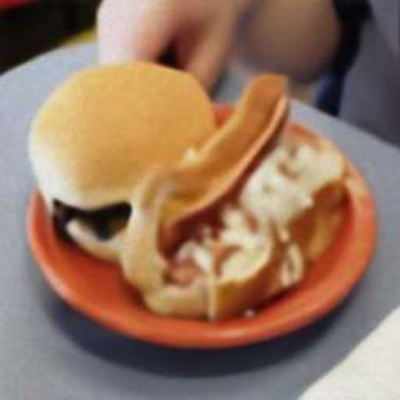}};
	\zoombox{0.58, 0.55}
\end{tikzpicture} &
\begin{tikzpicture}[zoomboxarray]
    \node [image node] {\sqimg{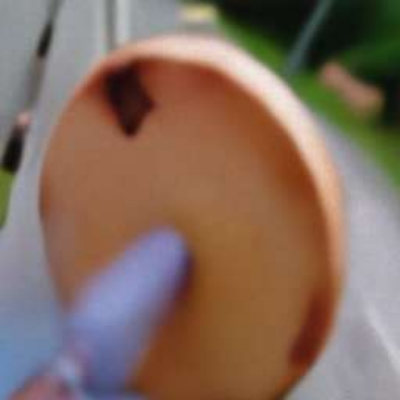} };
	\zoombox{0.3, 0.75}
\end{tikzpicture} &
\begin{tikzpicture}[zoomboxarray]
    \node [image node] {\sqimg{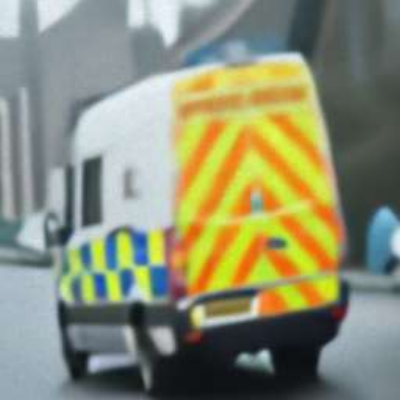} };
	\zoombox{0.52, 0.75}
\end{tikzpicture}\\[-2pt]

\end{tabular}
\caption{Deblurring experiment: visual results. Due to space constraints, the results recovered by TV-SPA and TV-PIDAL are not reproduced, as their quality is significantly inferior. They are reported in \cref{app:complementary_results} for comprehensiveness.}
\label{fig:imagenet_deblurring}
\end{figure}

{Beyond point estimation, the proposed Bayesian framework enables posterior uncertainty quantification, as illustrated in \Cref{fig:uncertainty_coverage_map}. For $\alpha = 10$, we report both the pixelwise posterior standard deviation and the coverage map, where each pixel value corresponds to the minimum probability ensuring that the highest posterior density (HPD) interval contains the true intensity. Uncertainty maps reveal localized uncertainty concentrated along textured and edge regions for TV-SPA, PnP-MLA, RED-HRLwSGS and RED-RPnP-ULA, which produce spatially coherent and interpretable uncertainty estimates, consistent with the expected behavior for deconvolution problems. In contrast, the BSD-based methods exhibit substantially lower posterior standard deviations, appearing nearly uniform and close to zero across the image. The coverage maps corroborate this discrepancy: for TV-SPA, PnP-MLA, RED-HRLwSGS and RED-RPnP-ULA, the probability that the HPD intervals contain the true intensity remains high (close to 1) across most of the image, with lower values mainly along edges.
}

\newcommand{\imgwidth}{0.136\textwidth}
\newcommand{\barwidth}{0.11\textwidth}

\begin{figure}[htbp]
    \centering
    \setlength{\tabcolsep}{1pt}
    \renewcommand{\arraystretch}{0}

    \begin{tabular}{c c c c c c c}

        &
        \tiny TV-SPA &
        \tiny PnP-MLA &
        \tiny RED-HRLwSGS &
        \tiny BSD-HRLwSGS &
        \tiny RED-RPnP-ULA &
        \tiny BSD-RPnP-ULA
        \\[1mm]

        \includegraphics[
            width=0.12\textwidth
        ]{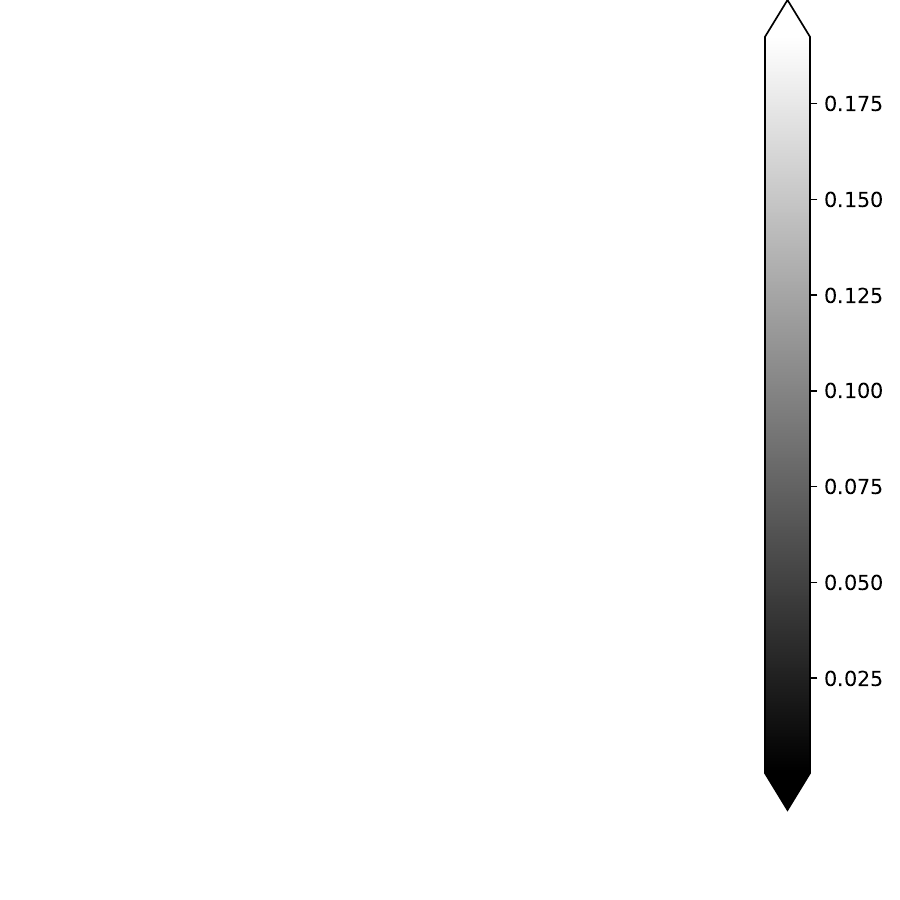}
        &
        \begin{tikzpicture}[zoomboxarray]
            \node[image node]{
                \includegraphics[
                    width=\imgwidth
                ]{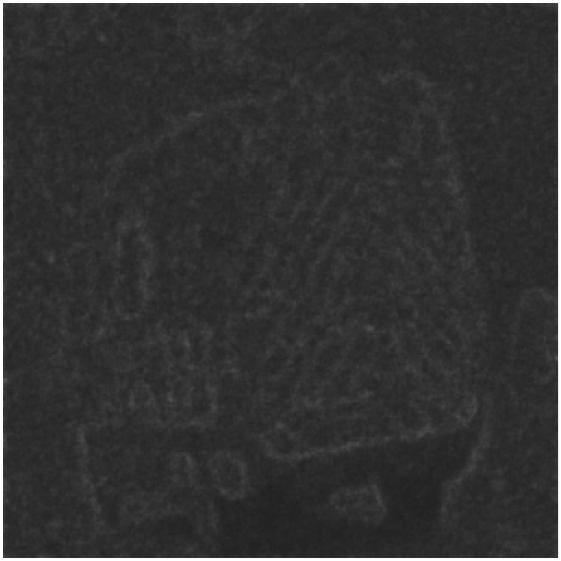}
            };
        \end{tikzpicture}
        &
        \begin{tikzpicture}[zoomboxarray]
            \node[image node]{
                \includegraphics[
                    width=\imgwidth
                ]{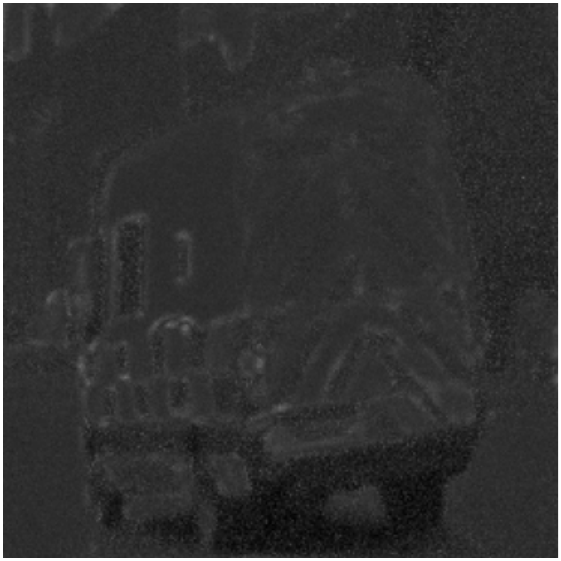}
            };
        \end{tikzpicture}
        &
        \begin{tikzpicture}[zoomboxarray]
            \node[image node]{
                \includegraphics[
                    width=\imgwidth
                ]{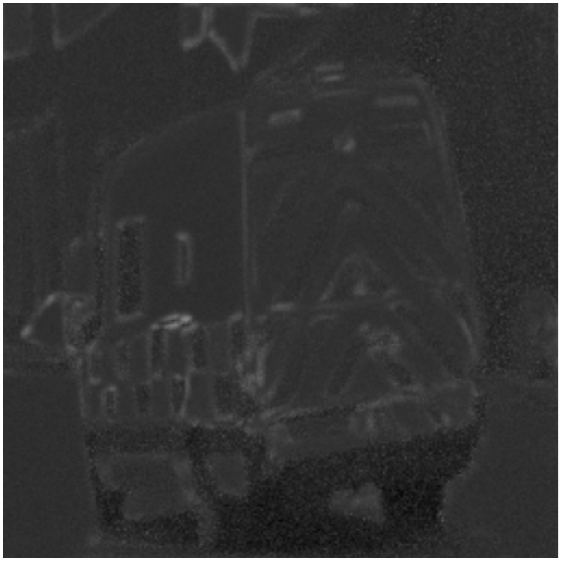}
            };
        \end{tikzpicture}
        &
        \begin{tikzpicture}[zoomboxarray]
            \node[image node]{
                \includegraphics[
                    width=\imgwidth
                ]{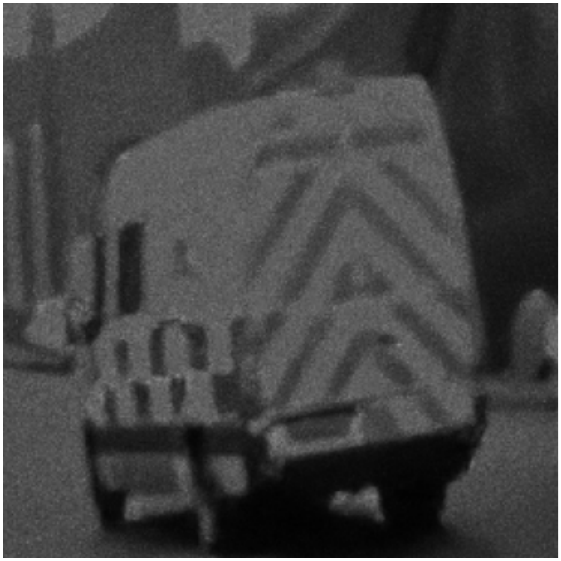}
            };
        \end{tikzpicture}
        &
        \begin{tikzpicture}[zoomboxarray]
            \node[image node]{
                \includegraphics[
                    width=\imgwidth
                ]{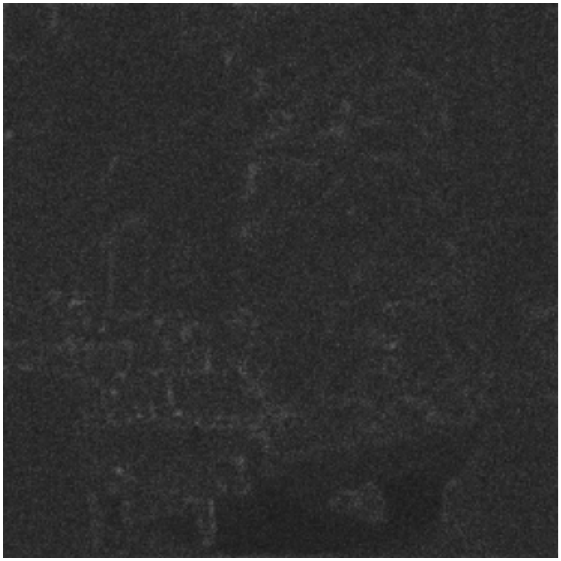}
            };
        \end{tikzpicture}
        &
        \begin{tikzpicture}[zoomboxarray]
            \node[image node]{
                \includegraphics[
                    width=\imgwidth
                ]{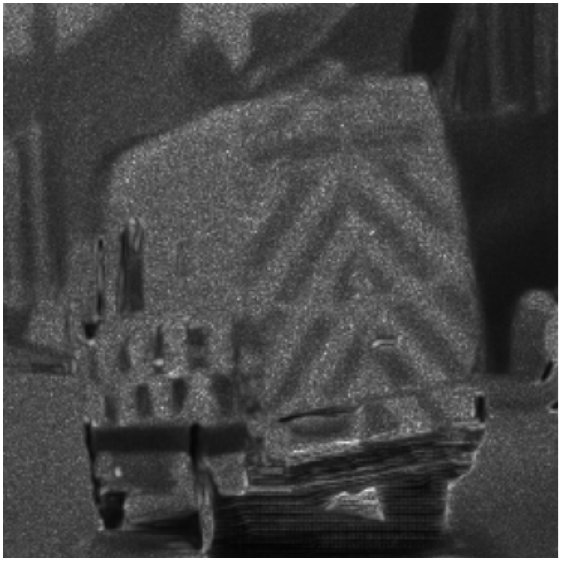}
            };
        \end{tikzpicture}
        \\[1mm]

        \includegraphics[
            width=\barwidth
        ]{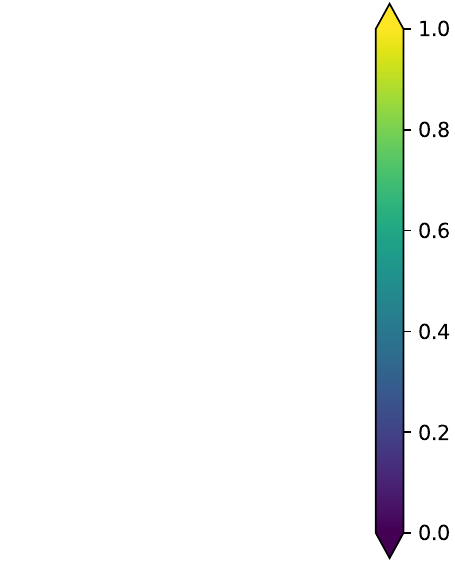}
        &
        \begin{tikzpicture}[zoomboxarray]
            \node[image node]{
                \includegraphics[
                    width=\imgwidth
                ]{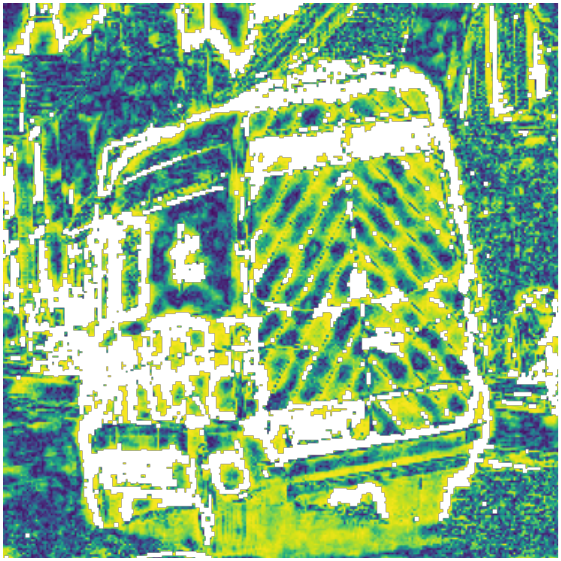}
            };
        \end{tikzpicture}
        &
        \begin{tikzpicture}[zoomboxarray]
            \node[image node]{
                \includegraphics[
                    width=\imgwidth
                ]{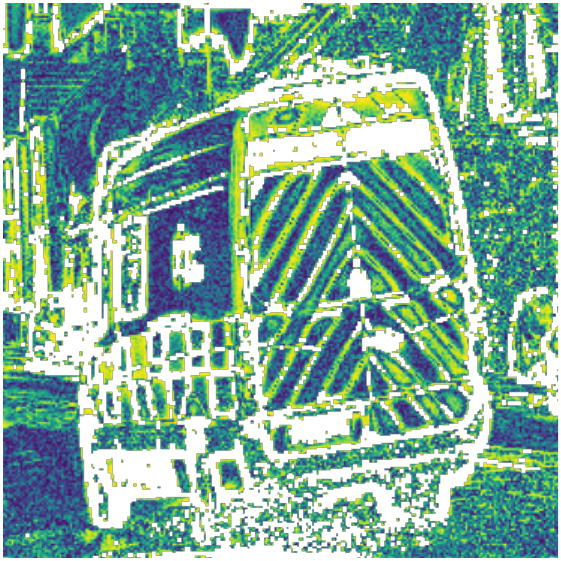}
            };
        \end{tikzpicture}
        &
        \begin{tikzpicture}[zoomboxarray]
            \node[image node]{
                \includegraphics[
                    width=\imgwidth
                ]{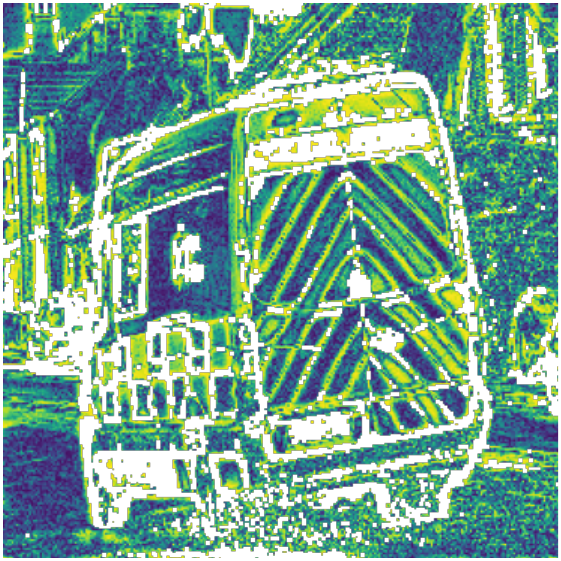}
            };
        \end{tikzpicture}
        &
        \begin{tikzpicture}[zoomboxarray]
            \node[image node]{
                \includegraphics[
                    width=\imgwidth
                ]{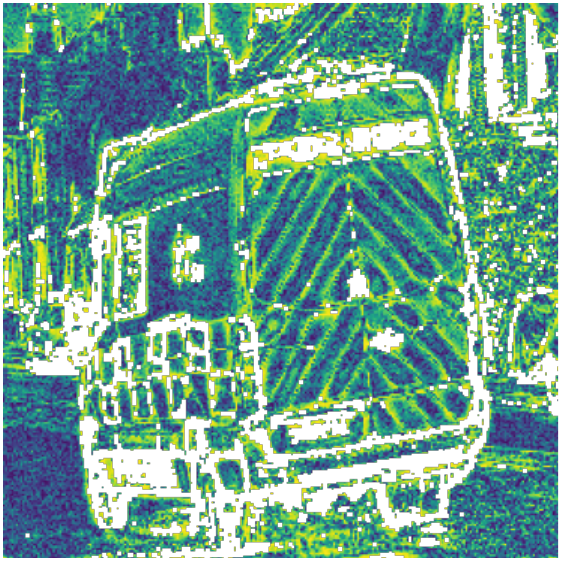}
            };
        \end{tikzpicture}
        &
        \begin{tikzpicture}[zoomboxarray]
            \node[image node]{
                \includegraphics[
                    width=\imgwidth
                ]{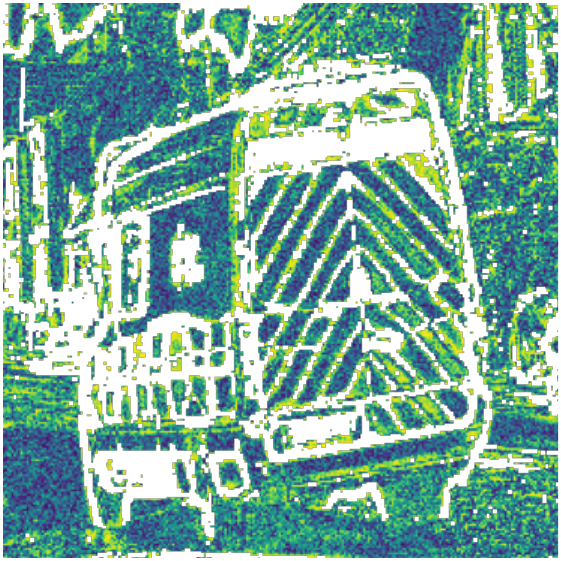}
            };
        \end{tikzpicture}
        &
        \begin{tikzpicture}[zoomboxarray]
            \node[image node]{
                \includegraphics[
                    width=\imgwidth
                ]{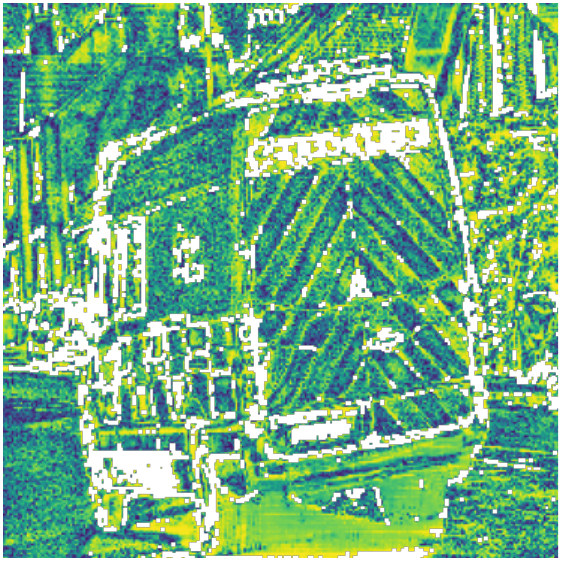}
            };
        \end{tikzpicture}

    \end{tabular}

    \caption{
        Deblurring experiment ($\alpha=10$).
        Posterior standard deviation maps (top) and
        corresponding HPD coverage probability maps (bottom).
    }
    \label{fig:uncertainty_coverage_map}
\end{figure}

\newcommand{\subplotwidth}{0.136\textwidth}
\newcommand{\subplotwidthbar}{0.10\textwidth}
\newcommand{\subplotwidthbis}{.0\textwidth}
\newcommand{\subplotwidthbisbis}{.04\textwidth}

\subsection{Positron emission tomographic reconstruction}\label{subsec:tomography}
We evaluate the performance of the proposed method when performing a PET reconstruction task using a synthetic phantom of size $128 \times 128$. 
The forward operator $\bH$ simulates a parallel-beam acquisition model with 512 projection angles uniformly distributed in $[0, \pi)$. The resulting system matrix $\bH \in \mathbb{R}^{93184 \times 16384}$ is constructed using the ASTRA toolbox \cite{van2016fast, van2015astra} and defines a severely ill-posed inverse problem due to the sparsity of the measurements. This configuration emulates a realistic tomographic setting and poses a significant computational challenge due to the high dimensionality and ill-posedness of the problem. The ground truth image and the measurements are depicted in \cref{fig:pet_}

\renewcommand{\subplotwidth}{.2\textwidth}
\begin{figure}[h!]
\centering
\begin{tikzpicture}[zoomboxarray]
\node [image node] { \includegraphics[width=\subplotwidth, height=\subplotwidth]{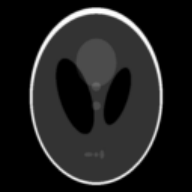} };
   \zoombox{0.5, 0.5}
\end{tikzpicture}
\begin{tikzpicture}[zoomboxarray]
    \node [image node] { \includegraphics[width=\subplotwidth, height=\subplotwidth]{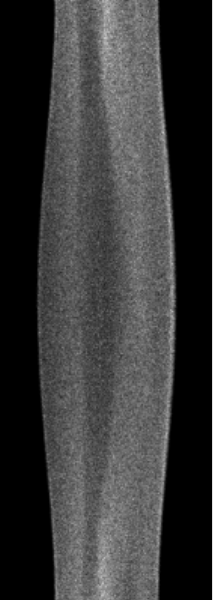} };
\end{tikzpicture}
\caption{PET reconstruction experiment: ground truth image (left) and measurements (right).}
\label{fig:pet_}
\end{figure}


{In addition to optimization-based baseline algorithms TV-PIDAL and TV-SPA, we also consider here the recent deep equilibrium using mirror descent  (DEQ-MD) reconstruction method and dedicated to Poisson inverse problems \cite{daniele2025deep}. In contrast with PnP/RED, DEQ does not rely on any pretrained denoiser but instead implicitly learns the operator for the specific reconstruction task. For our comparison, the DEQ-MD potential is based on a DRUNet with ELU activation. For training, we used the \texttt{Drunet} dataset\footnote{Available online at \url{https://huggingface.co/deepinv/drunet/}.} provided by the DeepInverse package \cite{Tachella2025} augmented with additional 2D slices from the \texttt{BrainWeb} dataset\footnote{Available online at \url{https://brainweb.bic.mni.mcgill.ca/brainweb/}.}.}

Quantitative results are reported in \cref{tab:pet}. {Among all the evaluated methods, DEQ-MD achieves the highest PSNR (31.94~dB) and the second-highest SSIM (0.925), while obtaining a low LPIPS value (0.115). TV-PIDAL attains a slightly higher SSIM (0.938), the highest among all evaluated methods, but its PSNR is substantially lower (28.11~dB), with an LPIPS of 0.173. Thus, DEQ-MD provides a more favorable balance between pixel-wise fidelity and perceptual similarity, whereas TV-PIDAL exhibits particularly strong structural similarity despite its lower PSNR.}

{Among the sampling-based approaches, RED-RPnP-ULA achieves the highest PSNR (31.85~dB), followed by RED-HRLwSGS (30.41~dB) and PnP-MLA (29.99~dB). In terms of SSIM, BSD-RPnP-ULA performs best among the sampling-based methods (0.814), followed by BSD-HRLwSGS (0.787) and RED-RPnP-ULA (0.777). Regarding perceptual similarity, RED-RPnP-ULA achieves the lowest LPIPS among the sampling-based approaches (0.166), followed by RED-HRLwSGS (0.174) and BSD-HRLwSGS (0.191). These results indicate that RED-RPnP-ULA provides the strongest overall reconstruction performance within the sampling-based group, particularly in terms of PSNR and LPIPS.}

{DEQ-MD and TV-PIDAL, generally achieve higher SSIM values than the sampling-based approaches. DEQ-MD also provides the highest PSNR overall (31.94~dB), although RED-RPnP-ULA obtains a very close value (31.85~dB) while achieving a lower LPIPS than TV-PIDAL. The sampling-based methods therefore provide competitive reconstruction quality while additionally enabling posterior uncertainty quantification, which is not directly available from the deterministic optimization-based approaches considered here.}

{For uncertainty calibration, the sampling-based methods exhibit different levels of performance. TV-SPA achieves the lowest ECE (0.035), followed by BSD-HRLwSGS (0.046), BSD-RPnP-ULA (0.061), and RED-RPnP-ULA (0.078). RED-HRLwSGS and PnP-MLA obtain higher ECE values of 0.096 and 0.099, respectively. These results highlight that reconstruction quality and uncertainty calibration do not necessarily follow the same ranking: for example, TV-SPA achieves the best ECE among the evaluated sampling-based methods, despite exhibiting the lowest PSNR and SSIM and the highest LPIPS.}

\begin{table}[h!]
    \centering  
    \setlength{\tabcolsep}{6pt}
    \caption{PET reconstruction experiment: performance.}
    \label{tab:pet}
    \begin{tabular}{lccccc}
    \toprule
    Method & PSNR (dB)$\uparrow$ & SSIM $\uparrow$& LPIPS $\downarrow$ & ECE $\downarrow$ \\
    \midrule
    TV-PIDAL      & 28.1064 & 0.9375 & 0.1729 & -- \\
    {DEQ-MD}      & 31.9402 & \textbf{0.9252} & 0.1153 & -- \\    
    TV-SPA        & 26.1020 & 0.6228 & 0.2026 & 0.0354 \\
    {PnP-MLA}     & 29.9896 & 0.7006 & 0.2232 & 0.0989 \\    
    RED-RPnP-ULA  & 31.8546 & 0.7767 & 0.1663 & 0.0778 \\
    BSD-RPnP-ULA  & 28.0048 & 0.8136 & 0.2227 & 0.0614 \\
    RED-HRLwSGS   & 30.4099 & 0.7630 & 0.1744 & 0.0961 \\
    BSD-HRLwSGS   & 29.3216 & 0.7873 & 0.1911 & 0.0462 \\
    \bottomrule
    \end{tabular}
\end{table}

\Cref{fig:pet} provides a qualitative comparison of the reconstructed images (first row), the pixelwise posterior standard deviations (second row) and the coverage probability maps (third row). These visualizations highlight both reconstruction fidelity and the reliability of the uncertainty quantification across the compared samplers. For the proposed RED-HRLwSGS, the reconstructed image is sharp and visually similar to the ground truth, with well-preserved edges and smooth homogeneous regions.
The associated uncertainty map shows localized uncertainty mainly along structural contours, while homogeneous regions remain stable with very low variance.
Its coverage map is spatially consistent and concentrated near one in most areas. The BSD-HRLwSGS reconstruction exhibits slightly higher noise levels and minor artefacts near the edges compared to the RED-HRLwSGS reconstruction, as reflected by the larger posterior standard deviations. Nevertheless, the uncertainty remains structured and well localized. The coverage map corroborates these observations, showing moderately reduced confidence around the object boundaries but overall satisfactory reliability.

\setlength{\imw}{\dimexpr(\linewidth-5\tabcolsep)/6\relax}

\begin{figure}
\centering
\setlength{\tabcolsep}{0pt}
\renewcommand{\arraystretch}{1}

\begin{tabular}{@{}cccccc@{}}

\makebox[\imw][c]{\tiny DEQ-MD} &
\makebox[\imw][c]{\tiny PnP-MLA} &
\makebox[\imw][c]{\tiny RED-HRLwSGS} &
\makebox[\imw][c]{\tiny BSD-HRLwSGS} &
\makebox[\imw][c]{\tiny RED-RPnP-ULA} &
\makebox[\imw][c]{\tiny BSD-RPnP-ULA}
\\[2pt]


\begin{tikzpicture}[zoomboxarray]
\node [image node] {
    \includegraphics[width=\imw,height=\imw]{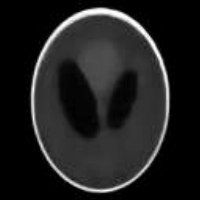}
};
\zoombox{0.5,0.5}
\end{tikzpicture}
&
\begin{tikzpicture}[zoomboxarray]
\node [image node] {
    \includegraphics[width=\imw, height=\imw]{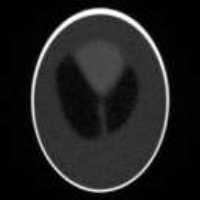}
};
\zoombox{0.5,0.5}
\end{tikzpicture}
&
\begin{tikzpicture}[zoomboxarray]
\node [image node] {
    \includegraphics[
        width=\imw,
        height=\imw
    ]{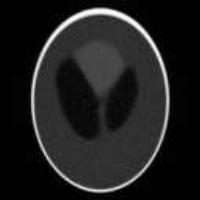}
};
\zoombox{0.5,0.5}
\end{tikzpicture}
&
\begin{tikzpicture}[zoomboxarray]
\node [image node] {
    \includegraphics[
        width=\imw,
        height=\imw
    ]{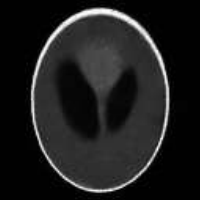}
};
\zoombox{0.5,0.5}
\end{tikzpicture}
&
\begin{tikzpicture}[zoomboxarray]
\node [image node] {
    \includegraphics[
        width=\imw,
        height=\imw
    ]{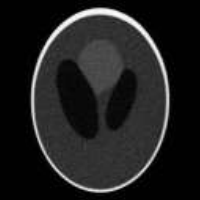}
};
\zoombox{0.5,0.5}
\end{tikzpicture}
&
\begin{tikzpicture}[zoomboxarray]
\node [image node] {
    \includegraphics[
        width=\imw,
        height=\imw
    ]{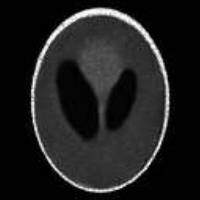}
};
\zoombox{0.5,0.5}
\end{tikzpicture}
\\


\begin{tikzpicture}[zoomboxarray]
\node [image node] {
    \includegraphics[
        width=\imw,
        height=\imw
    ]{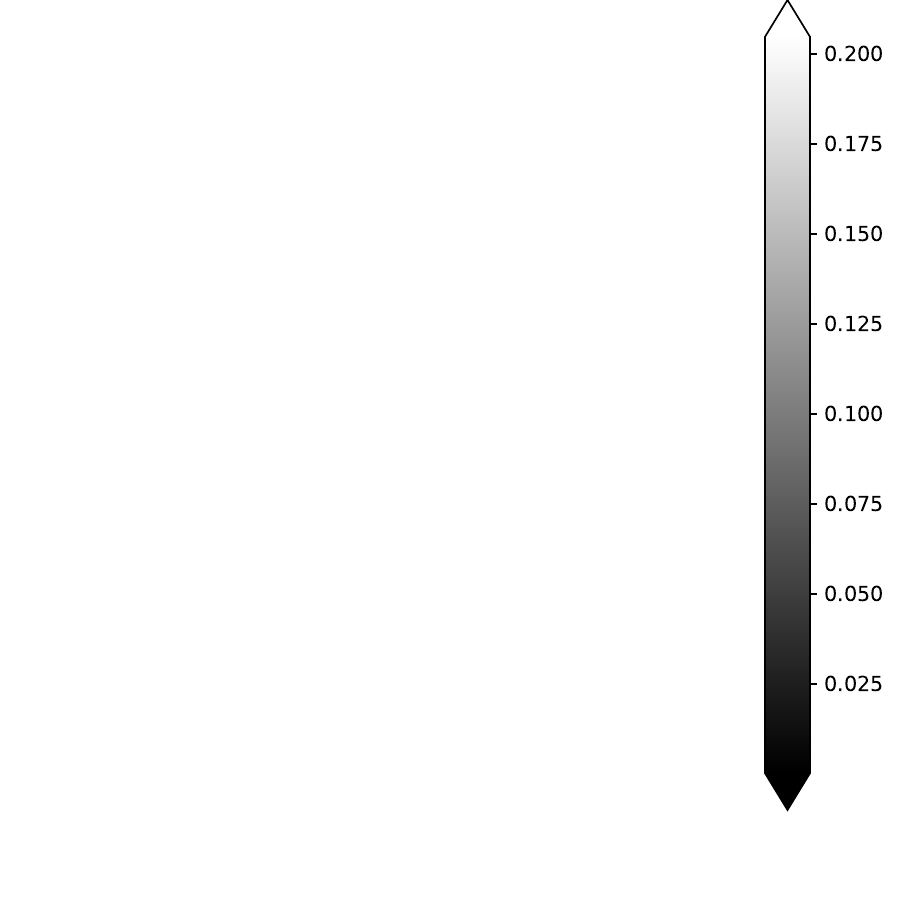}
};
\end{tikzpicture}
&
\begin{tikzpicture}[zoomboxarray]
\node [image node] {
    \includegraphics[
        width=\imw,
        height=\imw
    ]{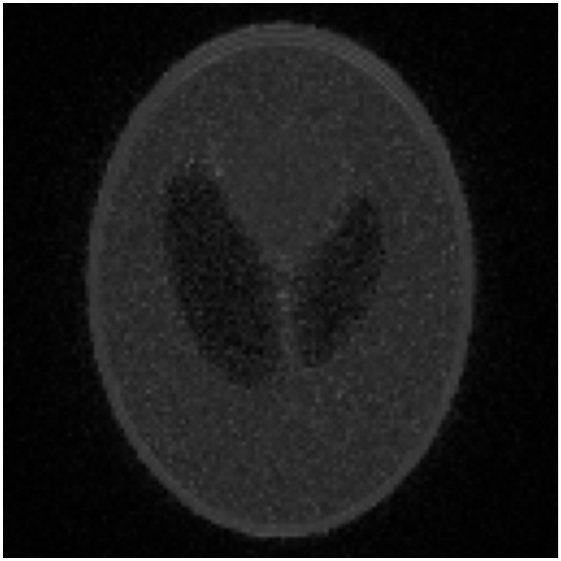}
};
\zoombox{0.5,0.5}
\end{tikzpicture}
&
\begin{tikzpicture}[zoomboxarray]
\node [image node] {
    \includegraphics[
        width=\imw,
        height=\imw
    ]{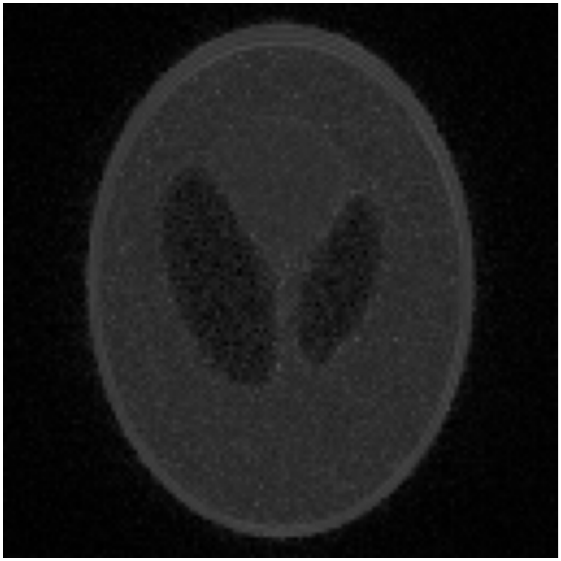}
};
\zoombox{0.5,0.5}
\end{tikzpicture}
&
\begin{tikzpicture}[zoomboxarray]
\node [image node] {
    \includegraphics[
        width=\imw,
        height=\imw
    ]{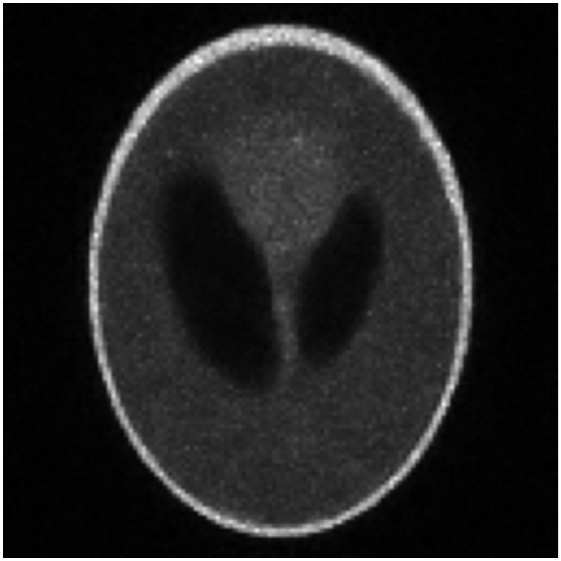}
};
\zoombox{0.5,0.5}
\end{tikzpicture}
&
\begin{tikzpicture}[zoomboxarray]
\node [image node] {
    \includegraphics[
        width=\imw,
        height=\imw
    ]{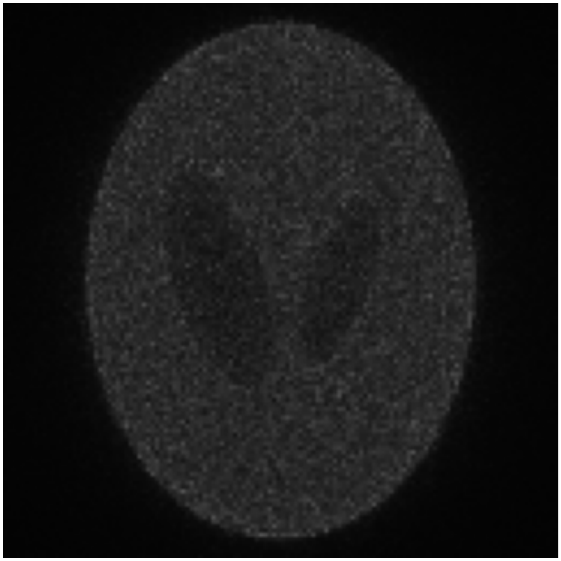}
};
\zoombox{0.5,0.5}
\end{tikzpicture}
&
\begin{tikzpicture}[zoomboxarray]
\node [image node] {
    \includegraphics[
        width=\imw,
        height=\imw
    ]{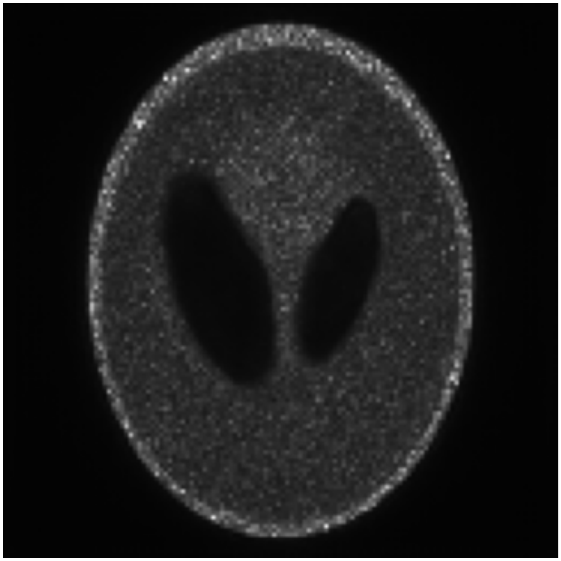}
};
\zoombox{0.5,0.5}
\end{tikzpicture}
\\


\begin{tikzpicture}[zoomboxarray]
\node [image node] {
    \includegraphics[
        width=\imw,
        height=\imw
    ]{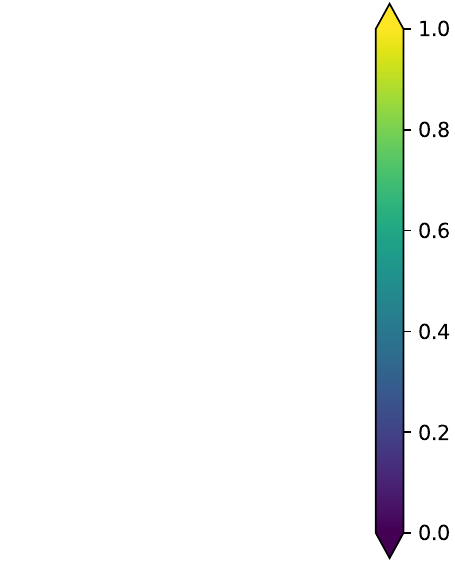}
};
\end{tikzpicture}
&
\begin{tikzpicture}[zoomboxarray]
\node [image node] {
    \includegraphics[
        width=\imw,
        height=\imw
    ]{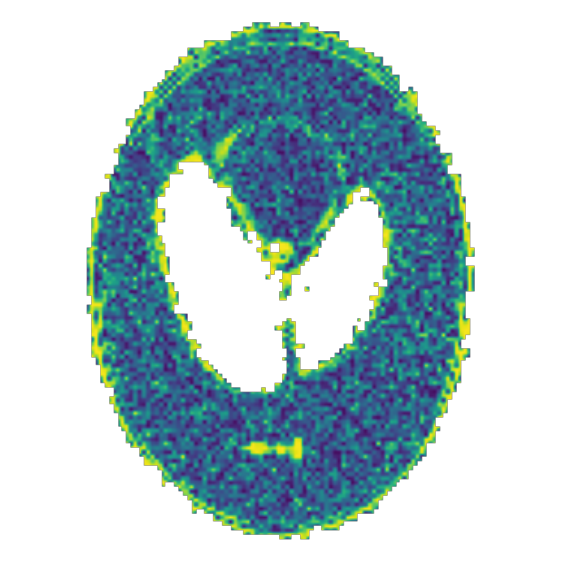}
};
\zoombox{0.5,0.5}
\end{tikzpicture}
&
\begin{tikzpicture}[zoomboxarray]
\node [image node] {
    \includegraphics[
        width=\imw,
        height=\imw
    ]{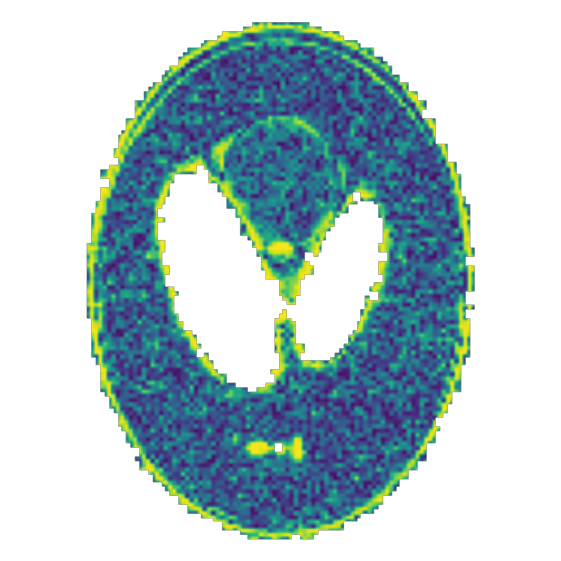}
};
\zoombox{0.5,0.5}
\end{tikzpicture}
&
\begin{tikzpicture}[zoomboxarray]
\node [image node] {
    \includegraphics[
        width=\imw,
        height=\imw
    ]{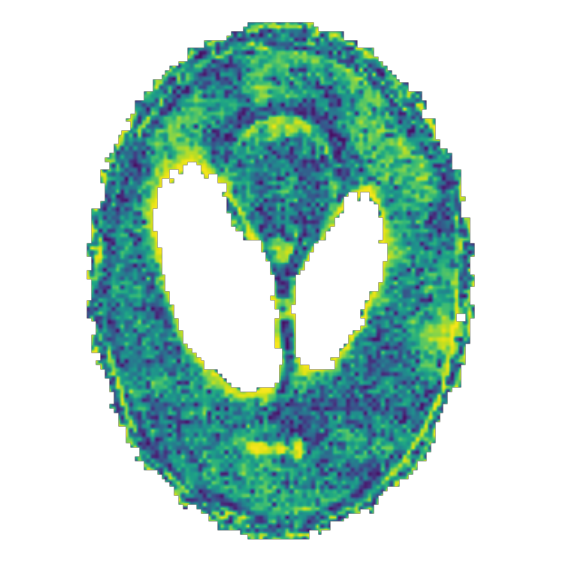}
};
\zoombox{0.5,0.5}
\end{tikzpicture}
&
\begin{tikzpicture}[zoomboxarray]
\node [image node] {
    \includegraphics[
        width=\imw,
        height=\imw
    ]{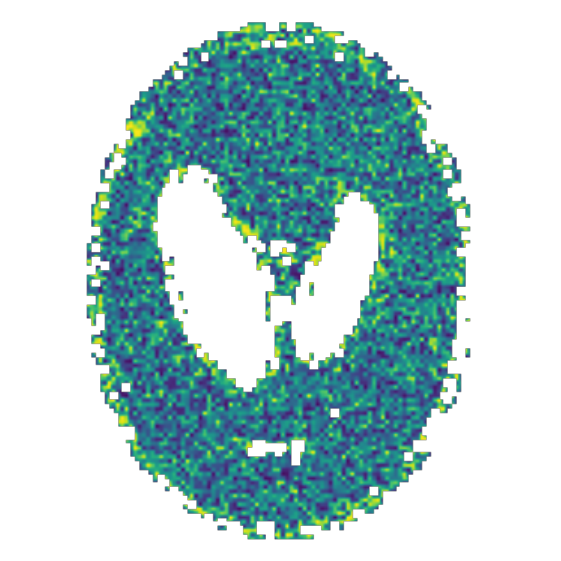}
};
\zoombox{0.5,0.5}
\end{tikzpicture}
&
\begin{tikzpicture}[zoomboxarray]
\node [image node] {
    \includegraphics[
        width=\imw,
        height=\imw
    ]{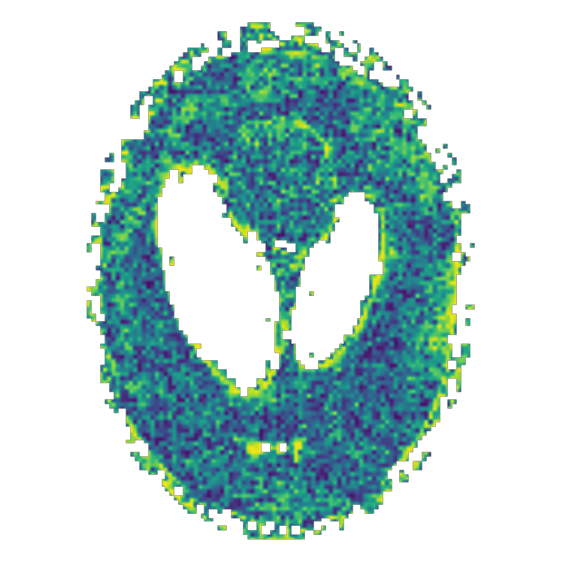}
};
\zoombox{0.5,0.5}
\end{tikzpicture}

\end{tabular}

\caption{PET reconstruction experiment: reconstructed images (top), standard deviations (middle), and coverage probability maps (bottom).}
\label{fig:pet}
\end{figure}

\Cref{fig:tep_afc} shows the performance of the sampling-based algorithms through autocorrelation functions (ACF) and calibration results. The ACF is a key indicator of MCMC sampling efficiency: a rapid decay towards zero indicates fast-mixing chains, low correlation between successive samples, and effective exploration of the posterior distribution. In the fastest scenario, TV-SPA and BSD-RPnP-ULA methods exhibit an almost instantaneous drop in the ACF, indicating weak correlation and superior mixing, {while RED-RPnP-ULA shows an intermediate decay, reaching a low correlation level only after a moderate number of lags}. By contrast, HRLwSGS-based
approaches (namely BSD-HRLwSGS and RED-HRLwSGS) and PnP-MLA show slower decay, suggesting stronger temporal dependencies. In the median case, this tendency partially reverses. {TV-SPA now achieves the fastest ACF decay}, closely followed by BSD-RPnP-ULA, whose ACF steadily decreases after a moderate number of iterations. {The remaining methods (RED-HRLwSGS, BSD-HRLwSGS, RED-RPnP-ULA, and PnP-MLA) struggle to reach decorrelation}, revealing slower convergence and poorer mixing behavior. Finally, in the slowest case, the difference becomes particularly striking. {BSD-HRLwSGS succeeds in maintaining a comparatively faster and more stable ACF decay}, demonstrating a relatively more efficient exploration of the posterior distribution, whereas {all the other algorithms (TV-SPA, RED-HRLwSGS, RED-RPnP-ULA, BSD-RPnP-ULA and PnP-MLA)} exhibit persistent correlations, which are symptomatic of limited sampling efficiency.

Calibration measures how accurately predicted uncertainties reflect true variability. This is evaluated by comparing a target coverage level with the actually achieved coverage, i.e., the fraction of pixels for which the true value lies within the posterior credible interval {(see \cref{app:HPD} for details)}. A perfectly calibrated model will follow the diagonal line, curves below this line indicate overconfidence, while curves above it indicate underconfidence. {In medical imaging, especially in safety-critical settings, well-calibrated and slightly conservative uncertainty estimates are generally preferred over overconfident predictions \cite{kompa2021second, lopez2025uncert}. This preference stems from the significant implications for clinical decision-making and patient safety: while overconfidence can lead to harmful diagnostic errors, conservative behavior typically results in additional review or second-opinion requests, which, though potentially incurring extra cost or time, are far less risky. Thus, automated medical systems tend to favor conservative uncertainty estimates and abstention mechanisms to mitigate the direct harm that overconfident errors could cause to patient care \cite{jaspan2022improving}. The results show that TV-SPA remains the closest to the diagonal, indicating the best calibration overall. The proposed HRLwSGS algorithm performs similarly as the other compared methods, consistently lying  above the diagonal across most of the coverage range, reflecting a more pronounced under-coverage.} 

\renewcommand{\subplotwidth}{.45\textwidth}
\newcommand{\subplotheight}{.35\textwidth}
\begin{figure}[htbp]
\centering
    \includegraphics[width=\subplotwidth]{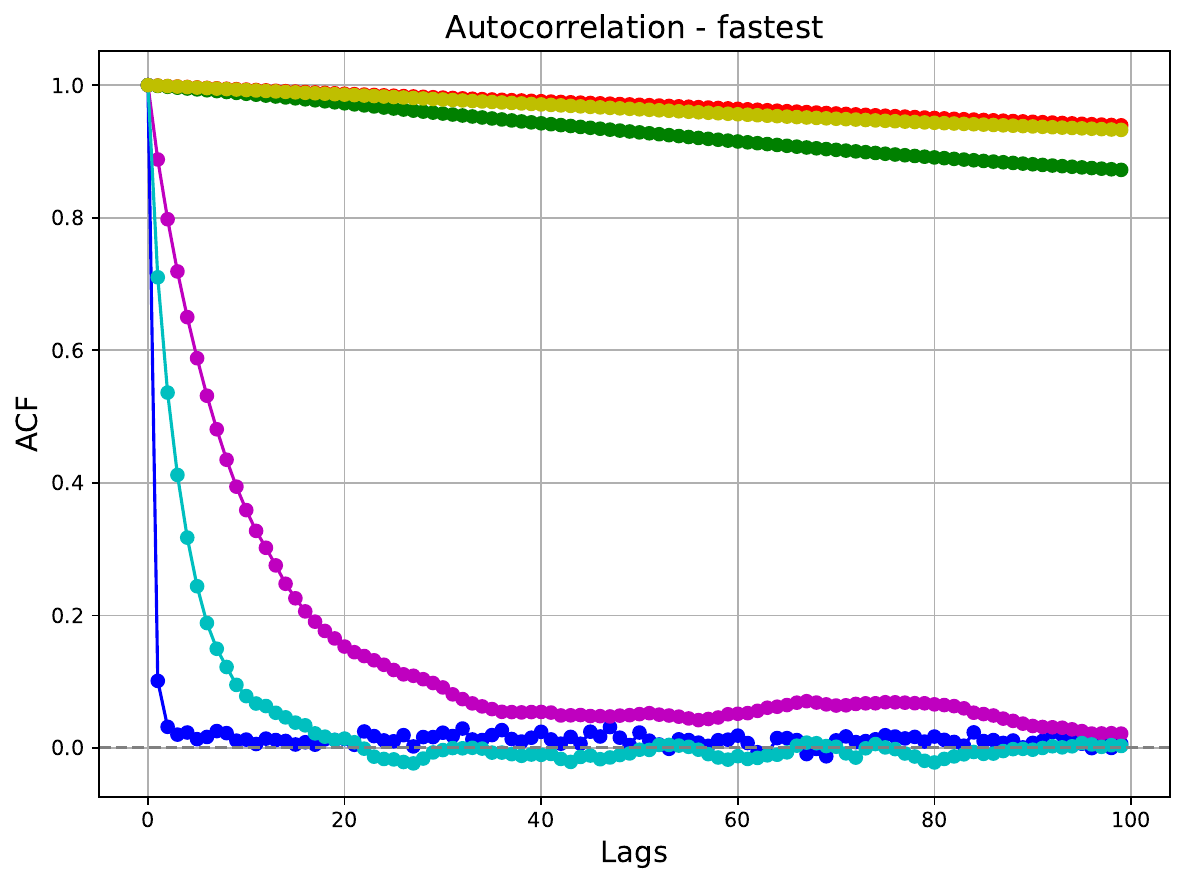}%
    \includegraphics[width=\subplotwidth]{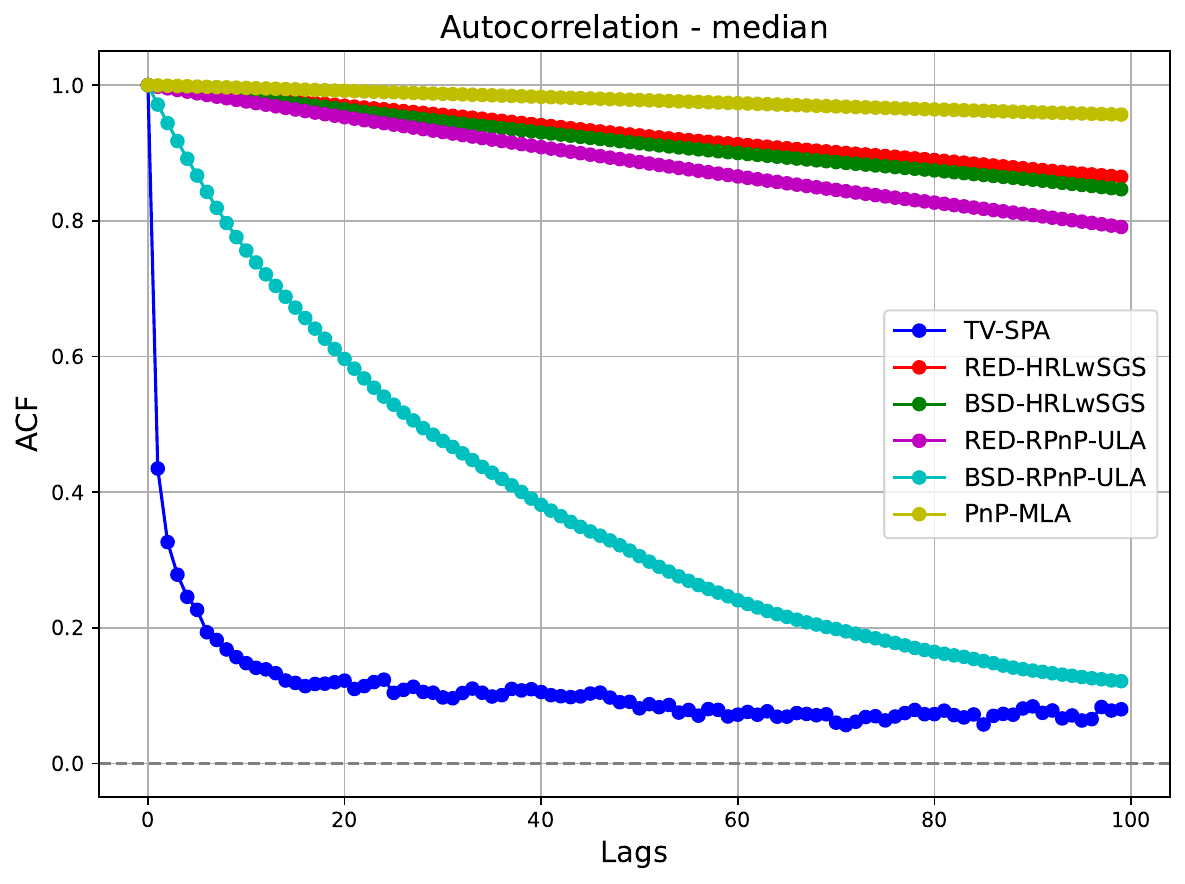}\\
    \includegraphics[width=\subplotwidth]{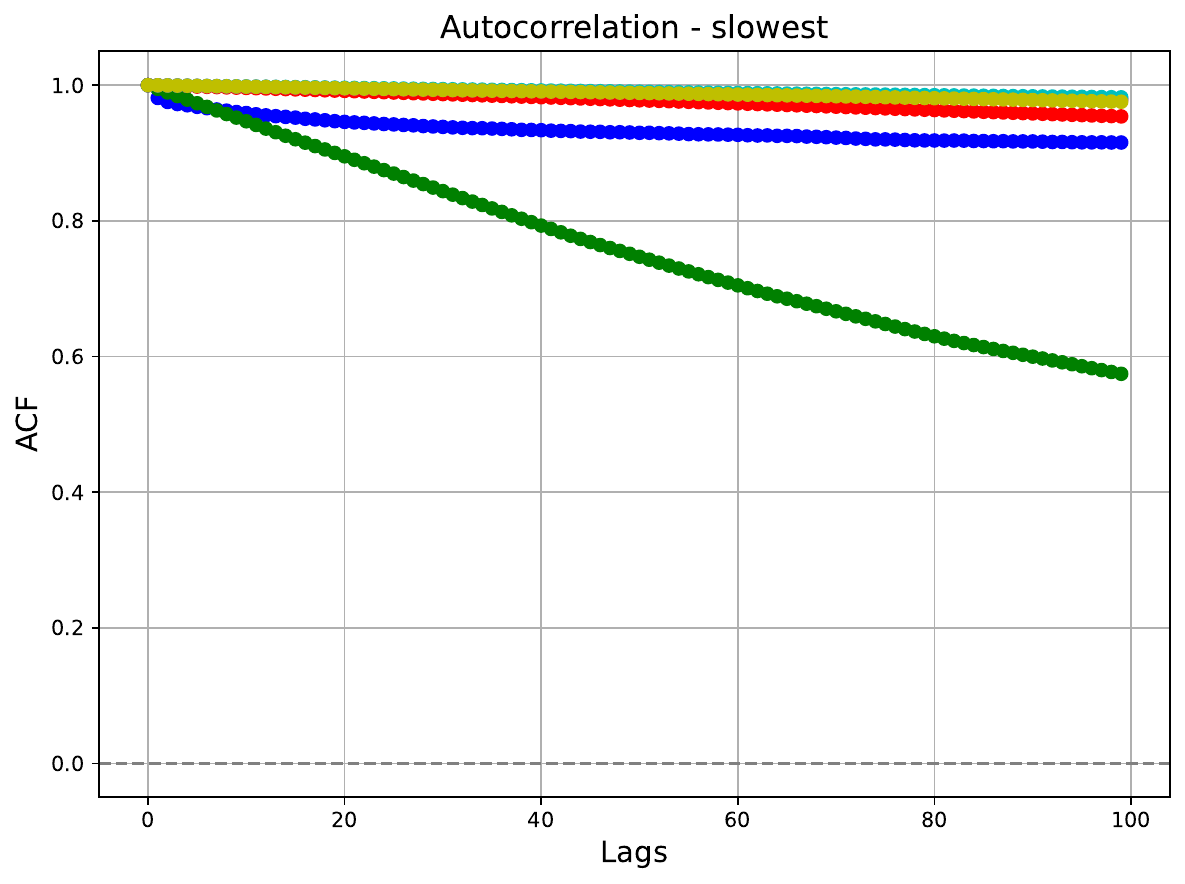}%
    \includegraphics[width=\subplotwidth]{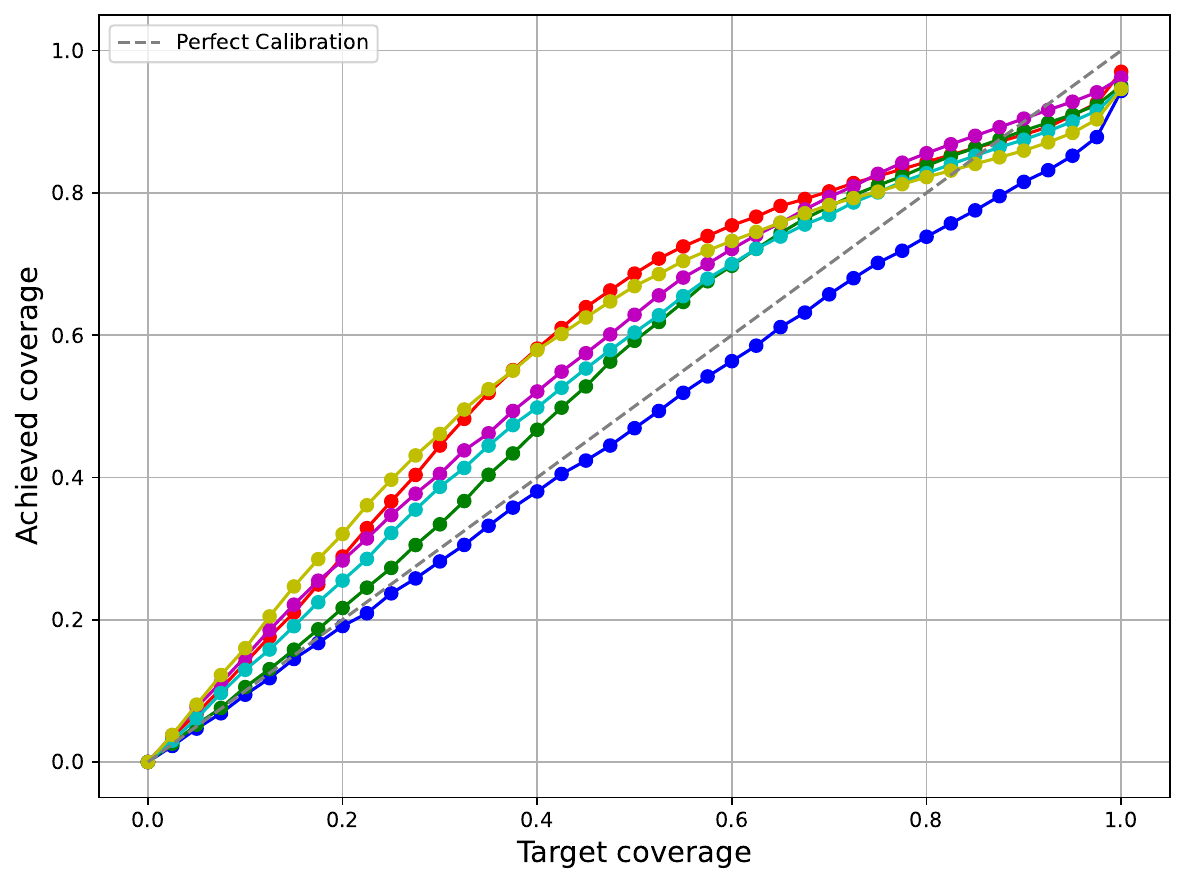}
\caption{PET reconstruction experiment: autocorrelation and calibration curves.}
\label{fig:tep_afc}
\end{figure}

\subsection{Discussion on the computational complexity}

{\Cref{tab:time} reports the computational times required by the different methods for solving the three considered inversion tasks. The results show that the computational cost of HRLwSGS increases substantially with the complexity of the forward operator. For denoising, HRLwSGS requires 732~s and 199~s for the RED and BSD variants, respectively, which are of the same order of magnitude as the runtimes of RPnP-ULA and PnP-MLA. In particular, BSD-HRLwSGS is faster than BSD-RPnP-ULA, while RED-HRLwSGS is moderately slower than RED-RPnP-ULA. Hence, the computational overhead of HRLwSGS depends not only on the sampling procedure but also on the computational cost associated with the chosen prior.}

{For deblurring, the computational times of HRLwSGS increase to 1381~s and 1488~s for the RED and BSD variants, respectively, compared with 669--751~s for RPnP-ULA and 675~s for PnP-MLA. This increase is related to the latent-variable sampling step of HRLwSGS. More specifically, the count variable $\bc$ is sampled according to the multinomial distributions in \eqref{eq:sample_n}, whose parameters depend on the current image estimate and on the forward operator $\bH$. In the denoising setting, $\bH$ is the identity operator, so that these multinomial parameters can be computed directly from the current pixel values. In contrast, for deblurring, their computation requires evaluating the contributions induced by the non-trivial forward operator, resulting in an additional computational cost.}

{This effect is particularly pronounced for the PET reconstruction problem. In our implementation, the PET forward operator $\bH$ is constructed as a sparse matrix using the ASTRA Toolbox. Thus, HRLwSGS does not explicitly manipulate a dense representation of $\bH$. Nevertheless, the multinomial sampling step requires accessing the non-zero entries of $\bH$ to construct the corresponding multinomial parameters. Because of the size and structure of the tomographic operator, this operation remains computationally expensive, even when $\bH$ is stored in sparse format. This results in runtimes of 6562~s and 6835~s for RED-HRLwSGS and BSD-HRLwSGS, respectively, compared with 1330~s and 1471~s for RPnP-ULA and 1320~s for PnP-MLA. The substantially larger runtime of HRLwSGS in this setting therefore mainly reflects the additional computational burden associated with the latent-variable sampling under a large-scale tomographic forward model.}

{It is also important to note that the reported runtimes of HRLwSGS are obtained with a sequential implementation of the Gibbs sampler. Conditional on the current state, the latent vectors $\bc_{i:}$, $i=1,\ldots,m$, can be sampled independently. The current implementation does not exploit this conditional independence through multi-core or distributed computation. Parallelizing this sampling step could therefore reduce the wall-clock time, particularly for the deblurring and PET problems. The actual speed-up, however, would depend on the hardware architecture and on the implementation of the sparse operations.}

{Finally, the comparison with PIDAL, and ProxDiffPIR should be interpreted in light of the different computational objectives of the considered methods. These optimization-based approaches generally exhibit substantially lower runtimes, but they are designed to compute a point estimate and do not provide posterior samples. In contrast, HRLwSGS, RPnP-ULA, and PnP-MLA perform posterior exploration and therefore require repeated stochastic updates. Their larger computational cost consequently reflects, at least in part, the additional objective of obtaining posterior information. Overall, the results indicate that HRLwSGS remains computationally feasible for all considered inversion tasks, while its computational overhead becomes more significant for large-scale forward models such as the PET operator.}

\begin{table}[h!]\small
\setlength{\tabcolsep}{3.5pt}
\centering
\caption{Computational time comparison.\label{tab:time}}
\begin{tabular}{lcccccccc} 
\toprule
&\multicolumn{2}{c}{{\textbf{HRLwSGS}}} & \multicolumn{2}{c}{{\textbf{RPnP-ULA}}} & {\textbf{SPA}} & {\textbf{PIDAL}} & {\textbf{ProxDiffPIR}} & {\textbf{PnP-MLA}} \\
\midrule
    {Regularization}     & RED & BSD  & RED & BSD  & TV & TV & DM & PnP \\ 
    \midrule
    {Denoising}  & 732  & 199   & 631 &  753 & 753 & 6 & 7 & 637 \\
    {Deblurring} & 1381 & 1488 & 669 & 751 & 1093 & 34 & 29 & 675 \\
    {PET}        & 6562 & 6835 & 1330 & 1471 & 2770 & 16 & - & 1320 \\
\bottomrule
\end{tabular}
\end{table}

\subsection{Discussion on the performance of the denoiser-based priors}
{From the experimental results reported above, it appears that the two considered denoiser-based priors, namely BSD and RED, lead to different inversion performance depending on the task and the algorithm. Several factors may explain these inconsistent behaviors.}

{When performing a denoising task such as in \cref{subsec:denoising}, the forward operator $\bH$ reduces to the identity matrix. In this case the proposed HRLwSGS algorithm combined with BSD leads to better results than the ones obtained with the RED-based prior. This may be explained by the fact  BSD has been trained to handle inverse Gamma noise. In this case, the BSD regularization perfectly matches the stochastic denoising task associated with the conditional distribution $p(\bz_1|\bz_2)$ defined by \eqref{eq:eq:cond_z1}. Indeed, this conditional distribution embeds  a coupling kernel $\kappa_{\rho}(\bz_1,\bz_2)$ defined by the Itakura-Saito divergence, corresponding to a denoising task in presence of an inverse Gamma noise. On the contrary the RED prior has been designed and trained to handle a Gaussian noise, seemingly inappropriate for the considered inversion problem. Conversely, when performing Poisson denoising with ULA-based sampling algorithms (i.e., with RPnP-ULA and PnP-MLA), considering a regularization designed to handle an inverse Gamma noise may be questionable. Indeed, those ULA-based algorithms does not leverage variable splitting. Instead, they impose a regularization applied to the Poisson noise-free image to be restored directly. For those splitting-free algorithms, better results could have been obtained by designing BSD under a geometry induced by  the (unnormalized) Shannon entropy in place of the Burg entropy (see \cref{subsec:bregman_div}). In this case, the underlying Bregman divergence would have boiled down to a Kullback-Leibler divergence, more compatible with a Poisson denoising task.}

{When considering more challenging inversion tasks such as deconvolution  and tomographic reconstruction as in \cref{subsec:deconvolution,subsec:tomography}, the BSD-based inversion methods  appear to be outperformed by the ones based on RED. In these cases,  the forward operator $\bH$ does not reduce to the identity matrix anymore. Thus, whatever the used algorithms (i.e., ULA or SGS), some operator residuals may induce a mismatch between the assumption driving the denoiser training and the statistics of the images to be denoised along the algorithm iterations. These mismatches may make the BSD regularization less suitable than for the denoising task, even when embedded in the proposed splitting-based HRLwSGS algorithm.  It is worth noting that the impact of such distribution mismatches has been characterized for the PnP-ULA algorithm when operating under an Euclidean geometry \cite{renaud2024plug}. It has also motivated the development of noising-denoising schemes to be included within plug-and-play inversion algorithms \cite{renaud2024plug}.}

{Besides, more generally, the considered BSD has been used off-the-shelf as designed, implemented and trained by the authors of \cite{hurault2023convergent}. In particular, the regularization potential $g(\cdot)$ has been explicitly chosen quadratic (see \eqref{eq:potential_BSD})  and the denoiser has been trained  using a quadratic loss (see \eqref{eq:training_BSD}). These choices of quadratic functions  significantly deviate from alternatives that would have been consistent with the Bregman geometry underlying the Poisson inversion (e.g., Kullback-Leibler divergence). Moreover, inverse Gamma noise is inherently a multiplicative noise, i.e., scale dependent. However the particular BSD implementation of \cite{hurault2023convergent} builds on the DRUNet architecture originally designed and optimized to handle additive Gaussian noises \cite{zhang2021plug}. These  shortcomings may have  impacts on the denoising performance of the BSD regularization and may explain the instabilities and artifacts already observed in \cite{klatzer2025efficient} despite leveraging denoiser equivariance during training \cite{terris2024equivariant}.}

\section{Conclusion}\label{sec:conclusion} 
This paper introduced a novel Bayesian model for addressing Poisson inverse problems,  combining exact and geometry-aware data augmentations. The model preserved the interpretability of the Poisson likelihood while ensuring favorable conjugacy properties and strict positivity via a Bregman divergence. The resulting augmented posterior exhibited a tractable conditional structure, enabling a sampling algorithm with three fully explicit updates and one geometry-aware Langevin step.
Specifically, efficient sampling was  achieved using a split Gibbs sampler that incorporated a Hessian Riemannian Langevin Monte Carlo (HRLMC) step, which was versatile enough to handle a wide range of priors, particularly those defined by an explicit or easily computable score function.

Experimental results on denoising, deblurring, and positron emission tomography (PET) reconstruction tasks demonstrated that the method achieved reconstruction quality competitive with or superior to  state-of-the-art optimization-based and Bayesian sampling approaches. Moreover, it provided valuable and reliable posterior uncertainty quantification, which was crucial in high-stakes applications like medical imaging. 

{Yet, a limitation of the deconvolution experiments conducted in the paper was that they consider a single blurring kernel. Although this setting provides a controlled benchmark for evaluating the performance of the proposed HRLwSGS algorithm, it does not fully reflect the diversity of possible degradations encountered in practical imaging scenarios. Assessing the robustness of the method with respect to different blur operators, kernel sizes, and potentially spatially varying degradations therefore constitutes an important direction for future work. Similarly, in the particular applicative context of medical imaging, a limitation of the PET experiments reported in the paper was that they relied on a synthetic phantom and simulated acquisition data. While such controlled settings were standard for quantitatively assessing reconstruction and uncertainty estimation methods, they did not fully capture the complexity of clinical acquisitions, including scanner imperfections, calibration errors, patient motion, attenuation mismatches, or model inaccuracies. Extending the proposed framework to real PET data therefore constitutes an important direction for future work. In particular, future investigations will focus on validating the robustness, calibration properties, and computational scalability of the method on clinical and/or experimentally acquired datasets.}

\appendix

\section{Proof of \Cref{prop:XDA}}{}\label{app:proof}
We marginalize out the latent variable $\bc$ from the augmented likelihood
	\begin{equation*}
        \sum_{\bc \in \mathbb{N}^{m \times n}} p(\by , \bc  | \bx) 
        = \sum_{\bc_{1,:} \in \mathbb{N}^{n}} \ldots \sum_{\bc_{m,:} \in \mathbb{N}^{n}} \exp \left[ - f(\bx, \bc; \by) \right]
    \end{equation*}
    where the extended data-fitting term is defined in \eqref{eq:extended_data_fitting_term} as
	\begin{equation}
		f(\bx, \bc; \by) = \sum_{i=1}^m \sum_{j=1}^n \left[ \alpha h_{ij} x_j - c_{ij} \log(\alpha h_{ij} x_j) +  \log(c_{ij}!) \right] + \iota_{\mathcal{C}(\by)}(\bc).
	\end{equation}
    Given $\bx$, the augmented likelihood can be factorized over each measurement $y_i \in \mathbb{R}$ and associated latent count vector $\bc_{i:}\in\mathbb{R}^n$, $i \in \{1,\ldots,m\}$. In particular the constraints defined by $\mathcal{C}(\by)$ decouple across the latent count vectors $\bc_{i:} = (c_{i1},\ldots,c_{in})$ of $\bc$, $i\in \left\{1,\ldots,m\right\}$, i.e., each vector satisfies $\sum_j c_{ij} = y_i$ independently  according to \eqref{eq:constraints}. Thus we have
     	\begin{equation} \label{eq:app_proof1_temp1}
        \sum_{\bc \in \mathbb{N}^{m \times n}} p(\by , \bc  | \bx) 
        = \prod_{i=1}^m \sum_{\bc_{i:} \in \mathbb{N}^{ n}} \exp \left[ - f_{i:}(\bx, \bc_{i:};y_i) \right]
    \end{equation}
    with
    \begin{equation*}
        f_{i:}(\bx,\bc_{i:};y_i) = \sum_{j=1}^n \left[ \alpha h_{ij} x_j - c_{ij} \log(\alpha h_{ij} x_j) +  \log(c_{ij}!) \right] + \iota_{\mathcal{C}_i(y_i)}(\bc_{i:}).
    \end{equation*}    
	Due to the indicator function in $f_i(\bx, \bc_{i:}; y_i)$, the sum over $\bc_{i:} \in \mathbb{N}^{ n}$ in \eqref{eq:app_proof1_temp1} restricts to the latent vector $\bc_{i:}$ such that $\bc_{i:} \in \mathcal{C}_i(y_i)$, i.e., 
    \begin{equation*}
       \sum_{\bc \in \mathbb{N}^{m \times n}} p(\by , \bc  | \bx) 
        =\prod_{i=1}^m \left\{ \exp \left[ -\sum_{j=1}^n  \alpha h_{ij} x_j \right]
             \sum_{\bc_{i:} \in \mathcal{C}_i(y_i)} \prod_{j=1}^n
           \frac{(\alpha h_{ij} x_j)^{c_{ij}}}{c_{ij}!} \right\}.
    \end{equation*}
     Moreover, for $i\in \left\{1,\ldots,m\right\}$, the multinomial theorem leads to \cite{stanley2011enumerative}
    \begin{equation*}        
       \sum_{\bc_{i:} \in \mathcal{C}_i(y_i)} \prod_{j=1}^n \frac{(\alpha h_{ij} x_j)^{c_{ij}}}{c_{ij}!}
        = \frac{\left(\sum_{j=1}^n \alpha h_{ij} x_j\right)^{y_i}}{y_i!}
        = \frac{(\alpha \bh_i^\top \bx)^{y_i}}{y_i!}.
    \end{equation*}
    Finally,
    \begin{align*}
        \sum_{\bc \in \mathbb{N}^{m \times n}} p(\by , \bc  | \bx) 
        &=   \exp \left\{ -\sum_{i=1}^m \left[
           \alpha \bh_i^\top \bx - y_i \log(\alpha \bh_i^\top \bx)
           + \log(y_i!) \right] \right\} \\
        &=   p(\by | \bx).
    \end{align*}
    This completes the proof.

\section{Proof of \Cref{prop:AXDA_TV2}}\label{app:proof_AXDA_TV2}
{
Starting from the target distribution $\pi(\bx)$, a first level of AXDA induces the augmented distribution 
\begin{equation}\label{eq:proof_AXDA_TV2:AXDA_1}
    \pi_{\rho}(\bx,\bz_{1}) = \pi(\bz_1)\kappa_{\rho}(\bz_1,\bx).
\end{equation}
Let
\begin{align}
    \tilde{\pi}_{\rho}(\bx) &= \int\pi_{\rho}(\bx,\bz_{1}) d\bz_1 \nonumber 
\end{align}
define the marginal distribution resulting from the marginalization of the augmented distribution $\pi_{\rho}(\bx,\bz_{1})$ w.r.t. the splitting variable $\bz_1$. Since $\kappa_\rho$ satisfies \Cref{assump:kernel}, then \Cref{prop:AXDA_TV} ensures that
\begin{equation}\label{eq:proof_term1}
  \left\|\tilde{\pi}_{\rho}(\bx) - {\pi}(\bx)\right\|_\mathrm{TV}  \xrightarrow[\rho \to 0]{} 0.
\end{equation}
A double-level AXDA conducted on the target distribution $\pi(\bx)$ induces the doubly augmented distribution 
\begin{equation}\label{eq:proof_AXDA_TV2:AXDA_2}
\pi_{\rho}(\bx,\bz_{1:2}) = \pi(\bz_1)\kappa_{\rho}(\bz_1,\bz_2)\kappa_{\rho}^{\prime}(\bz_2,\bx).
\end{equation}
Using the definition in \eqref{eq:proof_AXDA_TV2:AXDA_1}, the distribution in \eqref{eq:proof_AXDA_TV2:AXDA_2} can be rewritten as
$$\pi_{\rho}(\bx,\bz_{1:2}) = \pi_{\rho}(\bz_2,\bz_1)\kappa_{\rho}^{\prime}(\bz_2,\bx).$$
and interpreted as the result of a single-level AXDA conducted on the previously augmented distribution $\pi_{\rho}(\bx,\bz_{1})$. Let
\begin{equation*}
    \bar{\pi}_{\rho}(\bx,\bz_1) = \int\pi_{\rho}(\bx,\bz_{1:2}) d\bz_2 
                                    = \int \pi_{\rho}(\bz_2,\bz_1)\kappa_{\rho}^{\prime}(\bz_2,\bx) d\bz_2. 
\end{equation*}
define the singly augmented distribution resulting from the marginalization of the doubly augmented $\pi_{\rho}(\bx,\bz_{1:2})$ w.r.t. the splitting variable $\bz_2$. Since $\kappa_\rho^{\prime}$ satisfies \Cref{assump:kernel}, then \Cref{prop:AXDA_TV} ensures that
\begin{equation}\label{eq:proof_term2}
  \left\|\bar{\pi}_{\rho}(\bx,\bz_1) - {\pi}_{\rho}(\bx,\bz_1)\right\|_\mathrm{TV}  \xrightarrow[\rho \to 0]{} 0
\end{equation}
Let 
\begin{equation*}
    \tilde{\tilde{\pi}}_{\rho}(\bx) = \int{\pi}_{\rho}(\bx,\bz_{1:2}) d\bz_2 d\bz_1 
                                    = \int\bar{\pi}_{\rho}(\bx,\bz_{1}) d\bz_1
\end{equation*}
define the marginal distribution resulting from the marginalization of the doubly augmented distribution ${\pi}_{\rho}(\bx,\bz_{1:2})$ w.r.t. the two splitting variables $\bz_1$ and $\bz_2$. Then
\begin{align}
 \left\|\tilde{\tilde{\pi}}_{\rho}(\bx) - {\pi}(\bx)\right\|_\mathrm{TV} &= \left\|\tilde{\tilde{\pi}}_{\rho}(\bx) - \tilde{\pi}(\bx)  
            +  \tilde{\pi}(\bx)- {\pi}(\bx)\right\|_\mathrm{TV} \nonumber \\
  & \leq \left\|\tilde{\tilde{\pi}}_{\rho}(\bx) - \tilde{\pi}_{\rho}(\bx)  \right\|_\mathrm{TV}
            +  \left\|\tilde{\pi}_{\rho}(\bx)- {\pi}(\bx)\right\|_\mathrm{TV} \nonumber \\
  & \leq \left\|\int\bar{\pi}_{\rho}(\bx,\bz_{1}) - {\pi}_{\rho}(\bx,\bz_{1}) d\bz_1  \right\|_\mathrm{TV}
            +  \left\|\tilde{\pi}_{\rho}(\bx)- {\pi}(\bx)\right\|_\mathrm{TV} \nonumber \\
      & \leq \int \left\|\bar{\pi}_{\rho}(\bx,\bz_{1}) - {\pi}_{\rho}(\bx,\bz_{1}) \right\|_\mathrm{TV} d\bz_1 
            +  \left\|\tilde{\pi}_{\rho}(\bx)- {\pi}(\bx)\right\|_\mathrm{TV}.   \label{eq:proof_term12}
\end{align}
In the right-hand side of \eqref{eq:proof_term12}, as $\rho \rightarrow 0$, combining \eqref{eq:proof_term2} with the dominated convergence theorem ensures that the first term vanishes while the second term also vanishes by \eqref{eq:proof_term1}. This concludes the proof.}

\section{Alternative augmented Bayesian models}\label{sec:alternate_AXDA_scheme}
{This section systematically examines several possible variants of the augmentation scheme, with the dual purpose of motivating the choices underlying the Bayesian model proposed in \cref{sec:bayesian_model} and providing a self-contained comparison of the resulting trade-offs. These models naturally arise when questioning the rationale for conducting the doubled-level asymptotically exact augmentations leading to the target posterior \eqref{eq:augmented_posterior_2}. Thus, \cref{subsec:simple_AXDA} first highlights the limit of the Bayesian model which would result from a single AXDA scheme as conducted in \cref{subsubsec:AXDA_1}. Besides, since the splitting kernel derives from a non-symmetric Bregman divergence, swapping the main and splitting variables leads to different augmentation schemes. To appreciate the impact of this asymmetry, the Bayesian model resulting from a single-level AXDA as in \cref{subsec:simple_AXDA} but with a reversed splitting kernel is discussed in \cref{subsec:simple_AXDA_symmetrized}. Finally several variants of the double-level AXDA  scheme are discussed in \cref{subsec:double_AXDA} differing by the orientation of one or two of these kernels. For each discussed model, we derive the full conditional distributions underlying the resulting SGS and assess the tractability of each sampling step. This comparative analysis, summarized in \cref{tab:models}, makes explicit the structural reasons that ultimately justify the formulation retained in the core part of the paper (see \cref{sec:bayesian_model}).
}

{For all discussed variants, sampling from the conditional posterior of the counts $\bc$ remains identical and can be  straightforwardly performed according to \eqref{eq:sample_n}. Thus, this sampling step  will  not be discussed in the subsequent derivations.}

\begin{table}\small
\setlength{\tabcolsep}{4.2pt}
\renewcommand{\arraystretch}{1.5}
\begin{center}
\caption{{Comparison of the Bayesian models with respect to the coupling kernels and resulting sampling steps.}\label{tab:models}}
\begin{tabular}{ c | c | c | c | c | c | c}
\toprule
                 & \multirow{2}{*}{Proposed model} & \multicolumn{5}{c}{Alternative models} \\ 
                 &                & \multicolumn{2}{c|}{Single-level AXDA} &  \multicolumn{3}{c}{Double-level AXDA} \\
                 \cline{2-7}
                 &  sec. \ref{sec:bayesian_model}   & sec. \ref{subsec:simple_AXDA} &  sec. \ref{subsec:simple_AXDA_symmetrized}  & sec. \ref{sec:alternate_modelA} & sec. \ref{sec:alternate_modelC} & sec. \ref{sec:alternate_modelB} \\
\midrule                 
kernels & $\left\{\kappa_\rho,\kappa_\rho^{\prime}\right\}$ & $\kappa_\rho$ & $\kappa_\rho^\prime$  & $\left\{\kappa_\rho^{\prime},\kappa_\rho^{\prime}\right\}$ & $\left\{\kappa_\rho,\kappa_\rho\right\}$ & $\left\{\kappa_\rho^{\prime},\kappa_\rho\right\}$ \\
\midrule  
  $\bc$-update   & Multinomial             & Multinomial               & Multinomial             & Multinomial             & Multinomial               & Multinomial\\ 
 $\bx$-update   & Gamma             & GIG               & Gamma             & Gamma             & GIG               & GIG\\  
 $\bz_1$-update & well-cond.  & well-cond.  & ill-cond.   & ill-cond.   & well-cond.  & ill-cond.\\
 $\bz_2$-update & InvGamma          & --               &  --              & GIG               & GIG               & Gamma\\
 \bottomrule
\end{tabular}
\end{center}
\end{table}

\subsection{Single variable splitting}\label{subsec:simple_AXDA}
{This section discusses the limitation of performing a single-level AXDA with the coupling kernel $\kappa_{\rho}$ defined by \eqref{eq:IS_kernel} and recalled below
\begin{equation}\label{eq:IS_kernel_bis}
\kappa_{\rho}(\bz_1, \bx) = \exp \left\{ -\frac{1}{\rho} d_\mathrm{IS}(\bz_1, \bx)  + \varphi(\bz_1)  \right\}.
\end{equation}
This AXDA scheme yields the augmented distribution \eqref{eq:augmented_posterior_1}
\begin{equation}
		\pi_{\rho}(\bx, \bc, \bz_{1}) \propto \exp\bigg\{ - f(\bx, \bc; \by) - g(\bz_1)  -\frac{1}{\rho} d_\mathrm{IS}(\bz_1, \bx) + \varphi(\bz_1)  \bigg\}.
\end{equation}
Implementing an SGS targeting this distribution consists in sampling from the three associated conditional distributions. Specifically, thanks to the decoupling induced by AXDA, the conditional distribution of the variable $\bx$ factorizes over coordinates as a product of $n$ generalized inverse Gaussian distributions (GIG)
\begin{equation}
    x_j | \bc_{:,j}, z_{1j} \sim \mathrm{GIG}\left(2 \alpha \sum_{i=1}^m h_{ij},\;  \frac{2}{\rho} z_{1j}, \sum_{i=1}^m c_{ij} - \frac{1}{\rho} + 1\right), \ \forall j \in \left\{ 1, \ldots, n\right\}.
\end{equation}
As for the conditional distribution of the splitting variable $\bz_1$, it takes the form $p(\bz_1 | \bx) \propto \exp[- U(\bz_1)]$ with
\begin{equation*}
    U(\bz_1) = g(\bz_1) + \Bigl(1 - \tfrac{1}{\rho}\Bigr)\sum_{j=1}^n \log z_j + \frac{1}{\rho}\sum_{j=1}^n \frac{z_{1j}}{x_j}.
\end{equation*}
As in \cref{subsec:non_standard_sampling}, sampling from this distribution can be performed using an HRLMC step similar to \eqref{eq:sample_z1}.}

{As a conclusion, when compared to the double-level AXDA model proposed in \cref{sec:bayesian_model}, it appears that one limitation of this single-level AXDA scheme lies from the sampling of the GIG distributions. While the GIG family admits exact samplers, their implementation is non-trivial: dedicated algorithms are required, and their computational cost is substantially higher than that of standard Gamma or Gaussian samplers \cite{pena2025properties}. Therefore this step  represents a significant implementation burden. 
}

\subsection{Single variable splitting with reversed coupling kernel}\label{subsec:simple_AXDA_symmetrized}
{A natural variant of the single-level AXDA scheme discussed in \cref{subsec:simple_AXDA} is obtained by adopting the reversed coupling kernel
\begin{equation}\label{eq:IS_kernelter}
{\kappa}^{\prime}_\rho(\bz_1,\bx) = \exp \left\{ -\frac{1}{\rho} d_\mathrm{IS}(\bx, \bz_1)  + \varphi(\bz_1)  \right\}
\end{equation}
obtained by  reversing the argument ordering of the IS divergence of the coupling kernel ${\kappa}_\rho(\bz_1,\bx)$. This variant leads to the posterior distribution
\begin{equation}
		\pi_{\rho}(\bx, \bc, \bz_{1}) \propto \exp\bigg\{ - f(\bx, \bc; \by) - g(\bz_1)  -\frac{1}{\rho} d_\mathrm{IS}(\bx, \bz_1) + \varphi(\bz_1)  \bigg\}.
\end{equation}
In this case, again, the conditional distribution of the variable $\bx$ factorizes over coordinates but, here, as a product of Gamma distributions
\begin{equation*}
    x_j | \bc_{:,j}, z_{1j} \sim
    \mathrm{Gamma}\left( 
        \sum_{i=1}^m c_{ij} + \frac{1}{\rho} + 1,\; 
         \alpha \sum_{i=1}^m h_{ij} + \frac{1}{\rho z_j}
    \right).
\end{equation*}
This is a decisive advantage over the previous considered in \cref{subsec:simple_AXDA} since sampling from Gamma distributions is much simpler than sampling from GIG distributions. Regarding the conditional distribution of $\bz_1$, it now writes $p(\bz_1 | \bx) \propto \exp[- U(\bz_1)]$ with the potential
\begin{equation}
    U(\bz_1) = g(\bz_1) 
    + \Bigl(1 + \tfrac{1}{\rho}\Bigr)\sum_{j=1}^n \log z_j 
    + \frac{1}{\rho}\sum_{j=1}^n \frac{x_j}{z_{1j}}.
\end{equation}
As in \cref{subsec:non_standard_sampling}, one may think of applying  HRLMC steps to sample from this distribution. However, contrary to the previously considered models, the potential $U(\bz_1)$ now involves terms such as $x_j / z_{1j}$, which cause gradients and Hessians to become ill-conditioned when $z_{1j}$ approaches zero. As a result, the convergence analysis of the HRLMC algorithm requires stricter assumptions on the boundedness of both $\bx$ and $\bz_1$, and above all, a tighter control of the Lipschitz constant $L_{\text{D}}$ of the regularization potential $g(\cdot)$. For example, when $g(\cdot)$ derives from a prior associated with RED, ensuring the positivity of the relative strong convexity constant imposes particularly more restrictive conditions than the one required by the potential exhibited in \cref{subsec:non_standard_sampling} (see \cref{HRLMC_assumptions} for more details).
}

\subsection{Double variable splitting with reversed coupling kernels}\label{subsec:double_AXDA}
{Recall that each of the two single-level AXDA schemes considered in \cref{subsec:simple_AXDA,subsec:simple_AXDA_symmetrized} leads to a manageable conditional distribution for one variable  at the expense of difficulties raised by the other. Using the coupling kernel $\kappa_{\rho}$  yields a well-conditioned HRLMC step for $\bz_1$ but required a GIG sampler for $\bx$. Conversely, using the reversed coupling kernel $\kappa_{\rho}^\prime$   simplifies the $\bx$-updating rule by sampling from Gamma distributions at the price of severe ill-conditioning in the $\bz_1$-updating rule. \Cref{subsubsec:AXDA_2} shows that including a second level of AXDA simultaneously leads to a well-conditioned potential for the conditional distribution of the splitting variable $\bz_1$ while ensuring simple updating rules for $\bz_2$ and $\bx$. Again, this illustrates the interest of introducing a second splitting variable $\bz_2$ as a mediator to precisely to reconcile the competing requirements induced by the sampling of the two other variables.}

{However, again, considering two levels of AXDA with non-symmetric kernels questions the places of the variables in the two coupling kernels. The Bayesian model derived in \cref{sec:bayesian_model} relies on the successive use of the kernels $\kappa_{\rho}$ and then $\kappa_{\rho}^\prime$. The remaining of this section discusses possible variants  with respect to this initial Bayesian model referred to  as $\left\{\kappa_{\rho},\kappa_{\rho}^{\prime}\right\}$.}

\subsubsection{Variant $\left\{\kappa_{\rho}^{\prime},\kappa_{\rho}^{\prime}\right\}$: reversed first kernel}\label{sec:alternate_modelA}
{When the role of $\bz_1$ and $\bz_2$ are reversed in the first kernel, i.e., using $\kappa_{\rho}^{\prime}(\bz_1,\bx)$ instead of $\kappa_{\rho}(\bz_1,\bx)$, the joint distribution writes
\begin{equation} 
		\pi_{\rho} (\bx, \bc, \bz_{1:2}) \propto \exp\bigg\{ - f(\bx, \bc; \by) - g(\bz_1)  
         - \frac{1}{\rho} d_\mathrm{IS}(\bz_2, \bz_1) + \varphi(\bz_1) - \frac{1}{\rho} d_\mathrm{IS}(\bx, \bz_2) + \varphi(\bz_2) \bigg\}.
\end{equation}
whose conditional distributions write
\begin{align}
x_j | \bc_{:,j}, z_{2j} &\sim 
    \mathrm{Gamma} \left(
    \sum_{i=1}^m c_{ij}  + \frac{1}{\rho} + 1,\,
   \alpha \sum_{i=1}^m h_{ij} + \frac{1}{\rho z_{2j}}\right), \ \forall j \in \left\{ 1, \ldots, n\right\},
\label{eq:modelA_x}\\
p(\bz_1 \mid \bz_2) &\propto \exp\Bigl\{
    -g(\bz_1)
    - \Bigl(1+\tfrac{1}{\rho}\Bigr)\sum_j \log z_{1j}
    - \tfrac{1}{\rho}\sum_j \tfrac{z_{2j}}{z_{1j}}
\Bigr\},
\label{eq:modelA_z1} \\
z_{2j} | z_{1j},x_j &\sim 
    \mathrm{GIG} \left(\frac{2}{\rho z_{1j}},
    \frac{2x_j}{\rho}, 0\right), \ \forall j \in \left\{ 1, \ldots, n\right\}.
\label{eq:modelA_z2}
\end{align}
Handling such conditional distributions have been already discussed in previous models. It appears that the resulting sampling scheme combines the difficulty associated with the two single-level AXDA models discussed in \cref{subsec:simple_AXDA,subsec:simple_AXDA_symmetrized}: the $\mathbf{z}_2$-updating rule involves sampling from GIG distributions as in \cref{subsec:simple_AXDA} and the $\mathbf{z}_2$-updating rule involves a possibly ill-conditioned HRLMC step  as in \cref{subsec:simple_AXDA_symmetrized}.
}

\subsubsection{Variant $\left\{\kappa_{\rho},\kappa_{\rho}\right\}$: reversed second kernel}\label{sec:alternate_modelC}
{When the role of $\bz_1$ and $\bz_2$ are reversed in the second kernel, i.e., using $\kappa_{\rho}(\bz_2,\bx)$ instead of $\kappa_{\rho}^{\prime}(\bz_2,\bx)$, the joint distribution writes
\begin{equation} 
		\pi_{\rho} (\bx, \bc, \bz_{1:2}) \propto \exp\bigg\{ - f(\bx, \bc; \by) - g(\bz_1)  
         - \frac{1}{\rho} d_\mathrm{IS}(\bz_1, \bz_2) + \varphi(\bz_1) - \frac{1}{\rho} d_\mathrm{IS}(\bz_2,\bx) + \varphi(\bz_2) \bigg\}
\end{equation}
with the conditional distributions
\begin{align}
x_j | \bc_{:,j}, z_{2j}& \sim
    \mathrm{GIG}\left(2 \alpha \sum_{i=1}^m h_{ij},
    \frac{2z_{2j}}{\rho}, \sum_{i=1}^m c_{ij}  - \frac{1}{\rho}+1\right), \ \forall j \in \left\{ 1, \ldots, n\right\}
\label{eq:modelC_x}\\
p(\bz_1 | \bz_2) &\propto \exp\Bigl\{
    -g(\bz_1)
    + \sum_j \log z_{1j}
    + \tfrac{1}{\rho}\sum_{j=1}^n \left(\tfrac{z_{1j}}{z_{2j}}- \log z_{ij}\right)
\Bigr\},
\label{eq:modelC_z1} \\
z_{2j} | z_{1j},x_j &\sim 
    \mathrm{GIG} \left(\frac{2}{\rho x_j},
    \frac{2z_{1j}}{\rho}, 0\right), \ \forall j \in \left\{ 1, \ldots, n\right\}. \label{eq:modelC_z2}
\end{align}
When compared to the previous variants, the resulting SGS now involves the sampling from two sets of independent GIG distributions for the $\bx$- and $\bz_2$-updating rules, which requires dedicated algorithms. However the conditional distribution $p(\bz_1 | \bz_2) $ in \eqref{eq:modelC_z1} exhibits a smoother potential than the one resulting from the variant $\left\{\kappa_{\rho}^{\prime},\kappa_{\rho}^{\prime}\right\}$.}

\subsubsection{Variant $\left\{\kappa_{\rho}^{\prime},\kappa_{\rho}\right\}$: reversed first and second kernels}\label{sec:alternate_modelB}
{We now consider the alternative formulation in which the role of the variables have been reversed in the two kernels, leading to the posterior distribution
\begin{equation} 
		\pi_{\rho} (\bx, \bc, \bz_{1:2}) \propto \exp\bigg\{ - f(\bx, \bc; \by) - g(\bz_1)  
         - \frac{1}{\rho} d_\mathrm{IS}(\bz_2, \bz_1) + \varphi(\bz_1) - \frac{1}{\rho} d_\mathrm{IS}(\bz_2, \bx) + \varphi(\bz_2) \bigg\}.
\end{equation}
The conditional distributions now write
\begin{align}
x_j | \bc_{:,j}, z_{2j} & \sim
    \mathrm{GIG}\left(2 \alpha \sum_{i=1}^m h_{ij},
    \frac{2z_{2j}}{\rho}, \sum_{i=1}^m c_{ij}  - \frac{1}{\rho}+1\right), \ \forall j \in \left\{ 1, \ldots, n\right\}
\label{eq:modelB_x}\\
p(\bz_1 \mid \bz_2) &\propto \exp\Bigl\{
    -g(\bz_1)
    - \Bigl(1+\tfrac{1}{\rho}\Bigr)\sum_j \log z_{1j}
    - \tfrac{1}{\rho}\sum_j \tfrac{z_{2j}}{z_{1j}}
\Bigr\},
\label{eq:modelB_z1} \\
z_{2j} | z_{1j},x_j &\sim 
    \mathrm{Gamma}\left(\frac{2}{\rho},
    \frac{1}{\rho}\Bigl(\frac{1}{z_{1j}}+\frac{1}{x_j}\Bigr)\right), \ \forall j \in \left\{ 1, \ldots, n\right\}.
\label{eq:modelB_z2}
\end{align}
The $\mathbf{z}_2$-updating rule can be easily carried out via an exact Gamma sampler, which is simpler than the GIG sampler used in previous variants. However, the GIG structure reappears in the $\bx$-update \eqref{eq:modelB_x}, which now requires a non-trivial sampling procedure as described in \cref{subsec:simple_AXDA}. Moreover the $\mathbf{z}_1$-updating rule
involves a conditional distribution whose potential may lead to ill-conditionned HRLMC step. This variant combines all issues raised in the previous single- and double-level AXDA models.}

\section{Denoiser-based priors}\label{sec:denoisers}
This section reviews two classes of regularization which leverage denoisers, namely regularization-by-denoising (RED) and Bregman score denoisers (BSD). These regularizations  have been considered in this paper to derive particular instances of the proposed HRLwSGS algorithm, namely RED-HRLwSGS and BSD-HRLWsGS.

\subsection{Regularization-by-denoising (RED)}\label{subsec:RED}
The RED framework defines the regularization term $g(\cdot)$ using a denoising operator $\textsf{D}_\nu: \R^n \to \R^n$ \cite{romano2017little},
\begin{equation}\label{eq:red-regularization}
    g(\bx) = \frac{1}{2} \bx^\top (\bx - \textsf{D}_\nu(\bx))
\end{equation}
where $\nu$ controls the denoising strength. This formulation explicitly integrates the denoising operator into the regularization term, allowing the use of advanced denoisers when solving inverse problems. When the denoiser $\textsf{D}_\nu$ is locally homogeous and has a symmetric Jacobian \cite{reehorst2018regularization}, the gradient of the RED regularizer (i.e. the prior score function in a Bayesian context) has a simple form 
\begin{equation*}
     \nabla g(\bx) = \bx - \textsf{D}_\nu(\bx).
\end{equation*}
A probabilistic formulation of RED has been recently proposed \cite{Faye2024}, which enables the use of a RED potential to define a corresponding prior distribution within a Bayesian framework.

\subsection{Bregman score denoiser (BSD)}\label{sec:bregman_score_denoiser}
The Bregman score denoiser is a denoising operator designed to operate under a Bregman geometry \cite{hurault2023convergent}. It generalizes the Gaussian gradient-step denoiser \cite{hurault2021gradient} to handle noise models associated with Bregman divergences. Formally, the denoiser $\mathcal{B}_\nu \colon \mathbb{R}^n \to \mathbb{R}^n$ is defined by
\begin{equation}\label{eq:bregman_score_denoiser}
	\mathcal{B}_\nu(\bx) = \bx - \left( \nabla^2 h(\bx) \right)^{-1} \cdot \nabla g_\nu(\bx)
\end{equation}
where $h \colon \mathbb{R}^n \to \mathbb{R}$ is a strictly convex and $\mathcal{C}^2$ function defining the Bregman geometry, and $g_\nu \colon \mathbb{R}^n \to \mathbb{R}$ is a (possibly nonconvex) potential function. Following the formulation of Hurault \textit{et al.}~\cite{hurault2023convergent}, the noise-free observation ${\bx_0}\in \mathbb{R}^n$ is assumed to be corrupted by a Bregman noise, resulting in the likelihood function
\begin{equation}\label{eq:model_noise}
	p(\bx | \bx_0) \propto \exp\left\{ - \nu d_h(\bx_0, \bx) \right\}
\end{equation}
where $d_h$ denotes the Bregman divergence \eqref{eq:bregman_divergence} generated by $h$. This framework encompasses classical Gaussian models as a special case, and naturally extends to non-Euclidean settings. As considered in \cite{hurault2023convergent},  when $h$ is the Burg’s entropy and $\nu > 1$, the distribution $p(\bx | \bx_0)$ writes as a product of $n$ independent inverse Gamma distributions, i.e,
\begin{equation}\label{eq:inverse_gamma}
	x_j | x_{0j} \sim   \mathrm{InvGamma}(\nu - 1, \nu x_{0j}), \quad j=1, \dots, n.
\end{equation}
Importantly, under mild regularity assumptions, $\mathcal{B}_\nu$ also admits a variational interpretation as a Bregman proximal operator. Specifically, for any $\bx \in \operatorname{int} \operatorname{dom}(h)$, we have
\begin{equation}
	\mathcal{B}_\nu(\bx) \in \arg\min_{\bx_0 \in \mathbb{R}^n} \left\{ d_h(\bx_0, \bx) + \psi_\nu(\bx_0) \right\}
\end{equation}
where $ \psi_\nu$ is an explicit function derived from $g_\nu$. The Bregman score denoiser can be considered as a natural extension of score-based denoising to inverse problems governed by non-Euclidean geometries, such as those arising under Poisson noise. It provides a flexible and theoretically grounded framework for both variational and sampling-based image restoration methods.

In practice, Hurault \emph{et al.} propose to define the regularization function $g_\nu$ as the following quadratic potential \cite{hurault2023convergent, hurault2021gradient}
\begin{equation}\label{eq:potential_BSD}
    g_\nu(\bx) = \frac{1}{2} \| \bx - \textsf{N}_\nu(\bx) \|^2
\end{equation}
where $\textsf{N}_\nu \colon \mathbb{R}^n \to \mathbb{R}^n$ denotes a deep convolutional neural network, instantiated in practice with the DRUNet architecture~\cite{zhang2021plug}. {By denoting $p_0$ the distribution of noise-free images $\bx_0$ and $p$ the inverse Gamma distribution of the noisy images $\bx$ following the conditional model \eqref{eq:inverse_gamma}, BSD is then trained according to the loss function
\begin{equation}\label{eq:training_BSD}
    \mathcal{L}_{\nu} = \mathrm{E}_{\bx_0 \sim p_{0}, \bx \sim p} \left[\left\| \mathcal{B}_\nu(\bx) - \bx_0 \right\|^2_2\right].
\end{equation}}
The gradient $\nabla g_\nu$ of the regularization potential can finally be computed via automatic differentiation, enabling the seamless integration of complex neural architectures into the inversion pipeline.

\section{Properties of the HRLMC and related assumptions}
\label{HRLMC_assumptions}
This section discusses the assumptions required to ensure the convergence of the HRLMC algorithm.  Both works \cite{zhang2020wasserstein}  and \cite{li2022mirror} investigate the convergence properties of Langevin algorithms defined in non-Euclidean geometries, where updates are performed using mirror maps derived from strictly convex functions. These algorithms are particularly well-suited for sampling over constrained domains such as the positive orthant $\mathbb{R}_{+}^{n}$, as in the case of Poisson inverse problems. Zhang \emph{et al.} introduce this discretization scheme of the Riemannian Langevin diffusion and analyze its convergence in Wasserstein distance \cite{zhang2020wasserstein}. Under mild assumptions, a non-asymptotic upper-bound on the sampling error can be derived and HRLMC exhibits convergence guarantees analogous to those of LMC algorithms operating under Euclidean geometry. More recently, Li \emph{et al.} \cite{li2022mirror} conducted an improved analysis showing that the bias of the discretized sampler vanishes as the step size approaches zero. This analysis requires only a subset of the assumptions initially stated in \cite{zhang2020wasserstein}, i.e., 
\begin{itemize}
    \item[(A1)] Modified self concordance \label{item:A1}
    \begin{equation*}
\left\| \nabla^2\phi(\bz)^{1/2} - \nabla^2\phi(\bz')^{1/2} \right\|_{HS} \leq \sqrt{\alpha} \left\| \nabla \phi(\bz) - \nabla \phi(\bz') \right\|_2
\end{equation*}
   \item[(A2)] Relative strong convexity \label{item:A2}
\begin{equation*}
 m \|\nabla \phi(\bz) - \nabla \phi(\bz')\|_2^2 \leq \langle \nabla U(\bz) - \nabla U(\bz'), \nabla \phi(\bz) - \nabla \phi(\bz') \rangle
\end{equation*}
 \item[(A3)] Relative Lipschitz-smoothness
 \begin{equation*}
\|\nabla U(\bz) - \nabla U(\bz')\|_2 \leq M \|\nabla \phi(\bz) - \nabla \phi(\bz')\|_2.
\end{equation*}
\end{itemize}

In \cref{sec:proposed_algorithm}, the HRLMC algorithm is implemented using the mirror map $\phi(\cdot)$ defined in \eqref{eq:burg} to target the distribution whose  potential function $U(\cdot)$ is defined in \eqref{eq:potential}. In what follows, we show that these quantities satisfy the conditions (A1)--(A3) when $g(\cdot)$ is specifically chosen as the RED potential \eqref{eq:red-regularization}. To do so, we further assume that the denoiser $\mathsf{D}_\nu$ is $L_\mathsf{D}$-Lipschitz and the variables $z_j$ are bounded, i.e., it exists $\epsilon_z >0$ and $C_z < \infty$ such that $ \epsilon_z \leq z_j \leq C_z$.

\paragraph{(A1) Modified self-concordance}
The choice of the mirror map leads to
\begin{equation*}
\nabla^2\phi(\bz)^{1/2} - \nabla^2\phi(\bz')^{1/2} = \mathrm{diag}\left( \frac{1}{z_1} - \frac{1}{z'_1}, \ldots, \frac{1}{z_n} - \frac{1}{z'_n} \right),
\end{equation*}
and thus
\begin{equation*}
\left\| \nabla^2\phi(\bz)^{1/2} - \nabla^2\phi(\bz')^{1/2} \right\|_{HS}  = \sqrt{ \sum_{j=1}^n \left( \frac{1}{z_j} - \frac{1}{z'_j} \right)^2 }.
\end{equation*}
Similarly
\begin{equation*}
\left\| \nabla \phi(\bz) - \nabla \phi(\bz') \right\|_2 = \sqrt{ \sum_{j=1}^n \left( \frac{1}{z'_j} - \frac{1}{z_j} \right)^2 },
\end{equation*}
which shows that Assumption (A1) is verified with $\alpha=1$.

\paragraph{(A2) Relative strong convexity}
The gradient of the potential writes
\begin{equation*}
\nabla U(\bz) = \nabla g(\bz) + \left(\frac{1}{\rho} - 1\right) \nabla \phi(\bz) + \frac{1}{\rho \bz}.
\end{equation*}
When $g(\cdot)$ is chosen as the RED potential, its gradient is given by $\nabla g(\bz) = \beta ( \bz - \mathsf{D}_\nu(\bz))$ and the inner product in (A2) expands as
\begin{align*}
\beta \langle \nabla g(\bz) - \nabla g(\bz'), \nabla \phi(\bz) - \nabla \phi(\bz') \rangle + \left(\frac{1}{\rho} - 1\right) \|\nabla \phi(\bz) - \nabla \phi(\bz')\|_2^2.
\end{align*}
The first term decomposes as
\begin{align*}
\langle \bz - \bz', \nabla \phi(\bz) - \nabla \phi(\bz') \rangle - \langle \mathsf{D}_\nu(\bz) - \mathsf{D}_\nu(\bz'), \nabla \phi(\bz) - \nabla \phi(\bz') \rangle.
\end{align*}
with
\begin{align*}
\langle \bz - \bz', \nabla \phi(\bz) - \nabla \phi(\bz') \rangle = \sum_{j=1}^n \frac{(z_j - z'_j)^2}{z_j z'_j} \geq \frac{1}{C_z^2} \|\bz - \bz'\|_2^2,
\end{align*}
and
\begin{align*}
|\langle \mathsf{D}_\nu(\bz) - \mathsf{D}_\nu(\bz'), \nabla \phi(\bz) - \nabla \phi(\bz') \rangle| 
\leq L_\mathsf{D} C_z^2 \|\nabla \phi(\bz) - \nabla \phi(\bz')\|_2^2.
\end{align*}
Thus Assumption (A2) is ensured with
\begin{equation*}
m = \frac{1}{\rho} - 1 + \beta\left( \frac{\epsilon_z^4}{C_z^2} - L_\mathsf{D} C_z^2 \right).
\end{equation*}
The sign of $m$ depends on the parameters $\rho$, $\beta$, $L_\mathsf{D}$, $\epsilon_z$, and $C_z$.  
Since $\epsilon_z < C_z$, we have $\frac{\epsilon_z^4}{C_z^2} < \epsilon_z^2 < C_z^2$, hence 
$\beta\left(\frac{\epsilon_z^4}{C_z^2} - L_\mathsf{D} C_z^2\right) \leq \beta C_z^2 (1 - L_\mathsf{D})$.  
This term is positive if $\frac{\epsilon_z^4}{C_z^4} \geq L_\mathsf{D}$, in which case $m \geq 0$ whenever $\rho \leq 1$.  
If $\frac{\epsilon_z^4}{C_z^4} \leq L_\mathsf{D}$, the term is negative and $m \geq 0$ requires 
\begin{equation*}
\rho \leq \frac{1}{1 + \beta \left(L_\mathsf{D} C_z^2 - \frac{\epsilon_z^4}{C_z^2}\right)}.
\end{equation*}

\paragraph{(A3) Relative Lipschitz-smoothness}
Using $\|\nabla g(\bz) - \nabla g(\bz')\|_2 \leq (1+L_\mathsf{D})\|\bz - \bz'\|_2$ and $\|\bz - \bz'\|_2 \leq C_z^2 \|\nabla \phi(\bz) - \nabla \phi(\bz')\|_2$, we obtain Assumption (A3) with
\begin{equation*}
M = \beta (1 + L_\mathsf{D}) C_z^2 + \left| \frac{1}{\rho} - 1 \right|.
\end{equation*}

\section{Highest posterior density (HPD) intervals, calibration and coverage maps}\label{app:HPD}
{
The highest posterior density (HPD) region of confidence $100(1-\alpha) \%$ is the smallest credible region $\mathcal{X}_{\alpha} \subset \mathbb{R}^n $ with respect to the Lebesgue measure defined as 
\begin{equation}
    \mathcal{X}_{\alpha} = \left\{\bx \in \mathbb{R}^n : p(\bx|\by) > \eta_{\alpha}\right\}
\end{equation}
such that
\begin{equation}
    \mathrm{P}[\bx \in \mathcal{X}_{\alpha}  |\by] = 1- \alpha.
\end{equation}
A Bayesian reconstruction model is said \textit{perfectly calibrated} if, for all levels $(1-\alpha) \in[0,1]$, the HPD interval contains the true  image $\bx^*$ with probability $(1-\alpha)$.
}
{
To experimentally assess the calibration, pixel-wise HPD intervals denoted as $\widehat{\mathcal X}_{j,\alpha}=(\widehat x_{j,\alpha}^-, \widehat x_{j,\alpha}^+]$ are empirically computed from the samples drawn from the posterior distribution $p(\bx|\by)$.  The expected coverage for image $\bx^*$ is then estimated over the pixels as
\begin{equation*}
\hat{\mathrm{P}}^{\mathrm{cov}}_{\alpha} = \frac 1 n \sum_{j=1}^n \mathbf 1\{x_j^* \in \widehat{\mathcal X}_{j,\alpha}\}.
\end{equation*}
A perfect calibration can thus be empirically assessed as
\begin{equation*}
\hat{\mathrm{P}}^{\mathrm{cov}}_{\alpha} \approx 1-\alpha,\ \ \forall \alpha \in [0,1].
\end{equation*}
}

{A conventional diagnosis consists in plotting the expected coverage $\hat{\mathrm{P}}^{\mathrm{cov}}_{\alpha}$ as a function of $\alpha$, for a set of discretized credibility levels $1-\alpha_\ell=\frac \ell L$ with $\ell \in \left\{1,\ldots,L\right\}$. Perfect calibration corresponds to the diagonal from $(0,0)$ to $(1,1)$, as suggested in \cref{fig:tep_afc}. By noticing that the expected coverage can be rewritten as
\begin{equation*}
\hat{\mathrm{P}}^{\mathrm{cov}}_{\alpha_\ell}  = \frac 1 n \sum_{j=1}^n \mathbf 1\{\ell_j^- \le  \ell \}
\end{equation*}
with
\begin{equation*}
\ell_j^- = \min_{\ell \in \{1,\ldots,L\}} \left\{\ell : x_j^* \in \widehat{\mathcal X}_{j,\alpha_\ell}\right\},
\end{equation*}
the expected coverage can be interpreted as the cumulative distribution function of the smallest credibility level $1-\alpha^-$ for which the HPD interval contains the true pixel value $x_j^*$. The spatial representation of the pixel-wise minimum levels $1-\alpha_{j}^-=\frac{\ell_j^-}{L}$ ($j\in \left\{1,\ldots,n\right\}$) is referred to as the coverage map. 
}

\section{Experimental settings and tuning of the hyperparameters}\label{app:parameters}
{This appendix details the strategy adopted to adjust the parameters of the compared algorithms when conducting the three considered experiments, namely denoising, deblurring, and PET reconstruction.}

\subsection{Denoiser-based sampling methods}
{For all denoiser-based sampling methods (HRLwSGS, PnP-MLA, RPnP-ULA), sampling is preceded by a deterministic burn-in phase that drives the initial state toward a high-probability region of the posterior before stochastic sampling begins. During the burn-in period, for these three methods, two quantities are annealed from an initial value to a final value:
\begin{itemize}
    \item the Langevin step size $\gamma$, from $\gamma_{\mathrm{init}}$ to $\gamma_{\mathrm{final}}$,
    \item the denoising strength $\nu$, from $\nu_{\mathrm{init}}$ to $\nu_{\mathrm{final}}$, following the schedule of \cite{Faye2024}.
\end{itemize}
Moreover, for HRLwSGS, the coupling parameter $\rho$ also follows an annealing scheme from $\rho_{\mathrm{init}}=0.9$ (shared across all settings) to $\rho_{\mathrm{final}}$. After burn-in, these parameters $\gamma$, $\rho$, and $\nu$ are frozen at their final values for the entire posterior-sampling phase, for every method.}

{For a given inversion task, prior, and Poisson noise level, the parameters $\beta$ and $\nu$ adjusting the prior are shared across all denoiser-based methods to keep comparisons fair. For HRLwSGS, the prior is reformulated so that the gradient of its regularization term matches the form of the denoiser-based regularization gradient used by PnP-MLA and RPnP-ULA.
Tables \ref{tab:params_denoising_red}, \ref{tab:params_deblurring_red} and \ref{tab:params_pet} report, for each task, method, prior, and noise level: the noise level $\alpha$, the coupling schedule $(\rho_{\mathrm{init}},\rho_{\mathrm{final}})$ (HRLwSGS only), the regularization weight $\beta$, the denoising-strength schedule $\nu=(\nu_{\mathrm{init}},\nu_{\mathrm{final}})$, and the step-size schedule $(\gamma_{\mathrm{init}},\gamma_{\mathrm{final}})$.}

\begin{table}[h]
\centering
\caption{Denoising experiment: parameter values of denoiser-based sampling methods.}
\label{tab:params_denoising_red}
\resizebox{\textwidth}{!}{\begin{tabular}{llccccc}
\toprule
& Method & $\alpha$ & $(\rho_{\mathrm{init}},\rho_{\mathrm{final}})$ & $\beta$ & $\nu=(\nu_{\mathrm{init}},\nu_{\mathrm{final}})$ & $(\gamma_{\mathrm{init}},\gamma_{\mathrm{final}})$ \\
\midrule
\multirow{6}{*}{\rotatebox{90}{RED prior}}&\multirow{2}{*}{HRLwSGS}  & 10 & $(0.9,\,10^{-3})$ & 0.55 & $(40/255,\,8.5/255)$ & $(3\times10^{-3},\,3\times10^{-5})$ \\
&                          & 40 & $(0.9,\,10^{-3})$ & 0.60 & $(40/255,\,8.5/255)$ & $(9\times10^{-4},\,9\times10^{-6})$ \\
\cline{2-7}
&\multirow{2}{*}{PnP-MLA}  & 10 & --        & 0.55 & $(40/255,\,8.5/255)$ & $(3\times10^{-3},\,3\times10^{-5})$ \\
&                          & 40 & --        & 0.60 & $(40/255,\,8.5/255)$ & $(10^{-3},\,9\times10^{-6})$ \\
\cline{2-7}
&\multirow{2}{*}{RPnP-ULA} & 10 & --        & 0.55 & $(40/255,\,8.5/255)$ & $(2\times10^{-4},\,2\times10^{-6})$ \\
&                          & 40 & --        & 0.60 & $(40/255,\,8.5/255)$ & $(10^{-4},\,10^{-6})$ \\
\midrule
\multirow{4}{*}{\rotatebox{90}{BSD prior}}&\multirow{2}{*}{HRLwSGS}  & 10 & $(0.9,\,2\times10^{-3})$ & 1.00 & $(80.5/255,\,30.5/255)$ & $(3\times10^{-3},\,3\times10^{-5})$ \\
&                          & 40 & $(0.9,\,2\times10^{-3})$ & 1.25 & $(80.5/255,\,30.5/255)$ & $(10^{-3},\,10^{-5})$ \\
\cline{2-7}
&\multirow{2}{*}{RPnP-ULA} & 10 & --               & 1.00 & $(80.5/255,\,30.5/255)$ & $(10^{-2},\,10^{-6})$ \\
&                          & 40 & --               & 1.25 & $(80.5/255,\,30.5/255)$ & $(10^{-2},\,10^{-6})$ \\
\bottomrule
\end{tabular}}
\end{table}

\begin{table}[h]
\centering
\caption{Deblurring experiment: parameter values of denoiser-based sampling methods.}
\label{tab:params_deblurring_red}
\resizebox{\textwidth}{!}{\begin{tabular}{llccccc}
\toprule
&Method & $\alpha$ & $(\rho_{\mathrm{init}},\rho_{\mathrm{final}})$ & $\beta$ & $\nu=(\nu_{\mathrm{init}},\nu_{\mathrm{final}})$ & $(\gamma_{\mathrm{init}},\gamma_{\mathrm{final}})$ \\
\midrule
\multirow{6}{*}{\rotatebox{90}{RED prior}}&\multirow{2}{*}{HRLwSGS}  & 10 & $(0.9,\,10^{-3})$ & 1.00 & $(40/255,\,15/255)$ & $(3\times10^{-3},\,2.95\times10^{-5})$ \\
&                          & 40 & $(0.9,\,10^{-3})$ & 1.10 & $(40/255,\,15/255)$ & $(3\times10^{-3},\,8.5\times10^{-6})$ \\
\cline{2-7}
&\multirow{2}{*}{PnP-MLA}  & 10 & --        & 1.00 & $(40/255,\,15/255)$ & $(3\times10^{-3},\,1.25\times10^{-5})$ \\
&                          & 40 & --        & 1.10 & $(40/255,\,15/255)$ & $(3\times10^{-3},\,9\times10^{-6})$ \\
\cline{2-7}
&\multirow{2}{*}{RPnP-ULA} & 10 & --        & 1.00 & $(40/255,\,15/255)$ & $(5\times10^{-3},\,1.5\times10^{-6})$ \\
&                          & 40 & --        & 1.10 & $(40/255,\,15/255)$ & $(5\times10^{-3},\,7\times10^{-7})$ \\
\midrule
\multirow{4}{*}{\rotatebox{90}{BSD prior}}&\multirow{2}{*}{HRLwSGS}  & 10 & $(0.9,\,2\times10^{-3})$ & 1.25 & $(80.5/255,\,38.5/255)$ & $(0.3,\,2\times10^{-5})$ \\
&                          & 40 & $(0.9,\,2\times10^{-3})$ & 1.25 & $(80.5/255,\,38.5/255)$ & $(0.3,\,1.55\times10^{-5})$ \\
\cline{2-7}
&\multirow{2}{*}{RPnP-ULA} & 10 & --               & 1.25 & $(80.5/255,\,38.5/255)$ & $(10^{-3},\,5\times10^{-6})$ \\
&                          & 40 & --               & 1.25 & $(80.5/255,\,38.5/255)$ & $(10^{-3},\,3\times10^{-6})$ \\
\bottomrule
\end{tabular}}
\end{table}

\begin{table}[h]
\centering
\caption{PET reconstruction experiment: parameter values of denoiser-based sampling methods.}
\label{tab:params_pet}
\begin{tabular}{llcccc}
\toprule
 & Method & $\beta$ & $\rho$ & $\nu=(\nu_{\mathrm{init}},\nu_{\mathrm{final}})$ & $(\gamma_{\mathrm{init}},\gamma_{\mathrm{final}})$ \\
\midrule
\multirow{3}{*}{\rotatebox{90}{RED}} & HRLwSGS & 1.0 & $10^{-2}$ & $(25/255,\,10/255)$ & $(0.01,\,3.25\times10^{-5})$ \\
                     & RPnP-ULA & 1.0 & --        & $(25/255,\,10/255)$ & $(0.01,\,2\times10^{-7})$ \\
                     & PnP-MLA  & 1.0 & --        & $(25/255,\,10/255)$ & $(0.01,\,1.5\times10^{-5})$ \\
\midrule
\multirow{2}{*}{\rotatebox{90}{BSD}} & HRLwSGS  & 1.5 & $10^{-2}$ & $(80.5/255,\,57.02/255)$ & $(0.01,\,3\times10^{-5})$ \\
                     & RPnP-ULA & 1.5 & --        & $(80.5/255,\,57.02/255)$ & $(0.01,\,4.5\times10^{-7})$ \\
\bottomrule
\end{tabular}
\end{table}
\subsection{TV-based and diffusion baselines}
{Unlike the denoiser-based methods above, TV-PIDAL, TV-SPA, and Prox-DiffPIR use no burn-in annealing: all parameters are fixed for the entire run, and depend only on the Poisson noise level $\alpha$ and the inversion task. 
The algorithmic parameters have been adjusted as follows:
\begin{itemize}
    \item TV-PIDAL:
    \begin{itemize}
        \item Denoising: 1000 iterations, $(\tau,\mu) = (6,0.8)$ for $\alpha=10$ and  $(\tau,\mu) = (8,0.8)$ for $\alpha=40$.
        \item Deblurring: 1000 iterations, $(\tau,\mu) = (1.75,0.8)$ for $\alpha=10$ and $(\tau,\mu) =(2.75,0.8)$ for $\alpha=40$.
        \item PET: 1000 iterations, $\tau=4.84$, $\mu=1.95$.
    \end{itemize}
    \item {TV-SPA}:
    \begin{itemize}
        \item Denoising: $\beta=6$ ($\alpha=10$) or $\beta=8$ ($\alpha=40$), with $\rho_1=0.1$, $\rho_2=\rho_3=0.02$, $\alpha_1=\alpha_2=\alpha_3=0.01$, $\lambda_1=\lambda_2=\lambda_3=7\times10^{-5}$, $\gamma_1=\gamma_2=\gamma_3=7\times10^{-5}$ for both noise levels.
        \item Deblurring: $\beta=6$ for both noise levels, with $\rho_1=0.1$, $\rho_2=\rho_3=0.02$, $\alpha_1=\alpha_2=\alpha_3=0.01$, $\lambda_1=\lambda_2=\lambda_3=6\times10^{-5}$, $\gamma_1=\gamma_2=\gamma_3=5\times10^{-5}$ for $\alpha=10$ and $\rho_1=0.1$, $\rho_2=\rho_3=0.01$, $\alpha_1=\alpha_2=\alpha_3=0.001$, $\lambda_1=\lambda_2=\lambda_3=7\times10^{-5}$, $\gamma_1=\gamma_2=\gamma_3=7\times10^{-5}$ for $\alpha=40$.
        \item PET : $\rho_1=0.8$, $\rho_2=\rho_3=0.01$, $\alpha_1=\alpha_2=\alpha_3=0.001$, $\lambda_1=\lambda_2=\lambda_3=9\times10^{-5}$, $\gamma_1=\gamma_2=\gamma_3=6.5\times10^{-5}$, and $\beta=30$.
    \end{itemize}
    \item {Prox-DiffPIR}:  A single configuration is used across both Poisson noise levels and all tasks: $100$ iterations, $\lambda=4$, $\zeta=0.8$.
\end{itemize}}

\section{Complementary experimental results}\label{app:complementary_results}

{This section presents some complementary visual results related to the experiments described in \cref{sec:experiments}.} {First, \cref{fig:imagenet_denoising_1_complementary_results} depicts the results obtained by TV-SPA and TV-PIDAL for the denoising task. As discussed in \cref{subsec:denoising}, the quantitative evaluation in \cref{tab:imagenet_denoising} highlights a strong dependence on the noise level. For $\alpha = 10$, ProxDiffPIR achieves the best overall performance, with the highest PSNR, SSIM, and lowest LPIPS values, while TV-based approaches provide competitive but lower-quality reconstructions. However, under the moderate noise regime ($\alpha = 40$), TV-PIDAL demonstrates strong robustness, achieving 29.102 dB PSNR and 0.8051 SSIM. Although it remains below the best-performing approaches such as ProxDiffPIR and PnP-MLA, it significantly outperforms BSD-based prior, including BSD-RPnP-ULA  and BSD-HRLwSGS.}

\begin{figure}[h!]
\centering
\setlength{\tabcolsep}{0pt}
\renewcommand{\arraystretch}{1}
\begin{tabular}{L  M  M  M  M}
& \makebox[\imw]{\scriptsize $\alpha=40$} &
  \makebox[\imw]{\scriptsize $\alpha=40$} &
  \makebox[\imw]{\scriptsize $\alpha=10$} &
  \makebox[\imw]{\scriptsize $\alpha=10$} \\
\rotlabel{Ground Truth} &
\begin{tikzpicture}[zoomboxarray]
    \node [image node] {\sqimg{images/denoising/alpha-40/347.pdf}};
	\zoombox{0.3, 0.65}
\end{tikzpicture} &
\begin{tikzpicture}[zoomboxarray]
    \node [image node] {\sqimg{images/denoising/alpha-40/66.pdf}};
	\zoombox{0.88, 0.85}
\end{tikzpicture} &
\begin{tikzpicture}[zoomboxarray]
    \node [image node] {\sqimg{images/denoising/alpha-10/237.pdf}};
	\zoombox{0.45, 0.8}
\end{tikzpicture} &
\begin{tikzpicture}[zoomboxarray]
    \node [image node] {\sqimg{images/denoising/alpha-10/433.pdf}};
	\zoombox{0.6, 0.72}
\end{tikzpicture}\\[-2pt]
\rotlabel{Observation} &
\begin{tikzpicture}[zoomboxarray]
    \node [image node]{\sqimg{images/denoising/alpha-40/347_obs.pdf}};
	\zoombox{0.3, 0.65}
\end{tikzpicture} &
\begin{tikzpicture}[zoomboxarray]
    \node [image node]{\sqimg{images/denoising/alpha-40/66_obs.pdf} };
	\zoombox{0.88, 0.85}
\end{tikzpicture} &
\begin{tikzpicture}[zoomboxarray]
    \node [image node] {\sqimg{images/denoising/alpha-10/237_obs.pdf}};
	\zoombox{0.45, 0.8}
\end{tikzpicture} &
\begin{tikzpicture}[zoomboxarray]
    \node [image node] {\sqimg{images/denoising/alpha-10/433_obs.pdf}};
	\zoombox{0.6, 0.72}
\end{tikzpicture}\\[-2pt]
\rotlabel{TV-SPA} &
\begin{tikzpicture}[zoomboxarray]
    \node [image node]{\sqimg{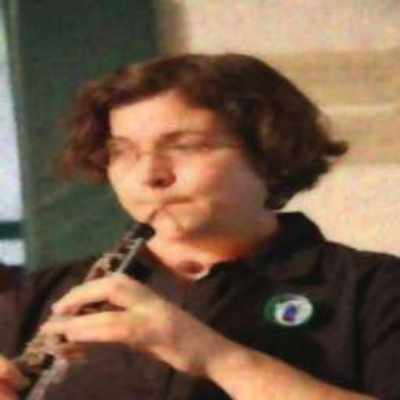}};
	\zoombox{0.3, 0.65}
\end{tikzpicture} &
\begin{tikzpicture}[zoomboxarray]
    \node [image node]{\sqimg{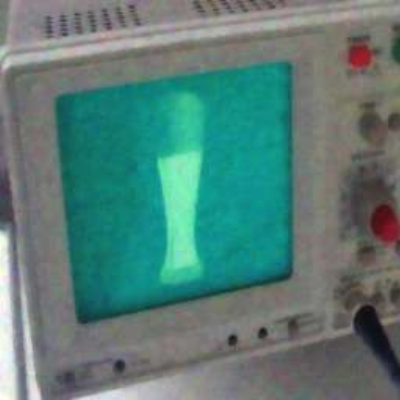}};
	\zoombox{0.88, 0.85}
\end{tikzpicture} &
\begin{tikzpicture}[zoomboxarray]
    \node [image node] {\sqimg{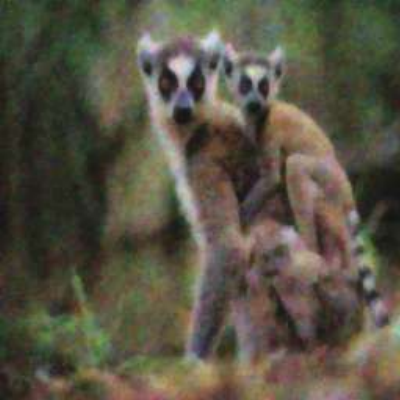} };
	\zoombox{0.45, 0.8}
\end{tikzpicture} &
\begin{tikzpicture}[zoomboxarray]
    \node[image node]{\sqimg{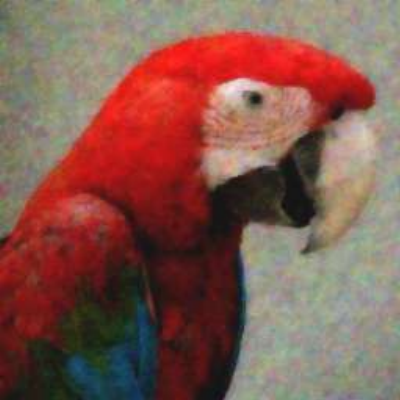}  };
	\zoombox{0.6, 0.72}
\end{tikzpicture}\\[-2pt]
\rotlabel{TV-PIDAL} &
\begin{tikzpicture}[zoomboxarray]
    \node [image node]{\sqimg{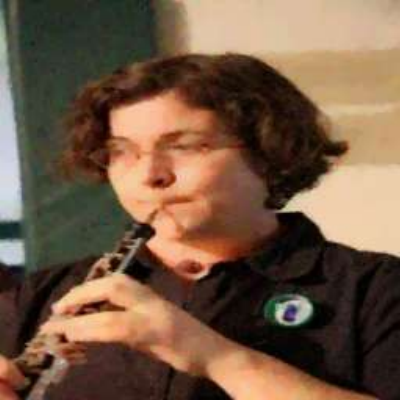}};
	\zoombox{0.3, 0.65}
\end{tikzpicture} &
\begin{tikzpicture}[zoomboxarray]
    \node [image node]{\sqimg{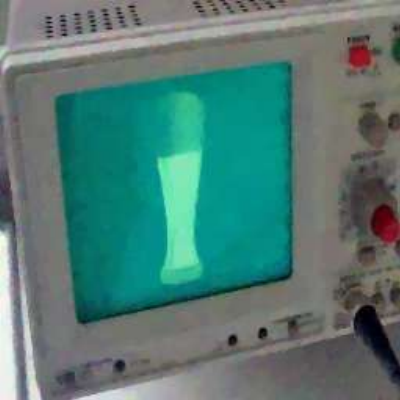}};
	\zoombox{0.88, 0.85}
\end{tikzpicture} &
\begin{tikzpicture}[zoomboxarray]
    \node [image node] {\sqimg{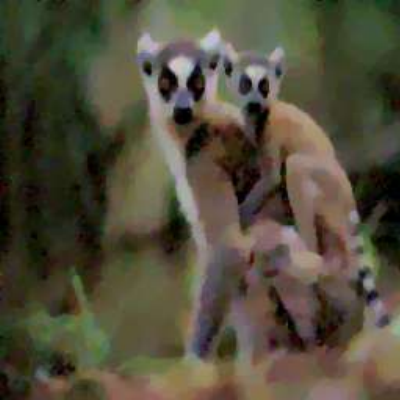}};
	\zoombox{0.45, 0.8}
\end{tikzpicture} &
\begin{tikzpicture}[zoomboxarray]
    \node[image node]{\sqimg{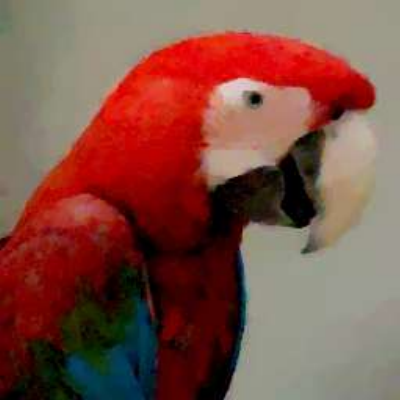} };
	\zoombox{0.6, 0.72}
\end{tikzpicture}
\end{tabular}
\caption{Denoising experiment: ground truth (1st row) and observed (2nd row) images, images reconstructed by TV-SPA (3rd row) and TV-PIDAL (4rd row)  for $\alpha = 40$ (first two columns) and $\alpha = 40$ (last two columns).}
\label{fig:imagenet_denoising_1_complementary_results}
\end{figure}

{Then, \cref{fig:imagenet_deblurring_compementary_results} depicts the results obtained by TV-SPA and TV-PIDAL when performing deblurring. Again, as discussed in \cref{subsec:deconvolution}, the results provided by both TV-based algorithms are of significantly lower quality for the two noise levels.}

\begin{figure}[h!]
\centering
\setlength{\tabcolsep}{0pt}
\renewcommand{\arraystretch}{1}
\begin{tabular}{L  M  M  M  M}
& \makebox[\imw]{\scriptsize $\alpha=40$} &
  \makebox[\imw]{\scriptsize $\alpha=40$} &
  \makebox[\imw]{\scriptsize $\alpha=10$} &
  \makebox[\imw]{\scriptsize $\alpha=10$} \\
\rotlabel{Ground Truth} &
\begin{tikzpicture}[zoomboxarray]
    \node [image node] {\sqimg{images/alpha40/deblurring/34.pdf}};
	 \zoombox{0.54, 0.74}
\end{tikzpicture} &
\begin{tikzpicture}[zoomboxarray]
    \node [image node] {\sqimg{images/alpha40/deblurring/51.pdf}};
	\zoombox{0.58, 0.55}
\end{tikzpicture} &
\begin{tikzpicture}[zoomboxarray]
    \node [image node] {\sqimg{images/alpha10/Deblurring/10.pdf}};
	\zoombox{0.3, 0.75}
\end{tikzpicture} &
\begin{tikzpicture}[zoomboxarray]
    \node [image node] {\sqimg{images/alpha10/Deblurring/1.pdf}};
	\zoombox{0.52, 0.75}
\end{tikzpicture}\\[-2pt]
\rotlabel{Observation} &
\begin{tikzpicture}[zoomboxarray]
    \node [image node]{\sqimg{images/alpha40/deblurring/34_obs.pdf}};
	\zoombox{0.54, 0.74}
\end{tikzpicture} &
\begin{tikzpicture}[zoomboxarray]
    \node [image node]{\sqimg{images/alpha40/deblurring/51_obs.pdf} };
	\zoombox{0.58, 0.55}
\end{tikzpicture} &
\begin{tikzpicture}[zoomboxarray]
    \node [image node] {\sqimg{images/alpha10/Deblurring/10_obs.pdf}};
	\zoombox{0.3, 0.75}
\end{tikzpicture} &
\begin{tikzpicture}[zoomboxarray]
    \node [image node] {\sqimg{images/alpha10/Deblurring/1_obs.pdf}};
	\zoombox{0.52, 0.75}
\end{tikzpicture}\\[-2pt]
\rotlabel{TV-SPA} &
\begin{tikzpicture}[zoomboxarray]
    \node [image node]{\sqimg{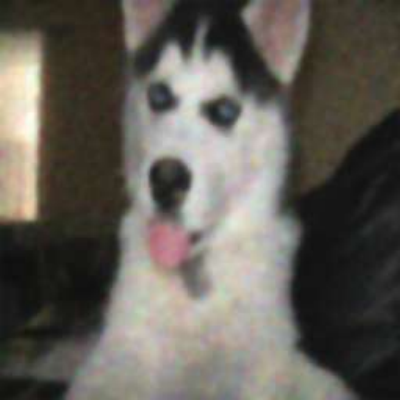}};
	\zoombox{0.54, 0.74}
\end{tikzpicture} &
\begin{tikzpicture}[zoomboxarray]
    \node [image node]{\sqimg{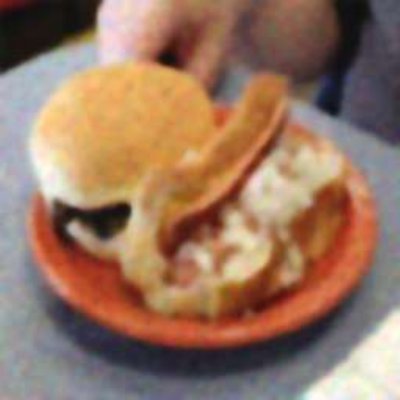}};
	\zoombox{0.58, 0.55}
\end{tikzpicture} &
\begin{tikzpicture}[zoomboxarray]
    \node [image node] {\sqimg{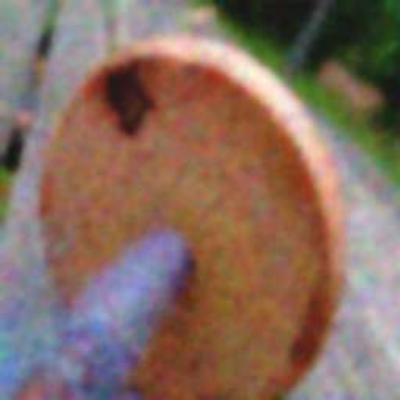} };
	\zoombox{0.3, 0.75}
\end{tikzpicture} &
\begin{tikzpicture}[zoomboxarray]
    \node[image node]{\sqimg{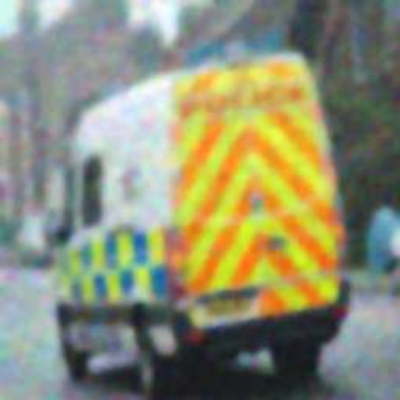}  };
	\zoombox{0.52, 0.75}
\end{tikzpicture}\\[-2pt]
\rotlabel{TV-PIDAL} &
\begin{tikzpicture}[zoomboxarray]
    \node [image node]{\sqimg{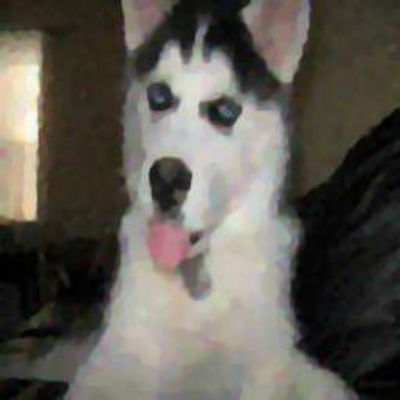}};
	\zoombox{0.54, 0.74}
\end{tikzpicture} &
\begin{tikzpicture}[zoomboxarray]
    \node [image node]{\sqimg{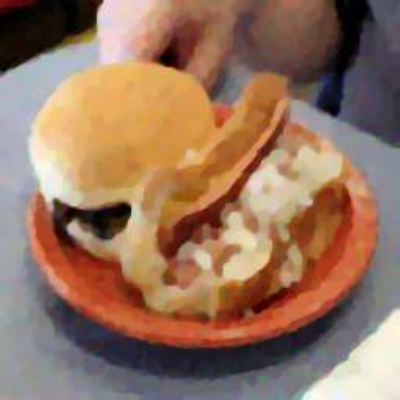}};
	\zoombox{0.58, 0.55}
\end{tikzpicture} &
\begin{tikzpicture}[zoomboxarray]
    \node [image node] {\sqimg{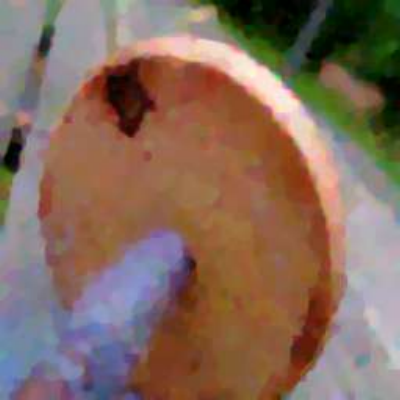}};
	\zoombox{0.3, 0.75}
\end{tikzpicture} &
\begin{tikzpicture}[zoomboxarray]
    \node[image node]{\sqimg{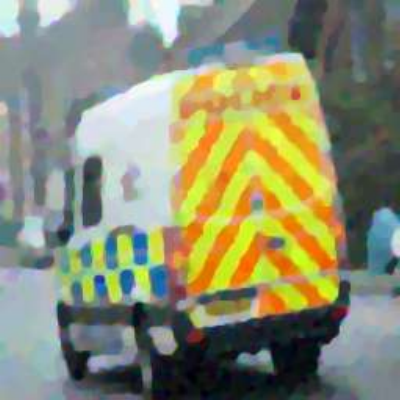} };
	\zoombox{0.52, 0.75}
\end{tikzpicture}
\end{tabular}
\caption{Deblurring experiment: ground truth (1st row) and observed (2nd row) images, images reconstructed by TV-SPA (3rd row) and TV-PIDAL (4rd row)  for $\alpha = 40$ (first two columns) and $\alpha = 40$ (last two columns).}
\label{fig:imagenet_deblurring_compementary_results}
\end{figure}

\section*{Acknowledgments}
The authors would like to thank Prof. Pierre Chainais (Univ. Lille, France) for his valuable feedback on this work. They also thank the Associate Editor and the anonymous reviewers for their insightful comments, which helped to improve this manuscript.

\bibliographystyle{ieeetran}
\bibliography{strings_all_ref,biblio}

\begin{thebibliography}{10}
\providecommand{\url}[1]{#1}
\csname url@samestyle\endcsname
\providecommand{\newblock}{\relax}
\providecommand{\bibinfo}[2]{#2}
\providecommand{\BIBentrySTDinterwordspacing}{\spaceskip=0pt\relax}
\providecommand{\BIBentryALTinterwordstretchfactor}{4}
\providecommand{\BIBentryALTinterwordspacing}{\spaceskip=\fontdimen2\font plus
\BIBentryALTinterwordstretchfactor\fontdimen3\font minus
  \fontdimen4\font\relax}
\providecommand{\BIBforeignlanguage}[2]{{%
\expandafter\ifx\csname l@#1\endcsname\relax
\typeout{** WARNING: IEEEtran.bst: No hyphenation pattern has been}%
\typeout{** loaded for the language `#1'. Using the pattern for}%
\typeout{** the default language instead.}%
\else
\language=\csname l@#1\endcsname
\fi
#2}}
\providecommand{\BIBdecl}{\relax}
\BIBdecl

\bibitem{hanisch1994restoration}
R.~J. Hanisch and R.~L. White, ``The restoration of {HST} images and spectra
  {II},'' \emph{Space Telescope Science Institute, Baltimore}, 1994.

\bibitem{shepp2007maximum}
L.~A. Shepp and Y.~Vardi, ``Maximum likelihood reconstruction for emission
  tomography,'' \emph{IEEE Trans. Med. Imag.}, vol.~1, no.~2, pp. 113--122,
  2007.

\bibitem{agard1983three}
D.~A. Agard and J.~W. Sedat, ``Three-dimensional architecture of a polytene
  nucleus,'' \emph{Nature}, vol. 302, no. 5910, pp. 676--681, 1983.

\bibitem{sarder2006deconvolution}
P.~Sarder and A.~Nehorai, ``Deconvolution methods for 3-d fluorescence
  microscopy images,'' \emph{IEEE signal processing magazine}, vol.~23, no.~3,
  pp. 32--45, 2006.

\bibitem{janesick2007photon}
J.~R. Janesick, \emph{Photon transfer}.\hskip 1em plus 0.5em minus 0.4em\relax
  SPIE press, 2007, no. PUBDB-2021-04195.

\bibitem{figueiredo2010restoration}
M.~A. Figueiredo and J.~M. Bioucas-Dias, ``Restoration of {P}oissonian images
  using alternating direction optimization,'' \emph{IEEE Trans. Image
  Process.}, vol.~19, no.~12, pp. 3133--3145, 2010.

\bibitem{afonso2010augmented}
M.~V. Afonso, J.~M. Bioucas-Dias, and M.~A. Figueiredo, ``An augmented
  {L}agrangian approach to the constrained optimization formulation of imaging
  inverse problems,'' \emph{IEEE Trans. Image Process.}, vol.~20, no.~3, pp.
  681--695, 2010.

\bibitem{rudin1992nonlinear}
L.~I. Rudin, S.~Osher, and E.~Fatemi, ``Nonlinear total variation based noise
  removal algorithms,'' \emph{Physica D: nonlinear phenomena}, vol.~60, no.
  1-4, pp. 259--268, 1992.

\bibitem{bauschke2018regularizing}
H.~H. Bauschke, M.~N. Dao, and S.~B. Lindstrom, ``Regularizing with
  {B}regman--{M}oreau envelopes,'' \emph{SIAM Journal on Optimization},
  vol.~28, no.~4, pp. 3208--3228, 2018.

\bibitem{lee2010hierarchical}
A.~Lee, F.~Caron, A.~Doucet, and C.~Holmes, ``A hierarchical {B}ayesian
  framework for constructing sparsity-inducing priors,'' \emph{arXiv preprint
  arXiv:1009.1914}, 2010.

\bibitem{park2008bayesian}
T.~Park and G.~Casella, ``The {B}ayesian lasso,'' \emph{J. Amer. Stat. Assoc.},
  vol. 103, no. 482, pp. 681--686, 2008.

\bibitem{dobigeon2009hierarchical}
N.~Dobigeon, A.~O. Hero, and J.-Y. Tourneret, ``Hierarchical {B}ayesian sparse
  image reconstruction with application to {MRFM},'' \emph{IEEE Trans. Image
  Process.}, vol.~18, no.~9, pp. 2059--2070, 2009.

\bibitem{venkatakrishnan2013plug}
S.~V. Venkatakrishnan, C.~A. Bouman, and B.~Wohlberg, ``Plug-and-play priors
  for model based reconstruction,'' in \emph{Proc. IEEE Global Conf. Signal
  Info. Process. (GlobalSIP)}.\hskip 1em plus 0.5em minus 0.4em\relax IEEE,
  2013, pp. 945--948.

\bibitem{romano2017little}
Y.~Romano, M.~Elad, and P.~Milanfar, ``The little engine that could:
  Regularization by denoising ({RED}),'' \emph{SIAM J. Imag. Sci.}, vol.~10,
  no.~4, pp. 1804--1844, 2017.

\bibitem{dabov2007image}
K.~Dabov, A.~Foi, V.~Katkovnik, and K.~Egiazarian, ``Image denoising by sparse
  3-{D} transform-domain collaborative filtering,'' \emph{IEEE Trans. Image
  Process.}, vol.~16, no.~8, pp. 2080--2095, 2007.

\bibitem{marais2017proximal}
W.~Marais and R.~Willett, ``Proximal-gradient methods for {P}oisson image
  reconstruction with {BM3D}-based regularization,'' in \emph{Proc. IEEE Int.
  Workshop Comput. Adv. Multi-Sensor Adaptive Process. (CAMSAP)}.\hskip 1em
  plus 0.5em minus 0.4em\relax IEEE, 2017, pp. 1--5.

\bibitem{milanfar2024denoising}
P.~Milanfar and M.~Delbracio, ``Denoising: A powerful building-block for
  imaging, inverse problems, and machine learning,'' \emph{Philosophical
  Transactions A}, vol. 383, no. 2299, p. 20240326, 2025.

\bibitem{hurault2023convergent}
S.~Hurault, U.~Kamilov, A.~Leclaire, and N.~Papadakis, ``Convergent {B}regman
  plug-and-play image restoration for {P}oisson inverse problems,'' \emph{Adv.
  in Neural Information Process. Systems (NeurIPS)}, vol.~36, pp.
  27\,251--27\,280, 2023.

\bibitem{Vono2019icassp}
M.~Vono, N.~Dobigeon, and P.~Chainais, ``{B}ayesian image restoration under
  {P}oisson noise and log-concave prior,'' in \emph{Proc. IEEE Int. Conf.
  Acoust., Speech and Signal Process. (ICASSP)}, Brighton, U.K., April 2019.

\bibitem{vono2020asymptotically}
------, ``Asymptotically exact data augmentation: Models, properties, and
  algorithms,'' \emph{J. Comput. Graph. Stat.}, vol.~30, no.~2, pp. 335--348,
  2020.

\bibitem{melidonis2023efficient}
S.~Melidonis, P.~Dobson, Y.~Altmann, M.~Pereyra, and K.~Zygalakis, ``Efficient
  {B}ayesian computation for low-photon imaging problems,'' \emph{SIAM J. Imag.
  Sci.}, vol.~16, no.~3, pp. 1195--1234, 2023.

\bibitem{klatzer2025efficient}
T.~Klatzer, S.~Melidonis, M.~Pereyra, and K.~C. Zygalakis, ``Efficient
  {B}ayesian computation using plug-and-play priors for {P}oisson inverse
  problems,'' \emph{SIAM J. Imag. Sci.}, vol.~19, no.~2, pp. 1325--1363, June
  2026.

\bibitem{song2020score}
Y.~Song, J.~Sohl-Dickstein, D.~P. Kingma, A.~Kumar, S.~Ermon, and B.~Poole,
  ``Score-based generative modeling through stochastic differential
  equations,'' in \emph{Proc. IEEE Int. Conf. Learn. Represent. (ICLR)}, 2021.

\bibitem{ho2020denoising}
J.~Ho, A.~Jain, and P.~Abbeel, ``Denoising diffusion probabilistic models,'' in
  \emph{Adv. in Neural Information Process. Systems (NeurIPS)}, vol.~33, 2020,
  pp. 6840--6851.

\bibitem{chung2022diffusion}
H.~Chung, J.~Kim, M.~T. Mccann, M.~L. Klasky, and J.~C. Ye, ``Diffusion
  posterior sampling for general noisy inverse problems,'' in \emph{Proc. IEEE
  Int. Conf. Learn. Represent. (ICLR)}, 2023.

\bibitem{singh2023score}
I.~R. Singh, A.~Denker, R.~Barbano, {\v{Z}}.~Kereta, B.~Jin, K.~Thielemans,
  P.~Maass, and S.~Arridge, ``Score-based generative models for {PET} image
  reconstruction,'' \emph{J. Machine Learning Biomed. Imag.}, vol.~2, pp.
  547--585, 2024.

\bibitem{webber2025likelihood}
G.~Webber, Y.~Mizuno, O.~D. Howes, A.~Hammers, A.~P. King, and A.~J. Reader,
  ``Likelihood-scheduled score-based generative modeling for fully {3D PET}
  image reconstruction,'' \emph{IEEE Trans. Med. Imag.}, 2025.

\bibitem{melidonis2025score}
S.~Melidonis, Y.~Xi, K.~C. Zygalakis, Y.~Altmann, and M.~Pereyra, ``Score-based
  denoising diffusion models for photon-starved image restoration problems,''
  \emph{Transactions on Machine Learning Research}, 2025.

\bibitem{zhang2020wasserstein}
K.~S. Zhang, G.~Peyr{\'e}, J.~Fadili, and M.~Pereyra, ``Wasserstein control of
  mirror {L}angevin {M}onte {C}arlo,'' in \emph{Proc. Conf. Learning Theory
  (COLT)}.\hskip 1em plus 0.5em minus 0.4em\relax PMLR, 2020, pp. 3814--3841.

\bibitem{vono2019split}
M.~Vono, N.~Dobigeon, and P.~Chainais, ``Split-and-augmented {G}ibbs sampler --
  {A}pplication to large-scale inference problems,'' \emph{IEEE Trans. Signal
  Process.}, vol.~67, no.~6, pp. 1648--1661, 2019.

\bibitem{bregman1967relaxation}
L.~M. {B}regman, ``The relaxation method of finding the common point of convex
  sets and its application to the solution of problems in convex programming,''
  \emph{USSR computational mathematics and mathematical physics}, vol.~7,
  no.~3, pp. 200--217, 1967.

\bibitem{eggermont1993maximum}
P.~P. Eggermont, ``Maximum entropy regularization for fredholm integral
  equations of the first kind,'' \emph{SIAM Journal on Mathematical Analysis},
  vol.~24, no.~6, pp. 1557--1576, 1993.

\bibitem{burger2004convergence}
M.~Burger and S.~Osher, ``Convergence rates of convex variational
  regularization,'' \emph{Inverse problems}, vol.~20, no.~5, p. 1411, 2004.

\bibitem{burg1975maximum}
J.~P. Burg, \emph{Maximum entropy spectral analysis.}\hskip 1em plus 0.5em
  minus 0.4em\relax Stanford University, 1975.

\bibitem{banerjee2005clustering}
A.~Banerjee, S.~Merugu, I.~S. Dhillon, and J.~Ghosh, ``Clustering with
  {B}regman divergences,'' \emph{J. Mach. Learning Research}, vol.~6, no. Oct,
  pp. 1705--1749, 2005.

\bibitem{fevotte2009nonnegative}
C.~F{\'e}votte, N.~Bertin, and J.-L. Durrieu, ``Nonnegative matrix
  factorization with the {I}takura-{S}aito divergence: With application to
  music analysis,'' \emph{Neural computation}, vol.~21, no.~3, pp. 793--830,
  2009.

\bibitem{cavalcanti2019factor}
Y.~C. Cavalcanti, T.~Oberlin, N.~Dobigeon, C.~F{\'e}votte, S.~Stute, M.-J.
  Ribeiro, and C.~Tauber, ``Factor analysis of dynamic {PET} images: beyond
  {G}aussian noise,'' \emph{IEEE Trans. Med. Imag.}, vol.~38, no.~9, pp.
  2231--2241, 2019.

\bibitem{vono2018sparse}
M.~Vono, N.~Dobigeon, and P.~Chainais, ``Sparse {B}ayesian binary logistic
  regression using the split-and-augmented {G}ibbs sampler,'' in \emph{Proc.
  IEEE Workshop Mach. Learning for Signal Process. (MLSP)}, 2018, pp. 1--6.

\bibitem{Coeurdoux2024pnp}
F.~Coeurdoux, N.~Dobigeon, and P.~Chainais, ``Plug-and-play split {G}ibbs
  sampler: embedding deep generative priors in {B}ayesian inference,''
  \emph{IEEE Trans. Image Process.}, vol.~33, pp. 3496--3507, May 2024.

\bibitem{wu2024principled}
Z.~Wu, Y.~Sun, Y.~Chen, B.~Zhang, Y.~Yue, and K.~Bouman, ``Principled
  probabilistic imaging using diffusion models as plug-and-play priors,'' in
  \emph{Adv. in Neural Information Process. Systems (NeurIPS)}, vol.~37, 2024,
  pp. 118\,389--118\,427.

\bibitem{Sun2024provable}
Y.~Sun, Z.~Wu, Y.~Chen, B.~T. Feng, and K.~L. Bouman, ``Provable probabilistic
  imaging using score-based generative priors,'' \emph{IEEE Trans. Comput.
  Imag.}, 2024.

\bibitem{Faye2024}
E.~C. Faye, M.~D. Fall, and N.~Dobigeon, ``Regularization by denoising:
  {B}ayesian model and {L}angevin-within-split {G}ibbs sampling,'' \emph{IEEE
  Trans. Image Process.}, vol.~34, pp. 221--234, Jan. 2025.

\bibitem{Geman1995nonlinear}
D.~Geman and C.~Yang, ``Nonlinear image recovery with half-quadratic
  regularization,'' \emph{IEEE Trans. Image Process.}, vol.~4, no.~7, pp.
  932--946, 1995.

\bibitem{vono2022high}
M.~Vono, N.~Dobigeon, and P.~Chainais, ``High-dimensional {G}aussian sampling:
  a review and a unifying approach based on a stochastic proximal point
  algorithm,'' \emph{SIAM Rev.}, vol.~64, no.~1, pp. 3--56, 2022.

\bibitem{plassier2021dg}
V.~Plassier, M.~Vono, A.~Durmus, and E.~Moulines, ``{DG-LMC}: a turn-key and
  scalable synchronous distributed {MCMC} algorithm via {L}angevin {M}onte
  {C}arlo within {G}ibbs,'' in \emph{Proc. Int. Conf. Machine Learning
  (ICML)}.\hskip 1em plus 0.5em minus 0.4em\relax PMLR, 2021, pp. 8577--8587.

\bibitem{vono2019bayesian}
M.~Vono, N.~Dobigeon, and P.~Chainais, ``Bayesian image restoration under
  {P}oisson noise and log-concave prior,'' in \emph{Proc. IEEE Int. Conf.
  Acoust., Speech and Signal Process. (ICASSP)}.\hskip 1em plus 0.5em minus
  0.4em\relax IEEE, 2019, pp. 1712--1716.

\bibitem{eldaly2026bayesian}
A.~K. Eldaly, M.~Figini, and D.~C. Alexander, ``Bayesian uncertainty-aware
  {MRI} reconstruction,'' in \emph{Proc. IEEE Int. Conf. Acoust., Speech and
  Signal Process. (ICASSP)}, 2026.

\bibitem{durmus2018efficient}
A.~Durmus, E.~Moulines, and M.~Pereyra, ``Efficient {B}ayesian computation by
  proximal {M}arkov chain {M}onte {C}arlo: when {L}angevin meets {M}oreau,''
  \emph{SIAM J. Imag. Sci.}, vol.~11, no.~1, pp. 473--506, 2018.

\bibitem{pereyra2023split}
M.~Pereyra, L.~A. Vargas-Mieles, and K.~C. Zygalakis, ``The split {G}ibbs
  sampler revisited: improvements to its algorithmic structure and augmented
  target distribution,'' \emph{SIAM J. Imag. Sci.}, vol.~16, no.~4, pp.
  2040--2071, 2023.

\bibitem{bouman2023generative}
C.~A. Bouman and G.~T. Buzzard, ``Generative plug and play: Posterior sampling
  for inverse problems,'' in \emph{Ann. Allerton Conf. Comm. Control Computing
  (Allerton)}, 2023, pp. 1--7.

\bibitem{setzer2010deblurring}
S.~Setzer, G.~Steidl, and T.~Teuber, ``Deblurring {P}oissonian images by split
  {B}regman techniques,'' \emph{J. Visual Communication Image Representation},
  vol.~21, no.~3, pp. 193--199, 2010.

\bibitem{carlavan2011sparse}
M.~Carlavan and L.~Blanc-F{\'e}raud, ``Sparse poisson noisy image deblurring,''
  \emph{IEEE Trans. Image Process.}, vol.~21, no.~4, pp. 1834--1846, 2011.

\bibitem{vardi1982maximum}
Y.~Vardi and L.~Shepp, ``Maximum likelihood reconstruction for emission
  tomography,'' \emph{IEEE Trans. Med. Imag.}, vol. 1001, no.~2, 1982.

\bibitem{lange1984reconstruction}
K.~Lange, R.~Carson \emph{et~al.}, ``{EM} reconstruction algorithms for
  emission and transmission tomography,'' \emph{J. Comput. Assist. Tomogr.},
  vol.~8, no.~2, pp. 306--16, 1984.

\bibitem{filipovic2018pet}
M.~Filipovi{\'c}, E.~Barat, T.~Dautremer, C.~Comtat, and S.~Stute, ``{PET}
  reconstruction of the posterior image probability, including multimodal
  images,'' \emph{IEEE Trans. Med. Imag.}, vol.~38, no.~7, pp. 1643--1654,
  2018.

\bibitem{goncharov2023nonparametric}
F.~Goncharov, {\'E}.~Barat, and T.~Dautremer, ``Nonparametric posterior
  learning for emission tomography,'' \emph{SIAM/ASA J. Uncertainty
  Quantification}, vol.~11, no.~2, pp. 452--479, 2023.

\bibitem{clerc2026reinterpreting}
A.~Clerc, S.~Martin, N.~Papadakis, and G.~Steidl, ``Reinterpreting {EMML} as
  mirror descent for constrained maximum likelihood estimation,'' \emph{arXiv
  preprint arXiv:2602.13063}, 2026.

\bibitem{Fall2019_IJB}
M.~D. Fall, ``Bayesian {N}onparametrics and {B}iostatistics: The case of {PET}
  imaging,'' \emph{The International Journal of Biostatistics}, vol.~15, no.~2,
  2019.

\bibitem{Fall2013_ICIP}
M.~D. Fall, {\'E}.~Barat, C.~Comtat, T.~Dautremer, T.~Montagu, and S.~Stute,
  ``Dynamic and clinical {PET} data reconstruction: A nonparametric {B}ayesian
  approach,'' in \emph{Proc. IEEE Int. Conf. Image Process. (ICIP)}.\hskip 1em
  plus 0.5em minus 0.4em\relax IEEE, 2013, pp. 345--349.

\bibitem{Fall2011_ICIP}
M.~D. Fall, {\'E}.~Barat, C.~Comtat, T.~Dautremer, T.~Montagu, and
  A.~Mohammad-Djafari, ``A discrete-continuous {B}ayesian model for emission
  tomography,'' in \emph{Proc. IEEE Int. Conf. Image Process. (ICIP)}.\hskip
  1em plus 0.5em minus 0.4em\relax IEEE, 2011, pp. 1373--1376.

\bibitem{sitek2010reconstruction}
A.~Sitek, ``Reconstruction of emission tomography data using origin
  ensembles,'' \emph{IEEE Trans. Med. Imag.}, vol.~30, no.~4, pp. 946--956,
  2010.

\bibitem{feller1971introduction}
W.~Feller, \emph{An introduction to probability theory and its applications,
  Volume 2}, 2nd~ed.\hskip 1em plus 0.5em minus 0.4em\relax John Wiley \& Sons,
  1971, vol.~2.

\bibitem{van2001art}
D.~A. Van~Dyk and X.-L. Meng, ``The art of data augmentation,'' \emph{J.
  Comput. Graph. Stat.}, vol.~10, no.~1, pp. 1--50, 2001.

\bibitem{Perez2017}
M.-E. P{\'e}rez, L.~R. Pericchi, and I.~C. Ram{\'i}rez, ``{The Scaled {B}eta2
  Distribution as a Robust Prior for Scales},'' \emph{Bayesian Analysis},
  vol.~12, no.~3, pp. 615 -- 637, 2017.

\bibitem{nielsen2011burbea}
F.~Nielsen and S.~Boltz, ``The {B}urbea-{R}ao and {B}hattacharyya centroids,''
  \emph{IEEE Trans. Inf. Theory}, vol.~57, no.~8, pp. 5455--5466, 2011.

\bibitem{thouvenin2022distributed}
P.-A. Thouvenin, A.~Repetti, and P.~Chainais, ``A distributed {G}ibbs sampler
  with hypergraph structure for high-dimensional inverse problems,'' \emph{J.
  Comput. Graph. Stat.}, 2024.

\bibitem{ross2023introduction}
S.~M. Ross, \emph{Introduction to Probability Models}, 10th~ed.\hskip 1em plus
  0.5em minus 0.4em\relax Academic Press, 2023.

\bibitem{gelman2013bayesian}
A.~Gelman, J.~B. Carlin, H.~S. Stern, D.~B. Dunson, A.~Vehtari, and D.~B.
  Rubin, \emph{Bayesian data analysis}, 3rd~ed.\hskip 1em plus 0.5em minus
  0.4em\relax Chapman and Hall/CRC, 2013.

\bibitem{laumont2022bayesian}
R.~Laumont, V.~D. Bortoli, A.~Almansa, J.~Delon, A.~Durmus, and M.~Pereyra,
  ``{B}ayesian imaging using plug \& play priors: when {L}angevin meets
  {T}weedie,'' \emph{SIAM J. Imag. Sci.}, vol.~15, no.~2, pp. 701--737, 2022.

\bibitem{li2022mirror}
R.~Li, M.~Tao, S.~S. Vempala, and A.~Wibisono, ``The mirror {L}angevin
  algorithm converges with vanishing bias,'' in \emph{Proc. Int. Conf.
  Algorithmic Learning Theory (ALT)}.\hskip 1em plus 0.5em minus 0.4em\relax
  PMLR, 2022, pp. 718--742.

\bibitem{bouton2026distributed}
M.~Bouton, P.-A. Thouvenin, A.~Repetti, and P.~Chainais, ``A distributed
  plug-and-play mcmc algorithm for high-dimensional inverse problems,''
  \emph{IEEE Trans. Image Process.}, 2026.

\bibitem{robert1999monte}
C.~P. Robert, G.~Casella, and G.~Casella, \emph{{M}onte {C}arlo statistical
  methods}.\hskip 1em plus 0.5em minus 0.4em\relax Springer, 1999, vol.~2.

\bibitem{deng2009imagenet}
J.~Deng, W.~Dong, R.~Socher, L.-J. Li, K.~Li, and L.~Fei-Fei, ``{ImageNet}: A
  large-scale hierarchical image database,'' in \emph{Proc. Int. Conf. Computer
  Vision Pattern Recognition (CVPR)}, 2009, pp. 248--255.

\bibitem{zhang2021plug}
K.~Zhang, Y.~Li, W.~Zuo, L.~Zhang, L.~Van~Gool, and R.~Timofte, ``Plug-and-play
  image restoration with deep denoiser prior,'' \emph{IEEE Trans. Patt. Anal.
  Mach. Intell.}, vol.~44, no.~10, pp. 6360--6376, 2021.

\bibitem{wang2004image}
Z.~Wang, A.~C. Bovik, H.~R. Sheikh, and E.~P. Simoncelli, ``Image quality
  assessment: from error visibility to structural similarity,'' \emph{IEEE
  Trans. Image Process.}, vol.~13, no.~4, pp. 600--612, 2004.

\bibitem{zhang2018unreasonable}
R.~Zhang, P.~Isola, A.~A. Efros, E.~Shechtman, and O.~Wang, ``The unreasonable
  effectiveness of deep features as a perceptual metric,'' in \emph{Proc. Int.
  Conf. Computer Vision Pattern Recognition (CVPR)}, 2018, pp. 586--595.

\bibitem{van2016fast}
W.~Van~Aarle, W.~J. Palenstijn, J.~Cant, E.~Janssens, F.~Bleichrodt,
  A.~Dabravolski, J.~De~Beenhouwer, K.~Joost~Batenburg, and J.~Sijbers, ``Fast
  and flexible x-ray tomography using the {ASTRA} toolbox,'' \emph{Optics
  express}, vol.~24, no.~22, pp. 25\,129--25\,147, 2016.

\bibitem{van2015astra}
W.~Van~Aarle, W.~J. Palenstijn, J.~De~Beenhouwer, T.~Altantzis, S.~Bals, K.~J.
  Batenburg, and J.~Sijbers, ``The {ASTRA} toolbox: A platform for advanced
  algorithm development in electron tomography,'' \emph{Ultramicroscopy}, vol.
  157, pp. 35--47, 2015.

\bibitem{daniele2025deep}
C.~Daniele, S.~Villa, S.~Vaiter, and L.~Calatroni, ``Deep equilibrium models
  for {P}oisson imaging inverse problems via mirror descent,'' \emph{arXiv
  preprint arXiv:2507.11461}, 2025.

\bibitem{Tachella2025}
J.~Tachella, M.~Terris, S.~Hurault, A.~Wang, L.~Davy, J.~Scanvic, V.~Sechaud,
  R.~Vo, T.~Moreau, T.~Davies, D.~Chen, N.~Laurent, B.~Monroy, J.~Dong, Z.~Hu,
  M.-H. Nguyen, F.~Sarron, P.~Weiss, P.~Escande, M.~Massias, T.~Modrzyk,
  B.~Levac, T.~I. Liaudat, M.~Song, J.~Hertrich, S.~Neumayer, and G.~Schramm,
  ``Deepinverse: A python package for solving imaging inverse problems with
  deep learning,'' \emph{Journal of Open Source Software}, vol.~10, no. 115, p.
  8923, 2025.

\bibitem{kompa2021second}
B.~Kompa, J.~Snoek, and A.~L. Beam, ``Second opinion needed: communicating
  uncertainty in medical machine learning,'' \emph{NPJ Digital Medicine},
  vol.~4, no.~1, p.~4, 2021.

\bibitem{lopez2025uncert}
L.~J.~L. Lopez, S.~Elsharief, D.~A. Jorf, F.~Darwish, C.~Ma, and F.~E. Shamout,
  ``Uncertainty quantification for machine learning in healthcare: A survey,''
  in \emph{Proc. Conf. Health Inference and Learning}, vol. 287.\hskip 1em plus
  0.5em minus 0.4em\relax PMLR, 25--27 Jun 2025, pp. 862--907.

\bibitem{jaspan2022improving}
O.~Jaspan, A.~Wysocka, C.~Sanchez, and A.~D. Schweitzer, ``Improving the
  relationship between confidence and competence: implications for diagnostic
  radiology training from the psychology and medical literature,''
  \emph{Academic Radiology}, vol.~29, no.~3, pp. 428--438, 2022.

\bibitem{renaud2024plug}
M.~Renaud, J.~Liu, V.~De~Bortoli, A.~Almansa, and U.~Kamilov, ``Plug-and-play
  posterior sampling under mismatched measurement and prior models,'' in
  \emph{Proc. Int. Conf. Machine Learning (ICML)}, vol. 2024, 2024, pp.
  10\,139--10\,168.

\bibitem{terris2024equivariant}
M.~Terris, T.~Moreau, N.~Pustelnik, and J.~Tachella, ``Equivariant
  plug-and-play image reconstruction,'' in \emph{Proc. Int. Conf. Computer
  Vision Pattern Recognition (CVPR)}, 2024, pp. 25\,255--25\,264.

\bibitem{stanley2011enumerative}
R.~P. Stanley, ``Enumerative combinatorics,'' \emph{Cambridge Studies in
  Advanced Mathematics}, vol.~1, 2011.

\bibitem{pena2025properties}
V.~Pe{\~n}a and M.~Jauch, ``Properties of the generalized inverse {G}aussian
  with applications to {M}onte {C}arlo simulation and distribution function
  evaluation,'' \emph{Statistics \& Probability Letters}, vol. 220, p. 110359,
  2025.

\bibitem{reehorst2018regularization}
E.~T. Reehorst and P.~Schniter, ``Regularization by denoising: Clarifications
  and new interpretations,'' \emph{IEEE Trans. Comput. Imag.}, vol.~5, no.~1,
  pp. 52--67, 2018.

\bibitem{hurault2021gradient}
S.~Hurault, A.~Leclaire, and N.~Papadakis, ``Gradient step denoiser for
  convergent plug-and-play,'' in \emph{Proc. IEEE Int. Conf. Learn. Represent.
  (ICLR)}, 2022.

\end{thebibliography}

\end{document}